\def\llncs{0}											
\def\LNCSpreview{0}								
\def\papertype{0}								   
\def\fullpage{1}									 
\def\pagelimit{}									 
\def\anonymous{0}								 
\def\acknowledgments{0}						 
\def\overflow{0}									 
\def\showlabels{0}								   
\def\authnotes{0}									
\def\dieordollar{0}									
\def\notxfont{0}
\def\stuffedtitlepage{0}						 
\def\abbrevref{0}									
\def\dropargs{0}									
\def\rmvtheoremspace{0}						 
\def\allowbreaks{0}								  
\def\choosebibstyle{1}							
\def\changefont{0}								  
\def\submission{0}
\def\cameraready{0}
\def\noaux{1}

\ifnum\submission=1
\def\llncs{1}
\def\anonymous{1}
\def\authnotes{0}
\def\choosebibstyle{0}
\def\noaux{0}
\else
\fi

\ifnum\cameraready=1
\def\submission{1}
\def\llncs{1}
\def\authnotes{0}
\def\choosebibstyle{0}
\def\acknowledgments{1}	
\else
\fi

\ifnum\anonymous=1
\def\authnotes{0}
\else
\fi

\ifnum\authnotes=0  
\newcommand{\fuyuki}[1]{}
\newcommand{\ryo}[1]{}
\else
\newcommand{\fuyuki}[1]{$\ll$\textsf{\color{blue} Fuyuki: { #1}}$\gg$}
\newcommand{\ryo}[1]{$\ll$\textsf{\color{red} Ryo: { #1}}$\gg$}
\fi

\def\confvers{0}										   

\ifnum\llncs=1
\def\titletext{											
Functional Encryption with Secure Key Leasing
}
\def\runningtitle{									
Functional Encryption with Secure Key Leasing
}
\else
\def\titletext{											
Functional Encryption with Secure Key Leasing
}
\def\runningtitle{									
}
\fi

\date{}										  

\def\choosepubinfo{5}						   

\def\pubinfoYEAR{2016}								
\def\pubinfoSUBMISSIONDATE{}		  
\def\pubinfoDOI{ }								 
\def\pubinfoBIBDATA{ }						 
\def\pubinfoCONFERENCE{CRYPTO}				

\ifnum\submission=1
\ifnum\noaux=1
\def\pubinfoindividual{}
\else
\def\pubinfoindividual{{\color{red}{\textbf{We attached the full version of this paper as a supplementary material}}}.}
\fi
\else
\def\pubinfoindividual{}
\fi

\makeatletter
\expandafter\newcommand{\createauthor}[5]{%
	\@namedef{#1name}{#2}%
	\@namedef{#1running}{#3}%
	\@namedef{#1institute}{#4}%
	\@namedef{#1thanks}{#5}%
}

\expandafter\newcommand{\createinstitute}[4]{%
	\@namedef{#1instname}{#2}%
	\@namedef{#1mail}{#3}%
	\@namedef{#1number}{#4}%
}
\makeatother

\newcounter{authorcount}
\newcommand{\newauthor}[4]{
	\stepcounter{authorcount}
	\createauthor{\theauthorcount}{#1}{#2}{#3}{#4}
}

\newcounter{institutecount}
\newcommand{\newinstitute}[3]{
	\stepcounter{institutecount}
	\createinstitute{\theinstitutecount}{#1}{#2}{#3}
}

\newauthor{Fuyuki~Kitagawa}{F.~Kitagawa}{1}{}
\newauthor{Ryo~Nishimaki}{R.~Nishimaki}{1}{}

\newinstitute{NTT Social Informatics Laboratories, Tokyo, Japan}{\{fuyuki.kitagawa.yh,ryo.nishimaki.zk\}@hco.ntt.co.jp}{1}

\def\contactmail{}

\def\keywords{
Functional encryption, secure leasing, copy-protection, quantum cryptography
}

\def\acknowledgmenttext{
}

\def\authorsofpdf{
\ifcsname 5name\endcsname
	\csname 1name\endcsname~et~al.
\else
	\ifcsname 1name\endcsname
		\csname 1name\endcsname
		\ifcsname 3name\endcsname
			,\ \csname 2name\endcsname
			\ifcsname 4name\endcsname
				,\ \csname 3name\endcsname\ and \csname 4name\endcsname
			\else%
				\ and\ \csname 3name\endcsname
			\fi
		\else
			\ifcsname 2name\endcsname
				\csname 2name\endcsname
			\fi
		\fi
	\fi
\fi
}
\ifnum\LNCSpreview=1
	\def\llncs{1}
\fi

\let\accentvec\vec

\ifnum\llncs=1
	\documentclass[envcountsect,a4paper,runningheads,10pt]{llncs}
\else
	\ifnum\papertype=0
		\documentclass[a4paper,11pt]{article}
	\else
		\documentclass[letterpaper,11pt]{article}
	\fi

	\ifnum\fullpage=1
		\usepackage{fullpage}
	\fi
\fi

\let\vec\accentvec
\usepackage[dvipsnames]{xcolor}
\definecolor{darkblue}{rgb}{0,0,0.6}
\definecolor{darkgreen}{rgb}{0,0.5,0}
\definecolor{darkviolet}{RGB}{130,95,141}

\usepackage
[pdfpagelabels=true,
linktocpage=false,
bookmarks=true,bookmarksnumbered=false,bookmarkstype=toc,
pagebackref,
colorlinks=true,linkcolor=darkblue,urlcolor=darkblue,citecolor=darkviolet]
{hyperref}

\usepackage{amscd,amsmath,amssymb,amsfonts}

\usepackage{amsthm}

\ifnum\abbrevref=0
	\usepackage[capitalise,noabbrev]{cleveref}
\else
	\usepackage[capitalise]{cleveref}
\fi
\usepackage[absolute]{textpos}
\usepackage[final]{microtype}
\usepackage[absolute]{textpos}
\usepackage{everypage}

\ifnum\dieordollar=2
	\usepackage{tikz}
\else
	\ifnum\dieordollar=1
		\usepackage{tikz}
	\fi
\fi

\ifnum\LNCSpreview=1
	\usepackage[paperwidth=152mm,paperheight=235mm,textwidth=122mm,textheight=193mm]{geometry}
	\AtBeginDocument{
		\begin{textblock}{5.25}[-0.59,0](0,12)
			\centering\Large
			\textcolor{red}{\textbf{LNCS Preview mode active.}}
		\end{textblock}
	}
	\def\choosepubinfo{0}
\fi

\ifnum\overflow=1
\fi

\ifnum\llncs=0
	\renewcommand*{\backref}[1]{(Cited on page~#1.)}
\fi

\ifnum\allowbreaks=1
	\allowdisplaybreaks
\fi

\ifnum\showlabels=1
	\usepackage{showkeys}
\fi

\ifnum\anonymous=1
	\def\authnotes{0}
\fi

\ifnum\changefont=1
	\usepackage{times}
\fi

\ifnum\LNCSpreview=0
	\ifx\pagelimit\empty
	\else
		\newcounter{pagewarning}
		\AddEverypageHook{
			\ifnum\thepage>\thepagewarning
				\begin{textblock}{1}[0.5,0](0,-13)
					\centering\Large
					\textcolor{red}{\textbf{This page is exceeding the page limit.}}
				\end{textblock}
			\fi
	}
	\fi
\fi

\newcommand{\authnote}[2]{\ifnum\authnotes=1 \begin{center}\fbox{\begin{minipage}{.98\textwidth}
\textbf{#1 says:} #2\end{minipage}}\end{center} \fi}

\newlength{\strutdepth}%
\newcommand{\notes}[3]{
\ifnum\authnotes=1
	\noindent{\bfseries
	\color{#1}{#3}\color{#1}}%
	\strut\vadjust{\kern-\strutdepth%
		\vtop to \strutdepth{%
			\baselineskip\strutdepth%
			\vss\llap{{\large\color{#1}\textbf{#2}\quad\color{black}}}\null%
		}%
	}%
\fi
}

\ifnum\LNCSpreview=0

\else

\fi

\mathchardef\hyphen="2D

\newtheoremstyle{thicktheorem}%
{\ifnum\rmvtheoremspace=1
	0.4
\fi
\topsep}
{\ifnum\rmvtheoremspace=1
	0.4
\fi
\topsep}
{\itshape}{}%
{\bfseries}%
{.}
{ }%
{\thmname{#1}\thmnumber{ #2}%
	\ifnum\dropargs=0
		\thmnote{ (#3)}%
	\fi
}

\newtheoremstyle{remark}
{\ifnum\rmvtheoremspace=1
	0.4
\fi
\topsep}
{\ifnum\rmvtheoremspace=1
	0.4
\fi
\topsep}
	{}
	{}
	{}
	{.}
	{ }
	{\textit{\thmname{#1}}\thmnumber{ #2}
		\ifnum\dropargs=0
			\thmnote{ (#3)}%
		\fi
	}

\ifnum\llncs=0
	\theoremstyle{thicktheorem}
	\newtheorem{theorem}{Theorem}[section]
	\newtheorem{lemma}[theorem]{Lemma}
	\newtheorem{corollary}[theorem]{Corollary}
	\newtheorem{proposition}[theorem]{Proposition}
	\newtheorem{definition}[theorem]{Definition}

	\theoremstyle{remark}
	\newtheorem{claim}[theorem]{Claim}
	\newtheorem{remark}[theorem]{Remark}

	\newtheorem{conjecture}[theorem]{Conjecture}
\fi

\theoremstyle{remark}
\newtheorem{assumption}[theorem]{Assumption}
\newtheorem{observation}[theorem]{Observation}
\newtheorem{fact}[theorem]{Fact}
\newtheorem{experiment}{Experiment}
\newtheorem{construction}[theorem]{Construction}
\newtheorem{counterexample}[theorem]{Counterexample}

\ifnum\abbrevref=0
	\crefname{assumption}{Assumption}{Assumptions}
	\crefname{construction}{Construction}{Constructions}
	\crefname{corollary}{Corollary}{Corollaries}
	\crefname{conjecture}{Conjecture}{Conjectures}
	\crefname{definition}{Definition}{Definitions}
	\crefname{exmaple}{Example}{Examples}
	\crefname{experiment}{Experiment}{Experiments}
	\crefname{counterexample}{Counterexample}{Counterexamples}
	\crefname{lemma}{Lemma}{Lemmata}
	\crefname{observation}{Observation}{Observations}
	\crefname{proposition}{Proposition}{Propositions}
	\crefname{remark}{Remark}{Remarks}
	\crefname{theorem}{Theorem}{Theorems}
\else
	\crefname{assumption}{Ass.}{Ass.}
	\crefname{construction}{Constr.}{Constr.}
	\crefname{corollary}{Cor.}{Cor.}
	\crefname{conjecture}{Conj.}{Conj.}
	\crefname{definition}{Def.}{Def.}
	\crefname{exmaple}{Ex.}{Ex.}
	\crefname{experiment}{Exp.}{Exp.}
	\crefname{counterexample}{Counterex.}{Counterex.}
	\crefname{lemma}{Lem.}{Lem.}
	\crefname{observation}{Obs.}{Obs.}
	\crefname{proposition}{Prop.}{Prop.}
	\crefname{remark}{Rem.}{Rem.}
	\crefname{theorem}{Thm.}{Thms.}
\fi

\crefname{claim}{Claim}{Claims}
\crefname{fact}{Fact}{Facts}
\crefname{note}{Note}{Notes}

\def\YYYSMcoin{\mbox{\begin{tikzpicture}[scale=0.0125]
\definecolor{coinbrown}{HTML}{D89E36}\definecolor{coindarkyellow}{HTML}{F8D81E}\definecolor{coinyellow}{HTML}{F8F800}\fill[coinyellow] (3,-1) rectangle (9,9);\fill(0,0) rectangle (1,8);\fill(1,8) rectangle (2,10);\fill(2,10) rectangle (4,11);\fill(4,11) rectangle (8,12);\fill(8,11) rectangle (10,10);\fill(10,10) rectangle (11,8);\fill(11,8) rectangle (12,0);\fill(10,-2) rectangle (11,0);\fill(8,-3) rectangle (10,-2);\fill(4,-4) rectangle (8,-3);\fill(2,-3) rectangle (4,-2);\fill(1,0) rectangle (2,-2);\fill (5,-1) rectangle (7,0);\fill (7,0) rectangle (8,8);\fill[coinbrown] (9,8) rectangle (10,10);\fill[coinbrown] (10,0) rectangle (11,8);\fill[coinbrown] (9,-2) rectangle (10,0);\fill[coinbrown] (8,-2) rectangle (9,-1);\fill[coinbrown] (4,-3) rectangle (8,-2);\fill[coindarkyellow] (2,-2) rectangle (3,8);\fill[coindarkyellow] (3,-2) rectangle (8,-1);\fill[coindarkyellow] (8,-1) rectangle (9,0);\fill[coindarkyellow] (9,0) rectangle (10,8);\fill[coindarkyellow] (8,8) rectangle (9,10);\fill[coindarkyellow] (4,9) rectangle (8,10);\fill[coindarkyellow] (3,8) rectangle (4,9);\fill[coindarkyellow] (5,0) rectangle (7,2);\fill[coindarkyellow] (6,2) rectangle (7,8);\fill[white] (4,0) rectangle (5,8);\fill[white] (5,8) rectangle (7,9);
\end{tikzpicture}}}

\def\YYYdie{\mbox{\begin{tikzpicture}[scale=0.85,x=1em,y=1em,radius=0.09]
\draw[rounded corners=1,line width=.25pt] (0,0) rectangle (1,1);\fill (0.275,0.275) circle;\fill (0.725,0.725) circle;\fill (0.5,0.5) circle;
\end{tikzpicture}}}

\newcommand{\getsr}{
	\ifnum\dieordollar=0
		\mathrel{\vbox{\offinterlineskip\ialign{
			\hfil##\hfil\cr
			\hspace{0.1em}$\scriptscriptstyle\$$\cr
			$\leftarrow$\cr
		}}}
	\fi
	\ifnum\dieordollar=1
		\mathrel{\vbox{\offinterlineskip\ialign{
			\hfil##\hfil\cr
			{\scalebox{0.5}{\hspace{0.4em}\YYYdie}}\cr
			\noalign{\kern0.05ex}
			$\leftarrow$\cr
		}}}
	\fi
	\ifnum\dieordollar=2
		\mathrel{\vbox{\offinterlineskip\ialign{
			\hfil##\hfil\cr
			\hspace{0.1em}$\YYYSMcoin$\cr
			$\leftarrow$\cr
		}}}
	\fi
}

\def\makeuppercase#1{
\expandafter\newcommand\csname sf#1\endcsname{\mathsf{#1}}
\expandafter\newcommand\csname frak#1\endcsname{\mathfrak{#1}}
\expandafter\newcommand\csname bb#1\endcsname{\mathbb{#1}}
\expandafter\newcommand\csname bf#1\endcsname{\textbf{#1}}
}

\def\makelowercase#1{
\expandafter\newcommand\csname frak#1\endcsname{\mathfrak{#1}}
\expandafter\newcommand\csname bf#1\endcsname{\textbf{#1}}
}

\newcounter{char}
\loop
	\edef\letter{\alph{char}}
	\edef\Letter{\Alph{char}}
	\expandafter\makelowercase\letter
	\expandafter\makeuppercase\Letter
	\stepcounter{char}
	\unless\ifnum\thechar>26
\repeat

\usepackage{lmodern}
\usepackage[T1]{fontenc}
\usepackage[utf8]{inputenc}
\usepackage{etex}
\usepackage{graphicx,pstricks}
\usepackage{tikz}
\usetikzlibrary{matrix,arrows}
\usetikzlibrary{positioning}
\usetikzlibrary{decorations,decorations.text}
\usetikzlibrary{decorations.pathmorphing}
\usetikzlibrary{shapes}
\usepackage{centernot}
\usepackage{xspace}
\usepackage{pifont}
\usepackage{paralist}
\usepackage{booktabs}
\usepackage{dashbox}
\usepackage{multirow}
\usepackage{lscape}
\usepackage{amsmath}
\usepackage{fancybox}
\usepackage{braket}
\usepackage{physics}
\usepackage{autonum}
\usepackage{dsfont}
\DeclareMathAlphabet{\mathpzc}{OT1}{pzc}{m}{it}

\ifnum\submission=1
\renewcommand*{\backref}[1]{}
\def\notxfont{1}
\else
\fi

\ifnum\llncs=1
\renewcommand{\subparagraph}{\paragraph}
\else
\ifnum\notxfont=1
\else
\usepackage{mathpazo}
\usepackage{newtxtext}
\usepackage{helvet}
\fi
\fi
\ifnum\showlabels=1
\usepackage{showkeys}
\else
\fi

\usepackage{enumitem}
\ifnum\submission=1
\setlist[itemize]{
  topsep=0.2\baselineskip,
  itemsep=0\baselineskip,
}
\setlist[description]{
  topsep=0.2\baselineskip,
  itemsep=0\baselineskip,
}
\setlist[enumerate]{
  topsep=0.2\baselineskip,
  itemsep=0\baselineskip,
}

\else
\setlist[itemize]{
  topsep=0.4\baselineskip,
  itemsep=0.1\baselineskip,
}
\setlist[description]{
  topsep=0.4\baselineskip,
  itemsep=0.1\baselineskip,
}
\setlist[enumerate]{
  topsep=0.4\baselineskip,
  itemsep=0.1\baselineskip,
}
\fi

\usepackage{fontawesome}

\makeatletter
\newcounter{game}
\def\newGames#1#2#3{%
  \xdef\gameNS{#1}\xdef\gamePrefix{#2}\setcounter{game}{#3}\addtocounter{game}{-1}%
  \immediate\write\@auxout{\string\expandafter\string\xdef\noexpand\csname game-prefix-#1\string\endcsname{#2}}%
}
\def\newGame#1{%
  \xdef\prevGame{\gamePrefix\arabic{game}}\stepcounter{game}\xdef\thisGame{\gamePrefix\arabic{game}}%
  \immediate\write\@auxout{\string\expandafter\string\xdef\noexpand\csname game-\gameNS-#1\string\endcsname{\arabic{game}}}%
}
\makeatother
\def\safecsname#1{\expandafter\ifx\csname#1\endcsname\relax\else\csname#1\endcsname\fi}

\renewcommand{\Game}[1][]{\mathcmd{\textrm{Game\if!#1!\else~\ensuremath{#1}\fi}}}

\newcommand{\sel}{\mathsf{sel}}

\usepackage[font=small,
format=plain,
labelformat=simple, 
singlelinecheck=false,
labelfont=bf,
up
]{caption}
\DeclareCaptionFormat{myformat}{#1#2#3\hrulefill}
\newcommand{\run}{\qalgo{Run}}

\newcommand{\qA}{\qalgo{A}}
\newcommand{\qS}{\qalgo{S}}

\DeclareFontFamily{U}{skulls}{}
\DeclareFontShape{U}{skulls}{m}{n}{ <-> skull }{}

\newcommand{\qstateq}{\qstate{q}}

\usepackage{mathtools}

\newcommand{\chosen}{\leftarrow}

\renewcommand{\gets}{\leftarrow}
\newcommand{\lrun}{\leftarrow}
\newcommand{\out}{=}

\newcommand{\la}{\leftarrow}
\newcommand{\ra}{\rightarrow}

\newcommand{\seteq}{\coloneqq}

\newcommand{\tensor}{\otimes}
\newcommand{\concat}{\|}

\newcommand{\setbracket}[1]{\{#1\}}

\newcommand{\setbk}[1]{\{#1\}}

\newcommand{\cA}{\mathcal{A}}
\newcommand{\cB}{\mathcal{B}}
\newcommand{\cC}{\mathcal{C}}
\newcommand{\cD}{\mathcal{D}}
\newcommand{\cE}{\mathcal{E}}
\newcommand{\cF}{\mathcal{F}}
\newcommand{\cG}{\mathcal{G}}
\newcommand{\cH}{\mathcal{H}}
\newcommand{\cI}{\mathcal{I}}

\newcommand{\cM}{\mathcal{M}}

\newcommand{\cO}{\mathcal{O}}
\newcommand{\cP}{\mathcal{P}}

\newcommand{\cR}{\mathcal{R}}
\newcommand{\cS}{\mathcal{S}}

\newcommand{\cU}{\mathcal{U}}

\newcommand{\cX}{\mathcal{X}}
\newcommand{\cY}{\mathcal{Y}}

\newcommand{\N}{\mathbb{N}}
\newcommand{\Z}{\mathbb{Z}}

\newcommand{\Zq}{\mathbb{Z}_q}

\newcommand{\Fs}{\mathcal{F}}
\newcommand{\Ms}{\mathcal{M}}

\newcommand{\Ks}{\mathcal{K}}
\newcommand{\Ys}{\mathcal{Y}}
\newcommand{\Xs}{\mathcal{X}}

\newcommand{\coin}{\keys{coin}}

\newcommand{\M}{\cM}

\newcommand{\Params}{\algo{\Params}}

\newcommand{\Sampler}{\algo{Samp}}

\newcommand{\sep}{\lambda}

\newcommand{\secp}{\lambda}

\newcommand{\gtag}{\tau}

\newcommand{\cert}{\mathsf{cert}}

\newcommand{\aux}{\mathsf{aux}}

\newcommand{\params}{\mathsf{params}}

\newcommand{\rmOWF}{\textrm{one-way function}\xspace}

\newcommand{\rmIO}{\textrm{IO}\xspace}

\newcommand{\ind}{\mathsf{ind}}

\newcommand{\LWE}{\textrm{LWE}}

\newcommand{\sfreal}[2]{\mathsf{Real}^{\mathsf{#1}\textrm{-}\mathsf{#2}}}
\newcommand{\sfsim}[2]{\mathsf{Sim}^{\mathsf{#1}\textrm{-}\mathsf{#2}}}

\newcommand{\Adv}{\mathsf{Adv}}

\newcommand{\adva}[2]{\mathsf{Adv}_{#1}^{\mathsf{#2}}}
\newcommand{\advb}[3]{\mathsf{Adv}_{#1}^{\mathsf{#2} \mbox{-} \mathsf{#3}}}
\newcommand{\advc}[4]{\mathsf{Adv}_{#1}^{\mathsf{#2} \mbox{-} \mathsf{#3} \mbox{-} \mathsf{#4}}}

\newcommand{\expa}[2]{\mathsf{Expt}_{#1}^{\mathsf{#2}}}
\newcommand{\expb}[3]{\mathsf{Exp}_{#1}^{ \mathsf{#2} \mbox{-} \mathsf{#3}}}
\newcommand{\expc}[4]{\mathsf{Exp}_{#1}^{ \mathsf{#2} \mbox{-} \mathsf{#3} \mbox{-} \mathsf{#4}}}
\newcommand{\expd}[5]{\mathsf{Exp}_{#1}^{\mathsf{#2} \mbox{-} \mathsf{#3} \mbox{-} \mathsf{#4} \mbox{-} \mathsf{#5}}}

\newcommand{\hybi}[1]{\mathsf{Hyb}_{#1}}
\newcommand{\hybij}[2]{\mathsf{Hyb}_{#1}^{#2}}

\newcommand*{\sk}{\keys{sk}}
\newcommand*{\pk}{\keys{pk}}

\newcommand*{\msk}{\keys{msk}}

\newcommand*{\dk}{\keys{dk}}

\newcommand*{\vk}{\keys{vk}}

\newcommand*{\ct}{\keys{ct}}

\newcommand{\qsk}{\qalgo{sk}}

\newcommand{\cthat}{\widehat{\ct}}

\newcommand{\pkhat}{\widehat{\pk}}

\newcommand{\mskhat}{\widehat{\msk}}
\newcommand{\qskhat}{\widehat{\qsk}}

\newcommand{\prfkey}{\keys{K}}

\newcommand*{\keys}[1]{\mathsf{#1}}
\newcommand*{\qstate}[1]{\mathpzc{#1}}
\newcommand*{\qreg}[1]{{\color{gray}{\mathsf{#1}}}}

\newcommand*{\algo}[1]{\ensuremath{\mathsf{#1}}}
\newcommand*{\qalgo}[1]{\ensuremath{\mathpzc{#1}}}

\newenvironment{boxfig}[2]{\begin{figure}[#1]\fbox{\begin{minipage}{0.97\linewidth}
                        \vspace{0.2em}
                        \makebox[0.025\linewidth]{}
                        \begin{minipage}{0.95\linewidth}
            {{
                        #2 }}
                        \end{minipage}
                        \vspace{0.2em}
                        \end{minipage}}}{\end{figure}}

\newcommand{\pprotocol}[4]{
\begin{boxfig}{t!}{\footnotesize 
\centering{\textbf{#1}}
    #4
} \caption{#2}
\label{#3}
\end{boxfig}
}

\newcommand{\protocol}[4]{
\pprotocol{#1}{#2}{#3}{#4} }

 \newcounter{expitem}[table]
\newcommand\expitem[2]{\refstepcounter{expitem}%
   \arabic{expitem}.&#1&\arabic{expitem}.&#2\\%
}
\newcommand\Heading[2]{\multicolumn2{c}{#1}&\multicolumn2{c}{#2}\\}

\newcommand{\bit}{\{0,1\}}

\newcommand{\mat}[1]{\boldsymbol{#1}}

\newcommand{\mm}[1]{\boldsymbol{#1}}
\newcommand{\mv}[1]{\boldsymbol{#1}}

\newcommand{\Rand}{\algo{R}}

\newcommand{\prf}{\algo{F}}

\newcommand{\Setup}{\algo{Setup}}
\newcommand{\Gen}{\algo{Gen}}

\newcommand{\KG}{\algo{KG}}
\newcommand{\Enc}{\algo{Enc}}
\newcommand{\Dec}{\algo{Dec}}

\newcommand{\Vrfy}{\algo{Vrfy}}

\newcommand{\qEnc}{\qalgo{Enc}}

\newcommand\PKFE{\algo{PKFE}}

\newcommand\SKFE{\algo{SKFE}}

\newcommand\oneFE{\algo{1FE}}

\newcommand{\sfonefe}{\mathsf{1fe}}

\newcommand{\iO}{i\cO}

\newcommand{\PRF}{\algo{PRF}}

\newcommand{\prfgen}{\PRF.\Gen}

\newcommand{\Puncture}{\algo{Puncture}}

\newcommand{\TI}{\qalgo{TI}}
\newcommand{\ATI}{\qalgo{ATI}}

\newcommand{\projimp}{\algo{ProjImp}}

\newcommand{\Decode}{\algo{Decode}}

\newcommand{\negl}{{\mathsf{negl}}}

\newcommand{\poly}{{\mathrm{poly}}}
\newcommand{\polylog}{{\mathrm{polylog}}}

\newcommand{\zo}[1]{\{0,1\}^{#1}}
\newcommand{\bin}{\{0,1\}}

\newcommand{\class}[1]{\mathsf{#1}}
\newcommand{\classPpoly}{\class{P}/\class{poly}}

\newcommand{\Ppoly}{\classPpoly}

\newcommand{\NCone}{\class{NC}^1}

\newcommand{\List}[1]{L_\mathtt{#1}}

\newcommand{\calF}{\mathcal{F}}

\newcommand{\calM}{\mathcal{M}}

\newcommand{\calQ}{\mathcal{Q}}

\newcommand{\padsize}{\mathsf{pad}}
\newcommand{\pad}{\mathsf{pad}}

\newcommand{\SUC}{{\tt SUC}}

\newcommand{\SKE}{\algo{SKE}}
\newcommand{\ske}{\keys{ske}}

\newcommand{\E}{\algo{E}}
\newcommand{\D}{\algo{D}}

\newcommand{\PuncPRF}{\algo{PPRF}}

\newcommand{\OneCT}[1]{\algo{1CT}^{#1}}

\newcommand{\state}{\mathsf{st}}
\newcommand{\tlx}{\widetilde{x}}

\newcommand{\SDE}{\algo{SDE}}
\newcommand{\QKeyGen}{\qalgo{QKG}}
\newcommand{\pirateD}{\qalgo{D}}
\newcommand{\SDFE}{\algo{SDFE}}
\newcommand{\oneSDFE}{\algo{1SDFE}}
\newcommand{\sde}{\algo{sde}}
\newcommand{\qdk}{\qstate{dk}}
\newcommand{\fe}{\algo{fe}}

\newcommand{\Ssecure}{S_{\mathtt{secure}}}
\newcommand{\Sinsecure}{S_{\mathtt{insecure}}}

\newcommand{\SetHSS}{\algo{SetHSS}}
\newcommand{\SetGen}{\algo{SetGen}}
\newcommand{\numset}{m}
\newcommand{\numall}{\ell}
\newcommand{\element}{e}

\newcommand{\zSKFEsbSKL}{\mathsf{zSKFE}\textrm{-}\mathsf{sbSKL}}

\newcommand{\qzKG}{\qalgo{zKG}}
\newcommand{\zEnc}{\algo{zEnc}}
\newcommand{\qzDec}{\qalgo{zDec}}

\newcommand{\zmsk}{\mathsf{zmsk}}
\newcommand{\qzfsk}{\qstate{zfsk}}
\newcommand{\zct}{\mathsf{zct}}

\newcommand{\InpEncode}{\algo{InpEncode}}
\newcommand{\FuncEncode}{\algo{FuncEncode}}
\newcommand{\Decore}{\algo{Decode}}
\newcommand{\share}{s}

\newcommand{\SKFESKL}{\mathsf{SKFE\textrm{-}SKL}}
\newcommand{\iSKFESKL}{\mathsf{iSKFE\textrm{-}sbSKL}}
\newcommand{\qKG}{\qalgo{KG}}
\newcommand{\qDec}{\qalgo{Dec}}
\newcommand{\qcert}{\qalgo{Cert}}
\newcommand{\certvrfy}{\algo{Vrfy}}

\newcommand{\qfsk}{\qalgo{fsk}}

\renewcommand{\index}{j}

\newcommand{\instance}{i}
\newcommand{\bfinstance}{{\boldsymbol{I}}}

\newcommand{\iSetup}{\algo{iSetup}}
\newcommand{\qiKG}{\qalgo{iKG}}
\newcommand{\iEnc}{\algo{iEnc}}
\newcommand{\qiDec}{\qalgo{iDec}}
\newcommand{\qicert}{\qalgo{iCert}}
\newcommand{\iVrfy}{\algo{iVrfy}}

\newcommand{\Kske}{K}
\newcommand{\ctske}{\mathsf{sct}}
\newcommand{\sbSKL}{\mathsf{sbSKL}}
\newcommand{\sbskl}{\mathsf{sbskl}}

\newcommand{\CDSKE}{\mathsf{CDSKE}}
\renewcommand{\CD}{\mathsf{CD}}
\newcommand{\cd}{\mathsf{cd}}
\newcommand{\SKL}{\mathsf{SKL}}
\newcommand{\skl}{\mathsf{skl}}

\newcommand{\qSEnc}{\qalgo{SEnc}}
\newcommand{\qSKG}{\qalgo{SKG}}

\newcommand{\rskfe}{r}
\newcommand{\rcd}{w}

\newcommand{\fsk}{\mathsf{fsk}}

\newcommand{\SKFEsbSKL}{\algo{SKFE}\textrm{-}\algo{sbSKL}}
\newcommand{\SKFEdbSKL}{\algo{SKFE}\textrm{-}\algo{SKL}}
\newcommand{\dbSKL}{\algo{SKL}}
\newcommand{\dbskl}{\algo{skl}}

\newcommand{\qencrypt}{\qalgo{Enc}}
\newcommand{\qdelete}{\qalgo{Del}}
\newcommand{\qdecrypt}{\qalgo{Dec}}
\newcommand{\qct}{\qalgo{ct}}

\newcommand{\numkey}{q}
\newcommand{\numct}{n}
\newcommand{\numcti}[1]{n_{#1}}

\newcommand{\qB}{\qalgo{B}}

\newcommand{\qD}{\qalgo{D}}

\newcommand{\Oracle}[1]{O_{\mathtt{#1}}}

\begin{document}

\newcount\authorcounter
\newcommand{\provideauthors}{%
		\ifnum\authorcounter<\theauthorcount
			\csname\the\authorcounter name\endcsname
			\expandafter\ifx\csname\the\authorcounter thanks\endcsname\empty
			\else
				\thanks{\csname\the\authorcounter thanks\endcsname}
			\fi%
			\inst{\csname\the\authorcounter institute\endcsname} 
			\and 
			\global\advance\authorcounter by 1 
			\provideauthors
		\else
			\csname\the\authorcounter name\endcsname 
			\expandafter\ifx\csname\the\authorcounter thanks\endcsname\empty 
			\else
				\thanks{\csname\the\authorcounter thanks\endcsname} 
			\fi%
			\inst{\csname\the\authorcounter institute\endcsname} 
		\fi
}

\def\atleastoneauthorplaced{0}
\newcommand{\providerunning}{%
	\ifnum\authorcounter<\theauthorcount%
		\expandafter\ifx\csname\the\authorcounter running\endcsname\empty
		\else
			\ifnum\authorcounter>1
				\ifnum\atleastoneauthorplaced=1
					\and%
				\fi
			\fi
			\csname\the\authorcounter running\endcsname
			\def\atleastoneauthorplaced{1}
		\fi
		\global\advance\authorcounter by 1
		\providerunning%
	\else%
		\expandafter\ifx\csname\the\authorcounter running\endcsname\empty
		\else
			\ifnum\authorcounter>1
				\ifnum\atleastoneauthorplaced=1
					\and%
				\fi
			\fi
			\csname\the\authorcounter running\endcsname
		\fi
	\fi
}

\newcount\institutecounter

\newcommand{\provideinstitutes}{%
	\ifnum\institutecounter<\theinstitutecount%
		\ifnum\llncs=0
			$^{\csname\the\institutecounter number\endcsname}$
		\fi
		\csname\the\institutecounter instname\endcsname
		
		\email{
			\ifx\contactmail\empty
				\csname\the\institutecounter mail\endcsname
			\else
				\href{mailto:\contactmail}{\csname\the\institutecounter mail\endcsname}
			\fi
		}
		
		\and%
			\global\advance\institutecounter by 1
		\provideinstitutes%
	\else%
		\ifnum\llncs=0
			\ifcsname 1name\endcsname
				$^{\csname\the\institutecounter number\endcsname}$
			\fi
		\fi
		\csname\the\institutecounter instname\endcsname
		
		\email{
			\ifx\contactmail\empty
				\csname\the\institutecounter mail\endcsname
			\else
				\href{mailto:\contactmail}{\csname\the\institutecounter mail\endcsname}
			\fi
		}
	\fi
}

\title{
	\ifnum\stuffedtitlepage=1
		\ifnum\llncs=1
			\vspace*{-7ex}
		\else
		\vspace*{-3ex}
		\fi
		\textbf{\titletext}
		\ifnum\llncs=1
			\vspace*{-2ex}
		\else
			\vspace*{-1ex}
		\fi
	\else
		\textbf{\titletext}
	\fi
}
\ifnum\anonymous=1
	\author{}
\else
	\ifnum\llncs=0
		\newcommand{\inst}[1]{{
			\ifcsname 1name\endcsname
				$^{#1}$
			\fi
			}}
	\fi
	\ifcsname 1name\endcsname
		\author{
			\global\authorcounter 1
			\provideauthors
		}
	\fi
\fi

\ifnum\llncs=1
	\titlerunning{\runningtitle}
	\ifnum\anonymous=1
		\institute{}
		\authorrunning{}
	\else
		\ifcsname 1instname\endcsname{
			\institute{
				\global\institutecounter 1
				\provideinstitutes
			}
		\fi
		\ifcsname 1name\endcsname{
			\authorrunning{
				\global \authorcounter 1
				\providerunning
			}
		\fi
	\fi
\fi
\maketitle
\ifnum\stuffedtitlepage=1
	\ifnum\llncs=0
		\vspace{-4ex}
	\fi
\fi

\ifnum\llncs=0
	\ifnum\anonymous=0
		\newcommand{\email}[1]{
			\texttt{
				\ifx\contactmail\empty
					#1
				\else
					\href{mailto:\contactmail}{#1}
				\fi
			}
		}
		\newcommand{\and}{}
		\ifnum\stuffedtitlepage=1
			\ifnum\llncs=0
				\vspace{-2ex}
			\fi
		\fi
		\begin{small}
			\begin{center}
				\global \institutecounter 1
				\provideinstitutes
			\end{center}
		\end{small}
	\fi
\fi

\ifnum\stuffedtitlepage=1
	\ifnum\llncs=1
		\vspace*{-4ex}
	\else
		\vspace*{-2ex}
	\fi
\fi

\begin{abstract}
	

Secure software leasing is a quantum cryptographic primitive that enables us to lease software to a user by encoding it into a quantum state. Secure software leasing has a mechanism that verifies whether a returned software is valid or not. The security notion guarantees that once a user returns a software in a valid form, the user no longer uses the software.

In this work, we introduce the notion of secret-key functional encryption (SKFE) with secure key leasing, where a decryption key can be securely leased in the sense of secure software leasing. We also instantiate it with standard cryptographic assumptions. More specifically, our contribution is as follows.
\begin{itemize}
\item We define the syntax and security definitions for SKFE with secure key leasing.
\item We achieve a transformation from standard SKFE into SKFE with secure key leasing \emph{without using additional assumptions}. Especially, we obtain bounded collusion-resistant SKFE for $\Ppoly$ with secure key leasing based on post-quantum one-way functions since we can instantiate bounded collusion-resistant SKFE for $\Ppoly$ with the assumption.
\end{itemize}
Some previous secure software leasing schemes capture only pirate software that runs on an \emph{honest} evaluation algorithm (on a legitimate platform). However, our secure key leasing notion captures arbitrary attack strategies and does not have such a limitation.

As an additional contribution, we introduce the notion of single-decryptor FE (SDFE), where each functional decryption key is copy-protected. Since copy-protection is a stronger primitive than secure software leasing, this notion can be seen as a stronger cryptographic primitive than FE with secure key leasing. More specifically, our additional contribution is as follows.
\begin{itemize}
\item We define the syntax and security definitions for SDFE.
\item We achieve collusion-resistant single-decryptor PKFE for $\Ppoly$ from post-quantum indistinguishability obfuscation and quantum hardness of the learning with errors problem.
\end{itemize}

	\vspace{1ex}
\ifnum\llncs=1
\else

	\textbf{Keywords\ifnum\llncs=1{.}\else{:}\fi}\keywords
\fi
\end{abstract}
\ifnum\stuffedtitlepage=1
	\ifnum\llncs=1
		\vspace*{-2ex}
	\fi
\fi

\ifnum\llncs=0
	\vspace{1ex}
\fi

\ifnum\choosepubinfo=1
\def\pubinfo{
	\noindent An extended abstract of this paper will appear at
	\ifx\pubinfoCONFERENCE\empty \textcolor{red}{conference missing}\else \pubinfoCONFERENCE\fi.
}
\fi

\ifnum\choosepubinfo=2
	\def\pubinfo{
		\noindent \copyright\ IACR 
		\ifx\pubinfoYEAR\empty \textcolor{red}{year missing}\else \pubinfoYEAR\fi.
		This article is the final version submitted by the author(s) to the IACR and to Springer-Verlag on
		\ifx\pubinfoSUBMISSIONDATE\empty \textcolor{red}{submission date missing}\else \pubinfoSUBMISSIONDATE\fi.
		The version published by Springer-Verlag is available at
		\ifx\pubinfoDOI\empty \textcolor{red}{DOI missing}\else \pubinfoDOI\fi.
	}
\fi

\ifnum\choosepubinfo=3
	\def\pubinfo{
		\noindent \copyright\ IACR
		\ifx\pubinfoYEAR\empty \textcolor{red}{year missing}\else \pubinfoYEAR\fi.
		This article is a minor revision of the version published by Springer-Verlag available at
		\ifx\pubinfoDOI\empty \textcolor{red}{DOI missing}\else \pubinfoDOI\fi.
	}
\fi

\ifnum\choosepubinfo=4
	\def\pubinfo{
		\noindent This article is based on an earlier article:
		\ifx\pubinfoBIBDATA\empty \textcolor{red}{bibliographic data missing}\else \pubinfoBIBDATA\fi,
		\copyright\ IACR
		\ifx\pubinfoYEAR\empty \textcolor{red}{year missing}\else \pubinfoYEAR\fi,
		\ifx\pubinfoDOI\empty \textcolor{red}{DOI missing}\else \pubinfoDOI\fi.
	}
\fi

\ifnum\choosepubinfo=5
		\def\pubinfo{
			\noindent \pubinfoindividual
		}
	\fi

\textblockorigin{0.5\paperwidth}{0.9\paperheight}
\setlength{\TPHorizModule}{\textwidth}

\newlength{\pubinfolength}
\ifnum\choosepubinfo=0
\else
	\settowidth{\pubinfolength}{\pubinfo}
	\begin{textblock}{1}[0.5,0](0,.25)
		 \ifnum\pubinfolength<\textwidth
			\centering
		\fi
		\pubinfo
	\end{textblock}
\fi
\thispagestyle{empty}
		\ifnum\confvers=1
\else
\ifnum\submission=1
\else
\newpage
  \setcounter{tocdepth}{2}      
  \setcounter{secnumdepth}{2}   
  \setcounter{page}{0}          
  \tableofcontents
  \thispagestyle{empty}
\clearpage
\fi
\fi


\newcommand{\subEvs}{\class{subEVS}}

\section{Introduction}\label{sec:intro}

\subsection{Background}\label{sec:background}

Functional encryption (FE)~\cite{TCC:BonSahWat11} is an advanced encryption system that enables us to compute on encrypted data.
In FE, an authority generates a master secret key and an encryption key.
An encryptor uses the encryption key to generate a ciphertext $\ct_x$ of a plaintext $x$.
The authority generates a functional decryption key $\fsk$ from a function $f$ and the master secret key.
When a decryptor receives $\fsk$ and $\ct_x$, it can compute $f(x)$ and obtains nothing beyond $f(x)$.
In secret-key FE (SKFE), the encryption key is the same as the master secret key, while the encryption key is public in public-key FE (PKFE).

FE offers flexible accessibility to encrypted data since multiple users can obtain various processed data via functional decryption keys. Public-key encryption (PKE) and attribute-based encryption (ABE)~\cite{EC:SahWat05} do not have this property since they recover an entire plaintext if decryption succeeds. This flexible feature is suitable for analyzing sensitive data and computing new data from personal data without compromising data privacy. For example, we can compute medical statistics from patients' data without directly accessing individual data. Some works present practical applications of FE (for limited functionalities): non-interactive protocol for hidden-weight coin flips~\cite{USENIX:ConSch19}, biometric authentication, nearest-neighbor search on encrypted data~\cite{SCN:KLMMRW18}, private inference on encrypted data~\cite{NIPS:RPBDG19}.

One issue is that once a user obtains $\fsk$, it can compute $f(x)$ from a ciphertext of $x$ forever.
An authority may not want to provide users with the permanent right to compute on encrypted data.
A motivative example is as follows. A research group member receives a functional decryption key $\fsk$ to compute some statistics from many encrypted data for their research.
When the member leaves the group, an authority wants to prevent the member from doing the same computation on another encrypted data due to terms and conditions. However, the member might \emph{keep a copy} of their functional decryption key and penetrate the database of the group to do the same computation.
Another motivation is that the subscription business model is common for many services such as cloud storage services (ex. OneDrive, Dropbox), video on demand (ex. Netflix, Hulu), software applications (ex. Office 365, Adobe Photoshop).
If we can keep a copy of functional decryption keys, we cannot use FE in the subscription business model (for example, FE can be used as broadcast encryption in a video on demand).
We can also consider the following subscription service.
A company provides encrypted data sets for machine learning and a functional decryption key. A researcher can perform some tasks using the encrypted data set and the key.

Achieving a revocation mechanism~\cite{FC:NaoPin00} is an option to solve the issue above.
Some works propose revocation mechanisms for advanced encryption such as ABE~\cite{C:SahSeyWat12} and FE~\cite{EC:NisWicZha16}.
However, revocation is not a perfect solution since we need to update ciphertexts to embed information about revoked users.
We want to avoid updating ciphertexts for several reasons.
One is a practical reason. We possibly handle a vast amount of data, and updating ciphertexts incurs significant overhead.
Another one is more fundamental. Even if we update ciphertexts, \emph{there is no guarantee that all old ciphertexts are appropriately deleted.} If some user keeps copies of old ciphertexts, and a data breach happens after revocation, another functional decryption key holder whose key was revoked still can decrypt the old ciphertexts.

This problem is rooted in classical computation since we cannot prevent copying digital data.
Ananth and La Placa introduce the notion of secure software leasing~\cite{EC:AnaLaP21} to solve the copy problem by using the power of quantum computation.
Secure software leasing enables us to encode software into a leased version.
The leased version has the same functionality as the original one and must be a quantum state to prevent copying.
\emph{After a lessor verifies that the returned software from a lessee is valid (or that the lessee deleted the software)}, the lessee cannot execute the software anymore.
Several works present secure software leasing for simple functionalities such as a sub-class of evasive functions ($\subEvs$), PKE, signatures, pseudorandom functions (PRFs)~\cite{EC:AnaLaP21,C:ALLZZ21,TCC:KitNisYam21,TCC:BJLPS21,EPRINT:ColMajPor20}.
If we can securely implement leasing and returning mechanisms for functional decryption keys, we can solve the problem above.
Such mechanisms help us to use FE in real-world applications.

Thus, the main question in this work is as follows.
\begin{center}
\emph{Can we achieve secure a leasing mechanism for functional decryption keys of FE?}
\end{center}

We can also consider copy-protection, which is stronger security than secure leasing.
Aaronson~\cite{CCC:Aaronson09} introduces the notion of quantum copy-protection.
Copy-protection prevents users from creating a pirate copy. \emph{It does not have a returning process}, and prevents copying software.
If a user returns the original software, no copy is left behind on the user, and it cannot run the software.
Coladangelo, Liu, Liu, and Zhandry~\cite{C:CLLZ21} achieve copy-protected PRFs and single-decryptor encryption (SDE)\footnote{SDE is PKE whose decryption keys are copy-protected.}.
Our second question in this work is as follows.

\begin{center}
\emph{Can we achieve copy-protection for functional decryption keys of FE?}
\end{center}

We affirmatively answer those questions in this work.

\subsection{Our Result}\label{sec:result}

\paragraph{Secure key leasing.}
Our main contributions are introducing the notion of SKFE with secure key leasing and instantiating it with standard cryptographic assumptions. More specifically,
\begin{itemize}
\item We define the syntax and security definitions for SKFE with secure key leasing.
\item We achieve a transformation from standard SKFE into SKFE with secure key leasing \emph{without using additional assumptions}.
\end{itemize}
In SKFE with secure key leasing, a functional decryption key is a quantum state.
More specifically, the key generation algorithm takes as input a master secret key, a function $f$, and an availability bound $\numct$ (in terms of the number of ciphertexts), and outputs a quantum decryption key $\qfsk$ tied to $f$.
We can generate a \emph{certificate for deleting} the decryption key $\qfsk$.
If the user of this decryption key deletes $\qfsk$ within the declared availability bound $\numct$ and the generated certificate is valid, the user cannot compute $f(x)$ from a ciphertext of $x$ anymore.
We provide a high-level overview of the security definition in \cref{sec:tech_overview}.

We can obtain bounded collusion-resistant SKFE for $\Ppoly$ with secure key leasing from OWFs since we can instantiate bounded collusion-resistant SKFE for $\Ppoly$ with OWFs.\footnote{If we start with fully collusion-resistant SKFE, we can obtain fully collusion-resistant SKFE with secure key leasing.}
Note that all building blocks in this work are post-quantum secure since we use quantum computation and we omit ``post-quantum''.

Our secure key leasing notion is similar to but different from secure software leasing~\cite{EC:AnaLaP21} for FE because adversaries in secure software leasing (for FE) must run their pirate software by an \emph{honest} evaluation algorithm (on a legitimate platform). This is a severe limitation. In our FE with secure key leasing setting, adversaries do \emph{not} necessarily run their pirate software (for functional decryption) by an honest evaluation algorithm and can take arbitrary attack strategies.

We develop a transformation from standard SKFE into SKFE with secure key leasing by using quantum power.
In particular, we use (reusable) secret-key encryption (SKE) with certified deletion~\cite{TCC:BroIsl20,AC:HMNY21}, where we can securely delete \emph{ciphertexts}, as a building block.
We also develop a technique based on the security bound amplification for FE~\cite{C:AJLMS19,C:JKMS20} to amplify the availability bound, that is, the number of encryption queries before $\ct^\ast$ is given. This technique deviates from known multi-party-computation-based techniques for achieving bounded many-ciphertext security for SKFE~\cite{C:GorVaiWee12,TCC:AnaVai19}.\footnote{These techniques~\cite{C:GorVaiWee12,TCC:AnaVai19} work as transformations from single-key FE into bounded collusion-resistant FE. However, they also work as transformations from single-ciphertext SKFE into bounded many-ciphertext SKFE. See~\cref{sec:more_related_work} for the detail. Many-ciphertext means that SKFE is secure even if adversaries can send unbounded polynomially many queries to an encryption oracle.} The security bound amplification-based technique is of independent interest since the security bound amplification is not directly related to the amplification of the number of queries. These are the main technical contributions of this work. See~\cref{sec:tech_overview} and main sections for more details.

\paragraph{Copy-protected functional decryption keys.}
The other contributions are copy-protected functional decryption keys.
We introduce the notion of single-decryptor FE (SDFE), where each functional decryption key is copy-protected. This notion can be seen as a stronger cryptographic primitive than FE with secure key leasing, as we argued in~\cref{sec:background}.
\begin{itemize}
\item We define the syntax and security definitions for SDFE.
\item We achieve collusion-resistant public key SDFE for $\Ppoly$ from sub-exponentially secure indistinguishability obfuscation (IO) and the sub-exponential hardness of the learning with errors problem (QLWE assumption).
\end{itemize}
First, we transform single-key PKFE for $\Ppoly$ into \emph{single-key} SDFE for $\Ppoly$ by using SDE. Then, we transform single-key SDFE $\Ppoly$ into collusion-resistant SDFE for $\Ppoly$ by using an IO-based key bundling technique~\cite{JC:KitNisTan21,JC:BNPW20}. We can instantiate SDE with IO and the QLWE assumption~\cite{C:CLLZ21,ARXIV:CulVid21} and single-key PKFE for $\Ppoly$ with PKE~\cite{CCS:SahSey10,C:GorVaiWee12}.


\subsection{Technical Overview}\label{sec:tech_overview}
We provide a high-level overview of our techniques.
Below, standard math font stands for classical algorithms and classical variables, and calligraphic font stands for quantum algorithms and quantum states.

\paragraph{Syntax of SKFE with secure key leasing.}
We first recall a standard SKFE scheme. It consists of four algorithms $(\Setup,\KG,\Enc,\Dec)$.
$\Setup$ is given a security parameter $1^\secp$ and a collusion bound $1^\numkey$ and generates a master secret key $\msk$.
$\Enc$ is given $\msk$ and a plaintext $x$ and outputs a ciphertext $\ct$.
$\KG$ is given $\msk$ and a function $f$ and outputs a decryption key $\fsk$ tied to $f$.
$\Dec$ is given $\fsk$ and $\ct$ and outputs $f(x)$.
Then, the indistinguishability-security of SKFE roughly states that any QPT adversary cannot distinguish encryptions of $x_0$ and $x_1$ under the existence of the encryption oracle and the key generation oracle.
Here, the adversary can access the key generation oracle at most $\numkey$ times and can query only a function $f$ such that $f(x_0)=f(x_1)$.

An SKFE scheme with secure key leasing (SKFE-SKL) is a tuple of six algorithms $(\Setup, \qKG, \Enc, \qDec,\allowbreak\qcert,\Vrfy)$, where the first four algorithms form a standard SKFE scheme except the following difference on $\qKG$.
In addition to a function $f$, $\qKG$ is given an availability bound $1^\numct$ in terms of the number of ciphertexts.
Also, given those inputs, $\qKG$ outputs a verification key $\vk$ together with a decryption key $\qfsk$ tied to $f$ encoded in a quantum state, as $(\qfsk,\vk)\la\qKG(\msk,f,1^\numct)$.
By using $\qcert$, we can generate a (classical) certificate that a quantum decryption key $\qfsk$ is deleted, as $\cert\la\qcert(\qfsk)$.
We check the validity of certificates by using $\vk$ and $\Vrfy$, as $\top/\bot\la\Vrfy(\vk,\cert)$.
In addition to the decryption correctness, an SKFE-SKL scheme is required to satisfy the verification correctness that states that a correctly generated certificate is accepted, that is, $\top=\Vrfy(\vk,\cert)$ for $(\qfsk,\vk)\la\qKG(\msk,f,1^\numct)$ and $\cert\la\qcert(\qfsk)$.

\paragraph{Security of SKFE-SKL.}
The security notion of SKFE-SKL we call lessor security intuitively guarantees that if an adversary given $\qfsk$ deletes it and the generated certificate is accepted within the declared availability bound, the adversary cannot use $\qfsk$ any more.
The following indistinguishability experiment formalizes this security notion.
For simplicity, we focus on a selective setting where the challenge plaintext pair $(x_0^*,x_1^*)$ and the collusion bound $\numkey$ are fixed outside of the security experiment in this overview.

\begin{enumerate}
\item Throughout the experiment, $\qA$ can get access to the following oracles, where $\List{\qKG}$ is a list that is initially empty.
\begin{description}
\item[$\Oracle{\Enc}(x)$:]This is the standard encryption oracle that returns $\Enc(\msk,x)$ given $x$.
\item[$\Oracle{\qKG}(f,1^\numct)$:]This oracle takes as input a function $f$ and an availability bound $1^\numct$, generate $(\qfsk,\vk)\la\qKG(\msk,f,1^\numct)$, returns $\qfsk$ to $\qA$, and adds $(f,1^\numct,\vk,\bot)$ to $\List{\qKG}$.
Differently from the standard SKFE, $\qA$ can query a function $f$ such that $f(x_0^*)\ne f(x_1^*)$.
$\qA$ can get access to the key generation oracle at most $\numkey$ times.
\item[$\Oracle{\Vrfy}(f,\cert)$:]Also, $\qA$ can get access to the verification oracle. Intuitively, this oracle checks that $\qA$ deletes leased decryption keys correctly within the declared availability bounds.
Given $(f,\cert)$, it finds an entry $(f,1^\numct,\vk,M)$ from $\List{\qKG}$. (If there is no such entry, it returns $\bot$.) If $\top=\Vrfy(\vk,\cert)$ and the number of queries to $\Oracle{\Enc}$ at this point is less than $\numct$, it returns $\top$ and updates the entry into $(f,1^\numct,\vk,\top)$. Otherwise, it returns $\bot$.
\end{description}
\item When $\qA$ requests the challenge ciphertext, the challenger checks if $\qA$ has correctly deleted all leased decryption keys for functions $f$ such that $f(x_0^*)\ne f(x_1^*)$.
If so, the challenger gives the challenge ciphertext $\ct^*\la\Enc(\msk,x_\coin^*)$ for random bit $\coin\la\bit$ to $\qA$, and otherwise the challenger output $0$.
Hereafter, $\qA$ is not allowed to send a function $f$ such that $f(x_0^*)\ne f(x_1^*)$ to $\Oracle{\qKG}$.
\item $\qA$ outputs a guess $\coin'$ of $\coin$.
\end{enumerate}
We say that the SKFE-SKL scheme is lessor secure if no QPT adversary can guess $\coin$ significantly better than random guessing.
We see that if $\qA$ can use a decryption key after once $\qA$ deletes and the deletion certificate is accepted, $\qA$ can detect $\coin$ with high probability since $\qA$ can obtain a decryption key for $f$ such that $f(x_0^*)\ne f(x_1^*)$.
Thus, this security notion captures the above intuition.
We see that \emph{lessor security implies standard indistinguishability-security for SKFE}.

We basically work with the above indistinguishability based selective security for simplicity.
In \cref{sec:stronger_FESKL_security}, we also provide the definitions of adaptive security and simulation based security notions and general transformations to achieve those security notions from indistinguishability based selective security.

\paragraph{Dynamic availability bound vs. static availability bound.}
In SKFE-SKL, we can set the availability bound for each decryption key differently.
We can also consider a weaker variant where we statically set the single availability bound applied to each decryption key at the setup algorithm.
We call this variant SKFE with static bound secure key leasing (SKFE-sbSKL).
In fact, by using a technique developed in the context of dynamic bounded collusion FE~\cite{C:AMVY21,EPRINT:GGLW21}, we can generically transform SKFE-sbSKL into SKFE-SKL if the underlying SKFE-sbSKL satisfies some additional security property and efficiency requirement.
For the overview of this transformation, see \cref{sec:def_SKFESKL_variants}.
Therefore, we below focus on how to achieve SKFE-sbSKL. For simplicity, we ignore those additional properties required for the transformation to SKFE-SKL.

\paragraph{SKFE-sbSKL with the availability bound $0$ from certified deletion.}
We start with a simple construction of an SKFE-sbSKL scheme secure for the availability bound $0$ based on an SKE scheme with certified deletion~\cite{TCC:BroIsl20,AC:HMNY21}.
The availability bound is $0$ means that it is secure if an adversary deletes decryption keys without seeing any ciphertext.

SKE with certified deletion consists of five algorithms $(\KG,\qencrypt,\qdecrypt,\qdelete,\Vrfy)$.
The first three algorithms form a standard SKE scheme except that $\qencrypt$ output a verification key $\vk$ together with a ciphertext encoded in a quantum state $\qct$.
By using $\qdelete$, we can generate a (classical) certificate that $\qct$ is deleted.
The certificate is verified using $\vk$ and $\Vrfy$.
In addition to the decryption correctness, it satisfies the verification correctness that guarantees that a correctly generated certificate is accepted.
The security notion roughly states that once an adversary deletes a ciphertext $\qct$ and the generated certificate is accepted, the adversary cannot obtain any plaintext information encrypted inside $\qct$, even if the adversary is given the secret key after the deletion.

We now construct an SKFE-sbSKL scheme $\zSKFEsbSKL$ that is secure for the availability bound $0$, based on a standard SKFE scheme $\SKFE=(\Setup,\KG,\allowbreak\Enc,\Dec)$ and an SKE scheme with certified deletion $\CDSKE=(\CD.\KG,\CD.\qencrypt,\allowbreak\CD.\qdecrypt,\CD.\qdelete,\CD.\Vrfy)$.
In the setup of $\zSKFEsbSKL$, we generate $\msk\la\Setup(1^\secp,1^\numkey)$ and $\cd.\sk\la\CD.\KG(1^\secp)$, and the master secret key of $\zSKFEsbSKL$ is set to $\zmsk=(\msk,\cd.\sk)$.
To generate a decryption key for $f$, we generate a decryption key for $f$ by $\SKFE$ as $\fsk\gets\KG(\msk,f)$ and encrypt it by $\CDSKE$ as $(\cd.\qct,\vk)\gets\CD.\qencrypt(\cd.\sk,\fsk)$. The resulting decryption key is $\qzfsk\seteq\cd.\qct$ and the corresponding verification key is $\vk$.
To encrypt a plaintext $x$, we just encrypt it by $\SKFE$ as $\ct\la\Enc(\msk,x)$ and append $\cd.\sk$ contained in $\zmsk$, as $\zct\seteq(\ct,\cd.\sk)$.
To decrypt $\zct$ with $\qzfsk$, we first retrieve $\fsk$ from $\cd.\ct$ and $\cd.\sk$, and compute $f(x)\la\Dec(\fsk,\ct)$.
The certificate generation and verification are simply defined as those of $\CDSKE$ since $\qzfsk$ is a ciphertext of $\CDSKE$.

The security of $\zSKFEsbSKL$ is easily analyzed.
Let $(x_0^*,x_1^*)$ be the challenge plaintext pair.
When an adversary $\qA$ queries $f$ to $\Oracle{\qKG}$, $\qA$ is given $\qzfsk\seteq\cd.\qct$, where $\fsk\gets\KG(\msk,f)$ and $(\cd.\qct,\vk)\gets\CD.\qencrypt(\cd.\sk,\fsk)$.
If $f(x_0^*)\ne f(x_1^*)$, $\qA$ is required to delete $\qzfsk$ without seeing any ciphertext.
This means that $\qA$ cannot obtain $\cd.\sk$ before $\qzfsk$ is deleted.
Then, from the security of $\CDSKE$, $\qA$ cannot obtain any information of $\fsk$.
This implies that $\qA$ can obtain a decryption key of $\SKFE$ only for a function $f$ such that $f(x_0^*)=f(x_1^*)$, and thus the lessor security of $\zSKFEsbSKL$ follows form the security of $\SKFE$.

\paragraph{How to amplify the availability bound?}
We now explain how to amplify the availability bound from $0$ to any polynomial $\numct$.
One possible solution is to rely on the techniques for bounded collusion FE~\cite{C:GorVaiWee12,TCC:AnaVai19}.
Although the bounded collusion techniques can be used to amplify ``$1$-bounded security'' to ``poly-bounded security'', it is not clear how to use it starting from ``$0$-bounded security''.
For more detailed discussion on this point, see \cref{remark:GVW_not_work_for_dynamic}.
Therefore, we use a different technique from the existing bounded collusion FE.
At a high level, we reduce the task of amplifying the availability bound to the task of amplifying the security bound, which has been studied in the context of standard FE~\cite{C:AJLMS19,C:JKMS20}.

We observe that we can obtain an SKFE-sbSKL scheme with availability bound $\numct$ for any $\numct$ that is secure with only inverse polynomial probability by just using many instances of $\zSKFEsbSKL$ in parallel.
Concretely, suppose we use $N=\alpha\numct$ instances of $\zSKFEsbSKL$ to achieve a scheme with availability bound $\numct$, where $\alpha\in\bbN$.
To generate a decryption key for $f$, we generate $(\qzfsk_\index,\vk_\index)\la\qzKG(\zmsk_\index,f)$ for every $\index\in[N]$, and set the resulting decryption key as $(\qzfsk_\index)_{\index\in[N]}$ and the corresponding verification key as $(\vk_\index)_{\index\in[N]}$.
To encrypt $x$, we randomly choose $\index\la[N]$, generate $\zct_\index\la\zEnc(\zmsk_\index,x)$, and set the resulting ciphertext as $(\index,\zct_\index)$.
To decrypt this ciphertext with $(\qzfsk_\index)_{\index\in[N]}$, we just compute $f(x)\la\qzDec(\qzfsk_\index,\zct_\index)$.
The certification generation and verification are done by performing them under all $N$ instances.
The security of this construction is analyzed as follows.
The probability that the $\index^*$ chosen when generating the challenge ciphertext collides with some of $\numct$ indices $\index_1,\cdots,\index_\numct$ used by the first $\numct$ calls of the encryption oracle, is at most $n/N=1/\alpha$.
If such a collision does not happen, we can use the security of $\index^*$-th instance of $\zSKFEsbSKL$ to prove the security of this construction.
Therefore, this construction is secure with probability roughly $1-1/\alpha$ (denoted by $1/\alpha$-secure scheme).

Thus, all we have to do is to convert an SKFE-sbSKL scheme with inverse polynomial security into one with negligible security.
As stated above, such security amplification has been studied for standard FE.
In this work, we adopt the amplification technique using homomorphic secret sharing (HSS)~\cite{C:AJLMS19,C:JKMS20}.

\paragraph{Amplification using HSS.}
In this overview, we describe our construction using HSS that requires the LWE assumption with super-polynomial modulus to give a high-level intuition. However, our actual construction uses a primitive called set homomorphic secret sharing (SetHSS)~\cite{C:JKMS20}, which is a weak form of HSS and \emph{can be based on OWFs}.
\footnote{The definition of HSS provided below is not standard. We modify the definition to be close to SetHSS. Note that HSS defined below can be constructed from multi-key fully homomorphic encryption with simulatable partial decryption property~\cite{EC:MukWic16}.}
See~\cref{sec:removing_index} for our construction based on OWFs.

An HSS scheme consists of three algorithms $(\InpEncode,\allowbreak\FuncEncode,\allowbreak\Decode)$.
$\InpEncode$ is given a security parameter $1^\secp$, a number $1^\numset$, and an input $x$, and outputs $\numset$ input shares $(\share_\instance)_{\instance\in[\numset]}$.
$\FuncEncode$ is given a security parameter $1^\secp$, a number $1^\numset$, and a function $f$, and outputs $\numset$ function shares $(f_{\instance})_{\instance\in[\numset]}$.
$\Decode$ takes a set of evaluations of function shares on their respective input shares $(f_{\instance}(\share_{\instance}))_{\instance\in[\numset]}$, and outputs a value $f(x)$.
Then, the security property of an HSS scheme roughly guarantees that for any $(\instance^*,x_0^*,x_1^*)$, given a set of input shares $(\share_{\instance})_{\instance\in[\numset]\setminus\{\instance^*\}}$ for some $\instance^*$, an adversary cannot detect from which of the challenge inputs they are generated, under the existence of function encode oracle that is given $f$ such that $f(x_0^*)=f(x_1^*)$ and returns $(f_{\instance}(\share_{\instance}))_{\instance\in[\numset]}$.

\newcommand{\pSKFEsbSKL}{\mathsf{pSKFE}\textrm{-}\mathsf{sbSKL}}
\newcommand{\HSS}{\mathsf{HSS}}

We describe SKFE-sbSKL scheme $\SKFEsbSKL$ with the availability bound $\numct\ge1$ of our choice using a HSS scheme $\HSS=(\InpEncode,\allowbreak\FuncEncode,\Decode)$.
In the setup of $\SKFEsbSKL$, we first set up $1/2$-secure SKFE-sbSKL scheme $\SKFEsbSKL^\prime$ with the availability bound $\numct$.
This is done by parallelizing $2\numct$ instances of $\zSKFEsbSKL$ as explained before.
We generate $\numset$ master secret keys $\msk_1,\cdots,\msk_\numset$ of $\SKFEsbSKL^\prime$.
Then, to generate a decryption key for $f$ by $\SKFEsbSKL$, we first generate $(f_{\instance})_{\instance\in[\numset]}\la\FuncEncode(1^\secp,1^\numset,f)$, and generate a decryption key $\qfsk_{\instance}$ tied to $f_{\instance}$ under $\msk_\instance$ for each $\instance\in[\numset]$.
To encrypt $x$ by $\SKFEsbSKL$, we first generate $(\share_{\instance})_{\instance\in[\numset]}\la\InpEncode(1^\secp,1^\numset,x)$ and generate a ciphertext $\ct_{\instance}$ of $\share_{\instance}$ under $\msk_\instance$ for each $\instance\in[\numset]$.
The certification generation and verification are done by performing those of $\SKFESKL^\prime$ for all of the $\numset$ instances.
When decrypting the ciphertext $(\ct_\instance)_{\instance\in[\numset]}$ by $(\qfsk_\instance)_{\instance\in[\numset]}$, we can obtain $f_{\instance}(\share_{\instance})$ by decrypting $\ct_{\instance}$ with $\qfsk_{\instance}$ for every $\instance\in[\numset]$.
By combining $(f_{\instance}(\share_{\instance}))_{\instance\in[\numset]}$ using $\Decode$, we can obtain $f(x)$.

The lessor security of $\SKFEsbSKL$ can be proved as follows.
Each of $\numset$ instances of $\SKFEsbSKL^\prime$ is secure independently with probability $1/2$.
Thus, there is at least one secure instance without probability $1/2^\numset$, which is negligible by setting $\numset=\omega(\log \secp)$.
Suppose $\instance^*$-th instance is a secure instance.
Let $(x_0^*,x_1^*)$ be the challenge plaintext pair, and let $(\share^*_{\instance})_{\instance\in[\numset]}\la\InpEncode(1^\secp,1^\numset,x^*_\coin)$ for $\coin\la\bit$.
In the security experiment, from the security of $\SKFEsbSKL^\prime$ under $\msk_{\instance^*}$, an adversary cannot obtain the information of $\share_{\instance^*}$ except for its evaluation on function shares for a function $f$ queried to $\Oracle{\qKG}$ that satisfies that $f(x_0^*)=f(x_1^*)$.
Especially, from the security of $\SKFEsbSKL^\prime$ under $\msk_{\instance^*}$, the adversary cannot obtain an evaluation of $\share_{\instance^*}$ on function shares for a function $f$ such that $f(x_0^*)\ne f(x_1^*)$, though $\qA$ can query such a function to $\Oracle{\qKG}$.
Then, we see that the lessor security of $\SKFEsbSKL$ can be reduced to the security of $\HSS$.
\footnote{Actual construction and security proof needs to use a technique called the trojan method~\cite{C:ABSV15}. We ignore the issue here for simplicity.}

In the actual construction, we use SetHSS instead of HSS, as stated before.
Also, in the main body, we abstract the parallelized $\zSKFEsbSKL$ as index-based SKFE-sbSKL.
This makes the security proof of our construction using SetHSS simple.
Moreover, in the actual construction of an index-based SKFE-sbSKL, we bundle the parallelized instances of $\zSKFEsbSKL$ using a PRF.
This modification is necessary to achieve the efficiency required for the above transformation into SKFE-SKL.

Goyal et al.~\cite{EC:CiaGoyOst21} use a similar technique using HSS in a different setting (private simultaneous message protocols).
However, their technique relies on the LWE assumption unlike ours.

\paragraph{Single decryptor PKFE.}
In this work, we also define the notion of single decryptor PKFE, which is PKFE whose functional decryption key is copy-protected.
The definition is a natural extension of SDE (PKE with copy-protected decryption keys).
An adversary $\qA$ tries to copy a target functional decryption key $\qsk_{f^\ast}$.
More specifically, $\qA$ is given $\qsk_{f^\ast}$ and outputs two possibly entangled quantum distinguishers $\pirateD_1$ and $\pirateD_2$ and two plaintexts $(x_0,x_1)$ such that $f^\ast(x_0)\ne f^\ast(x_1)$. If $\pirateD_1$ or $\pirateD_2$ cannot distinguish a given ciphertext is encryption of $x_0$ or $x_1$, $\qsk_{f^\ast}$ is copy-protected.\footnote{In the security definition of SDE, quantum \emph{decryptors} try to recover the entire plaintext~\cite{C:CLLZ21}. We extend the definition because quantum decryptors are not sufficient for using SDFE as a building block. See~\cref{sec:single_dec_FE} for detail.} If both $\pirateD_1$ and $\pirateD_2$ have $\qsk_{f^\ast}$, they can trivially distinguish the challenge ciphertext. Thus, the definition guarantees copy-protection security.
We provide a collusion-resistant single-decryptor PKFE scheme, where adversaries obtain polynomially many functional decryption keys, based on IO.

We first show that a single-key single-decryptor PKFE can be constructed from a single-key standard PKFE scheme and SDE scheme.
The construction is simple nested encryption.
Namely, when encrypting a plaintext $x$, we first encrypt it by the standard PKFE scheme and then encrypt the ciphertext by the SDE scheme.
The secret key of the SDE scheme is included in the functional decryption key of the resulting single-decryptor PKFE scheme.
Although a PKFE functional decryption key can be copied, the SDE decryption key cannot be copied and adversaries cannot break the security of PKFE. This is because we need to run the SDE decryption algorithm before we run the PKFE decryption algorithm.

The security notion for SDE by Coladangelo et al.~\cite{C:CLLZ21} is not sufficient for our purpose since SDE plaintexts are ciphertexts of the standard PKFE in the construction.
We need to extend the security notion for SDE to prove the security of this construction because we need to handle randomized messages (the PKFE encryption is a randomized algorithm).
Roughly speaking, this new security notion guarantees that the security of SDE holds for plaintexts of the form $g(x;r)$, where $g$ and $x$ respectively are a function and an input chosen by an adversary and $r$ is a random coin chosen by the experiment. 
We can observe that the SDE scheme proposed by Coladangelo et al.~\cite{C:CLLZ21} based on IO satisfies this security notion.
Then, by setting $g$ as the encryption circuit of the standard PKFE, the security of the single-key single-decryptor PKFE scheme above can be immediately reduced to the security of the SDE scheme.
We also extend adversarial quantum decryptors, which try to output an entire plaintext, to adversarial quantum distinguishers, which try to guess a $1$-bit coin used to generate a ciphertext. We need this extension to use SDE as a building block. It is easy to observe the SDE scheme by Coladangelo et al.~\cite{C:CLLZ21} is secure even against quantum distinguishers.

Once we obtain a single-key single-decryptor PKFE scheme, we can transform it into a collusion-resistant single-decryptor PKFE scheme by again using IO. This transformation is based on one from a single-key standard PKFE scheme into a collusion-resistant standard PKFE scheme~\cite{JC:BNPW20,JC:KitNisTan21}. The idea is as follows. We need to generate a fresh instance of the single-key scheme above for \emph{each random tag} and bundle (unbounded) polynomially many instances to achieve collusion-resistance. We use IO to bundle multiple instances of single-key SDFE. More specifically, a public key is an obfuscated circuit of the following setup circuit. The setup circuit takes a public tag $\tau$ as input, generates a key pair $(\pk_\tau,\msk_\tau)$ of the single-key SDFE scheme using PRF value $\prf_{\prfkey}(\tau)$ as randomness, and outputs only $\pk_\tau$. The master secret key is the PRF key $\prfkey$. We can generate a functional decryption key for $f$ by choosing a random tag $\tau$ and generating a functional decryption key $\qsk_{f,\tau}$ under $\msk_{\tau}$. A functional decryption key of our collusion-resistant scheme consists of $(\tau,\qsk_{f,\tau})$. A ciphertext is an obfuscated circuit of the following encryption circuit, where a plaintext $x$ is hardwired. The encryption circuit takes a public tag $\tau$, generates $\pk_\tau$ by using the public key explained above, and outputs a ciphertext of $x$ under $\pk_\tau$. Due to this mechanism, only one functional decryption key $\qsk_{f,\tau}$ under $\msk_\tau$ is issued for each $\tau$, but we can generate polynomially many functional decryption keys by using many tags. If we use a different tag $\tau^\prime$, an independent key pair $(\pk_{\tau^\prime},\msk_{\tau^\prime})$ is generated and it is useless for another instance under $(\pk_\tau,\msk_\tau)$. The IO security guarantees that the information about $\prfkey$ (and $\msk_\tau$) is hidden.\footnote{We use puncturable PRFs and the puncturing technique here as the standard technique for cryptographic primitives based on IO~\cite{SIAMCOMP:SahWat21}.} Thus, we can reduce the collusion-resistance to the single-key security of the underlying single-decryptor PKFE. Note that we need to consider super-polynomially many hybrid games to complete the proof since the tag space size must be super-polynomial to treat \emph{unbounded} polynomially many tags. This is the reason why we need the sub-exponential security for building blocks.

\ifnum\submission=1

\subsection{Organization}
Due to the space limitation, we provide preliminaries including notations in \cref{sec:prelim}.
Also, we focus on SKFE-SKL in the main body, and we provide results on single decryptor FE in \cref{sec:single_dec_FE}.
We provide a high-level overview of our single-decryptor FE in~\cref{sec:tech_overview}.

In \cref{sec:def_SKFE_SKL}, we provide the definition of SKFE-SKL, and its variants SKFE-sbSKL and index-based SKFE-sbSKL.
In \cref{sec:iSKFESKL}, we construct an index-based SKFE-sbSKL scheme.
In \cref{sec:removing_index}, we show how to transform an index-based SKFE-sbSKL scheme into an SKFE-sbSKL scheme.
In \cref{sec:dynamic_SKFESKL}, we show how to construct an SKFE-SKL scheme from an SKFE-sbSKL scheme.
\else

\ifnum\submission=0
\subsection{More on Related Work}\label{sec:more_related_work}
\else
\section{More on Related Work}\label{sec:more_related_work}
\fi

\paragraph{Secure software leasing.}
Ananth and La Place~\cite{EC:AnaLaP21} achieve secure software leasing for $\subEvs$ from the QLWE assumption and IO.
Kitagawa, Nishimaki, and Yamakawa~\cite{TCC:KitNisYam21} achieve secure software leasing for PRFs and $\subEvs$ only from the QLWE (that is, without IO).
Broadbent, Jeffery, Lord, Podder, and Sundaram~\cite{TCC:BJLPS21} achieve secure software leasing for $\subEvs$ without assumptions.
Coladangelo, Majenz, and Poremba~\cite{EPRINT:ColMajPor20} also achieve secure software leasing for $\subEvs$ in the quantum random oracle (QROM) model (without assumptions).

In secure software leasing, we force adversaries to run pirate software by an \emph{honest} evaluation algorithm denoted by $\run$~\cite{EC:AnaLaP21}.
An honest evaluation algorithm verifies the format of software when software is executed, and adversaries cannot adopt arbitrary strategies for creating pirate software.
Aaronson, Liu, Liu, Zhandry, and Zhang~\cite{C:ALLZZ21} introduce copy-detection, which is essentially the same notion as secure software leasing. Although a check algorithm verifies returned software, there is no honest evaluation algorithm, and adversaries can adopt arbitrary strategies in the copy-detection setting. Aaronson et al.~\cite{C:ALLZZ21} achieve copy-detection for PRFs, PKE, and signatures from IO and standard cryptographic assumptions.
Those cryptographic functionalities are much weaker than FE.
Moreover, we need IO to achieve security against adversaries that adopt any strategy to execute pirate software.
Secure software leasing by Coladangelo et al.~\cite{EPRINT:ColMajPor20} is also secure against such arbitrary adversaries, but their construction relies on QROM, and its functionality is severely limited.

\paragraph{Certified deletion.}
Broadbent and Islam~\cite{TCC:BroIsl20} present the notion of quantum encryption with certified deletion. They achieve one-time SKE with certified deletion without assumptions.
Hiroka, Morimae, Nishimaki, and Yamakawa~\cite{AC:HMNY21} extend the notion to reusable SKE, PKE, and ABE with certified deletion, and instantiate them with standard SKE, PKE, and IO and OWFs, respectively.
Our FE with secure key leasing can also be seen as certified deletion for functional decryption keys.

\paragraph{Copy-protection.}
Aaronson~\cite{CCC:Aaronson09} achieves copy-protection for arbitrary unlearnable Boolean functions \emph{relative to a quantum oracle}, and presents two \emph{heuristic} copy-protection schemes for point functions.
Aaronson et al.~\cite{C:ALLZZ21} achieve quantum copy-protection for all unlearnable functions by using \emph{classical oracles}.
Georgiou and Zhandry~\cite{EPRINT:GeoZha20} present the notion of SDE and instantiate it with one-shot signatures~\cite{STOC:AGKZ20} and extractable witness encryption (WE)~\cite{C:GKPVZ13}. They also introduce the notion of broadcast encryption with unclonable decryption and splittable ABE, which can be seen as copy-protected variants of broadcast encryption and ABE, respectively. They instantiate splittable ABE with the same tools as those for SDE. They instantiate broadcast encryption with unclonable decryption with tokenized signatures~\cite{EPRINT:BenSat17}, extractable WE, and collision-resistant hash functions.
Coladangelo et al.~\cite{C:CLLZ21} achieve quantum copy-protection for PRFs and SDE from IO and OWFs, the QLWE assumption, and the strong monogamy-of-entanglement conjecture.
Later, Culf and Vidick prove that the strong monogamy-of-entanglement conjecture is true without any assumption~\cite{ARXIV:CulVid21}.
Those cryptographic functionalities are weaker than FE, as in the case of secure software leasing.
Most of those works rely on magical oracles or knowledge-type assumptions such as the existence of extractable WE, which is implausible~\cite{ALGMC:GGHW17}.
Only the constructions by Coladangelo et al.~\cite{C:CLLZ21,ARXIV:CulVid21} do not rely on such strong tools.
However, note that there is no provably secure \emph{post-quantum} IO from well-founded assumptions so far.\footnote{There are some candidates~\cite{TCC:BGMZ18,TCC:CHVW19,EC:AgrPel20,TCC:DQVWW21}. The known IO constructions from well-founded assumptions~\cite{STOC:JaiLinSah21,EC:JaiLinSah22} are vulnerable against quantum adversaries.}

Ben-David and Sattath~\cite{EPRINT:BenSat17} present the notion of tokenized signatures, where we can generate a delegated signing token from a signing key. A signing token enables a user to sign one and only one message. They instantiate it with virtual black-box obfuscation~\cite{JACM:BGIRSVY12}.
Amos, Georgiou, Kiayias, and Zhandry present the notion of one-shot signatures, where adversaries cannot generate two valid signatures for different two messages even under an adversarially generated verification key. They instantiate it with \emph{classical oracles}, which can be seen as an idealized virtual black-box obfuscation. Those are copy-protected variants of signatures. Relationships between (copy-protected) FE and them are not known so far.

\paragraph{Functional encryption.}
FE has been extensively studied since Boneh, Sahai, and Waters~\cite{TCC:BonSahWat11} formally defined its notion.
Although there are many works on FE, we focus on FE for $\Ppoly$ that are closely related to our work in this paper.

Sahai and Seyalioglu~\cite{CCS:SahSey10} present the first (selectively secure) single-key PKFE for $\Ppoly$ from PKE.
Gorbunov, Vaikuntanathan, and Wee~\cite{C:GorVaiWee12} present transformations from (adaptively secure) single-key PKFE for $\NCone$ into (adaptively secure) bounded collusion-resistant PKFE for $\Ppoly$ by using pseudorandom generators (PRGs) in $\NCone$.
The transformation also converts single-key and many-ciphertext SKFE for $\NCone$ into bounded collusion-resistant and many-ciphertext SKFE for $\Ppoly$.
They also show that we can achieve adaptively secure single-key PKFE (resp. SKFE) for $\NCone$ from PKE (resp. OWFs).
Ananth, Brakerski, Segev, and Vaikuntanathan~\cite{C:ABSV15} observe that we can invert the roles of the functions and plaintexts in SKFE by using the function-privacy of SKFE~\cite{JC:BraSeg18}, and obtain collusion-resistant and bounded many-ciphertext SKFE for $\Ppoly$ from OWFs by using the transformations.
Later, Ananth and Vaikuntanathan~\cite{TCC:AnaVai19} improve the assumptions (remove PRGs in $\NCone$) for the transformation and achieve the optimal ciphertext size (in the asymptotic sense).
These transformations~\cite{C:GorVaiWee12,TCC:AnaVai19} run $N$ independent copies of a single-key FE scheme and encrypt the views of some $N$-party multi-party computation (MPC) protocol.
Agrawal and Rosen~\cite{TCC:AgrRos17} construct bounded collusion-resistant PKFE for $\NCone$ by using FE for inner-product~\cite{PKC:ABDP15,C:AgrLibSte16} and the homomorphic property of some encryption scheme~\cite{C:BraVai11}.\footnote{We can upgrade FE for $\NCone$ to FE for $\Ppoly$ by using randomized encoding~\cite{C:ABSV15}.} Their construction supports arithmetic circuits.
Agrawal, Maitra, Vempati, and Yamada~\cite{C:AMVY21} and Garg, Goyal, Lu, and Waters~\cite{EPRINT:GGLW21} concurrently and independently present the notion of dynamic bounded collusion-resistant FE and how to achieve it from identity-based encryption (IBE).
In the dynamic bounded collusion setting, the setup algorithm of FE does not depend on the number of key queries, and only the encryption algorithm of FE depends on it.

Security amplification is transforming an $\epsilon$-secure cryptographic scheme into a $\nu$-secure one, where $\epsilon \in (0,1)$ is a constant and $\nu$ is a negligible function.
Ananth, Jain, Lin, Matt, and Sahai~\cite{C:AJLMS19} propose security amplification techniques for standard FE.
Their techniques are based on multi-key homomorphic encryption, which can be instantiated with the LWE assumption.
Jain, Korb, Manohar, and Sahai~\cite{C:JKMS20} also propose how to amplify the security of standard FE by using SetHSS, which can be instantiated with OWFs.

IO for $\Ppoly$ is necessary and sufficient for collusion-resistant FE (for $\NCone$) up to sub-exponential security loss~\cite{SIAMCOMP:GGHRSW16,JACM:BitVai18,C:AnaJai15,EPRINT:AnaJaiSah15,JC:KitNisTan22}.
Jain, Lin, and Sahai~\cite{STOC:JaiLinSah21,EC:JaiLinSah22} achieve the first IO (and collusion-resistant FE for $\Ppoly$) from well-founded assumptions. However, they are not post-quantum secure since the constructions heavily rely on cryptographic bilinear maps.

Revocation mechanisms have been extensively studied, especially in the context of broadcast encryption~\cite{FC:NaoPin00,C:NaoNaoLot01}.\footnote{We omit references since there are too many previous works on broadcast encryption and trace-and-revoke systems.}
Nishimaki et al.~\cite{EC:NisWicZha16} introduce revocable FE to achieve a trace-and-revoke system with flexible identities. They achieve the first revocable FE by using IO and OWFs. We need to manage a revocation list and update ciphertexts to revoke a user.

\fi


\ifnum\submission=0

\section{Preliminaries}\label{sec:prelim}

\paragraph{Notations and conventions.}
In this paper, standard math or sans serif font stands for classical algorithms (e.g., $C$ or $\algo{Gen}$) and classical variables (e.g., $x$ or $\keys{pk}$).
Calligraphic font stands for quantum algorithms (e.g., $\qalgo{Gen}$) and calligraphic font and/or the bracket notation for (mixed) quantum states (e.g., $\qstateq$ or $\ket{\psi}$).
For strings $x$ and $y$, $x \concat y$ denotes the concatenation of $x$ and $y$.
Let $\mv{0}$ denote a string consisting of an appropriate number of $0$.
Let $[\ell]$ denote the set of integers $\{1, \cdots, \ell \}$, $\secp$ denote a security parameter, and $y \seteq z$ denote that $y$ is set, defined, or substituted by $z$.

In this paper, for a finite set $X$ and a distribution $D$, $x \chosen X$ denotes selecting an element from $X$ uniformly at random, $x \chosen D$ denotes sampling an element $x$ according to $D$. Let $y \gets \algo{A}(x)$ and $y \gets \qalgo{A}(\qstate{x})$ denote assigning to $y$ the output of a probabilistic or deterministic algorithm $\algo{A}$ and a quantum algorithm $\qalgo{A}$ on an input $x$ and $\qstate{x}$, respectively. When we explicitly show that $\algo{A}$ uses randomness $r$, we write $y \gets \algo{A}(x;r)$.
PPT and QPT algorithms stand for probabilistic polynomial-time algorithms and polynomial-time quantum algorithms, respectively.
Let $\negl$ denote a negligible function.
\else\fi

\ifnum\submission=0

\subsection{Standard Cryptographic Tools}\label{sec:crypto_tools}

\begin{definition}[Pseudorandom Function]\label{def:prf}
Let $\{\algo{F}_{K}: \bin^{\ell_1} \ra \allowbreak \bin^{\ell_2} \mid K \in \bin^\secp\}$ be a family of polynomially computable functions, where $\ell_1$ and $\ell_2$ are some polynomials of $\secp$.
We say that $\prf$ is a pseudorandom function (PRF) family if, for any PPT distinguisher $\cA$, there exists $\negl(\cdot)$ such that it holds that
\begin{align}
\abs{\Pr[\qA^{\algo{F}_{K}(\cdot)}(1^\secp) \out 1 \mid K \chosen \bin^{\secp}]
-\Pr[\qA^{\Rand(\cdot)}(1^\secp) \out 1 \mid \Rand \chosen \cU]
} \le\negl(\secp),
\end{align}
where $\cU$ is the set of all functions from $\bin^{\ell_1}$ to $\bin^{\ell_2}$.
\end{definition}

\begin{theorem}[\cite{JACM:GolGolMic86}]\label{thm:prf-owf} If one-way functions exist, then for all efficiently computable functions $n(\lambda)$ and $m(\lambda)$, there exists a PRF that maps $n(\lambda)$ bits to $m(\lambda)$ bits.
\end{theorem}

\begin{definition}[Secret Key Encryption]\label{def:ske}
An SKE scheme $\SKE$ is a two tuple $(\E, \D)$ of PPT algorithms. 
\begin{itemize}
\item The encryption algorithm $\E$, given a key $K \in \bin^\lambda$ and a plaintext $m \in \M$, outputs a ciphertext $\ct$,
where $\M$ is the plaintext space of $\SKE$.

\item The decryption algorithm $\D$, given a key $K$ and a ciphertext $\ct$, outputs a plaintext $\tilde{m} \in \{ \bot \} \cup \M$.
This algorithm is deterministic.
\end{itemize}

\begin{description}
\item[Correctness:] We require $\D(K, \E(K, m)) = m$ for every $m \in \M$ and key $K \in \bin^\lambda$.
\item[CPA Security:]
We define the following experiment $\expa{\qA,\SKE}{cpa}(1^\secp,\coin)$ between a challenger and an adversary $\qA$.
\begin{enumerate}
\item The challenger generates $K \chosen \bin^\lambda$.
Then, the challenger sends $1^\lambda$ to $\qA$.
\item $\qA$ may make polynomially many encryption queries adaptively.
$\qA$ sends $(m_0, m_1) \in \M \times \M$ to the challenger.
Then, the challenger returns $\ct \la \E(K, m_\coin)$.
\item $\qA$ outputs $\coin' \in \bin$.
\end{enumerate}
We say that $\SKE$ is CPA secure if for any QPT adversary $\qA$, we have 
\[
\adva{\SKE, \qA}{cpa}(\lambda) = \abs{\Pr[\expa{\qA,\SKE}{cpa}(1^\secp,0)=1]-\Pr[\expa{\qA,\SKE}{cpa}(1^\secp,1)=1]} \le \negl(\secp).
\]
\end{description}
\end{definition}

\begin{definition}[Ciphertext Pseudorandomness for SKE]\label{def:ske_pseudorandomct}
Let $\zo{\ell}$ be the ciphertext space of $\SKE$.
We define the following experiment $\expb{\qA,\SKE}{pr}{ct}(1^\secp,\coin)$ between a challenger and an adversary $\qA$.
\begin{enumerate}
\item The challenger generates $K \chosen \bin^\lambda$.
Then, the challenger sends $1^\lambda$ to $\qA$.
\item $\qA$ may make polynomially many encryption queries adaptively.
$\qA$ sends $m \in \M $ to the challenger.
Then, the challenger returns $\ct \la \E(K, m)$ if $\coin=0$, otherwise $\ct \chosen \zo{\ell}$.
\item $\qA$ outputs $\coin' \in \bin$.
\end{enumerate}
We say that $\SKE$ is pseudorandom-secure if for any QPT adversary $\qA$, we have 
\[
\advb{\SKE, \qA}{pr}{ct}(\lambda) = \abs{\Pr[\expb{\qA,\SKE}{pr}{ct}(1^\secp,0)=1]-\Pr[\expb{\qA,\SKE}{pr}{ct}(1^\secp,1)=1]} \le \negl(\secp).
\]
\end{definition}

\begin{theorem}\label{thm:pseudorandom_ske}
If OWFs exist, there exists a pseudorandom-secure SKE scheme.
\end{theorem}

\begin{definition}[Secret-Key Functional Encryption]\label{def:SKFE}
An SKFE scheme $\SKFE$ is a tuple of four PPT algorithms $(\Setup, \KG, \Enc, \Dec)$. 
Below, let $\cX$, $\cY$, and $\cF$ be the plaintext, output, and function spaces $\SKFE$, respectively.
\begin{description}
\item[$\Setup(1^\secp,1^{\numkey})\ra\msk$:] The setup algorithm takes a security parameter $1^\lambda$ and a collusion bound $1^{\numkey}$, and outputs a master secret key $\msk$.
\item[$\KG(\msk,f)\ra\sk_f$:] The key generation algorithm takes a master secret key $\msk$ and a function $f \in \calF$, and outputs a functional decryption key $\sk_f$.

\item[$\Enc(\msk,x)\ra\ct$:] The encryption algorithm takes a master secret key $\msk$ and a plaintext $x \in \cX$, and outputs a ciphertext $\ct$.

\item[$\Dec(\sk_f,\ct)\ra y$:] The decryption algorithm takes a functional decryption key $\sk_f$ and a ciphertext $\ct$, and outputs  $y \in \{ \bot \} \cup \cY$.

\item[Correctness:] We require that for every $x \in \cX$, $f \in \calF$, $\numkey\in\bbN$, we have that
\[
\Pr\left[
\Dec(\sk_f, \ct) = f(x)
 \ \middle |
\begin{array}{rl}
 &\msk \la \Setup(1^\lambda,1^{\numkey}),\\
 & \sk_f \gets \KG(\msk,f), \\
 &\ct \gets \Enc(\msk,x)
\end{array}
\right]=1 - \negl(\secp).
\]
\end{description}
\end{definition}

\begin{definition}[Selective Indistinguishability-Security]\label{def:sel_ind_SKFE}
We say that $\SKFE$ is a selectively indistinguishability-secure SKFE scheme for $\Xs,\Ys$, and $\Fs$, if it satisfies the following requirement, formalized from the experiment $\expb{\qA,\SKFE}{sel}{ind}(1^\secp,\coin)$ between an adversary $\qA$ and a challenger:
        \begin{enumerate}
            \item At the beginning, $\qA$ sends $(1^{\numkey},x_0^*,x_1^*)$ to the challenger. The challenger runs $\msk\gets\Setup(1^\secp,1^\numkey)$. Also, the challenger generates $\ct^*\la\Enc(\msk,x_\coin^*)$ and sends $\ct^*$ to $\qA$.
            Throughout the experiment, $\qA$ can access the following oracles.
            \begin{description}
            \item[$\Oracle{\Enc}(x)$:] Given $x$, it returns $\Enc(\msk,x)$.
            \item[$\Oracle{\KG}(f)$:] Given $f$, if $f(x_0^*)\ne f(x_1^*)$, it returns $\bot$. Otherwise, it returns $\KG(\msk,f)$. $\qA$ can access this oracle at most $\numkey$ times.
            \end{description}
            \item $\qA$ outputs a guess $\coin^\prime$ for $\coin$. The challenger outputs $\coin'$ as the final output of the experiment.

        \end{enumerate}
        We say that $\SKFE$ is selectively indistinguishability-secure if, for any QPT $\qA$, it holds that
\begin{align}
\advb{\SKFE,\qA}{sel}{ind}(\secp) \seteq \abs{\Pr[\expb{\SKFESKL,\qA}{sel}{ind} (1^\secp,0) \out 1] - \Pr[\expb{\SKFE,\qA}{sel}{ind} (1^\secp,1) \out 1] }\le \negl(\secp).
\end{align}
\end{definition}

It is easy to see that we can consider the adaptively indistinguishability-secure variant as in~\cref{def:ad_ind_PKFE}.
\begin{theorem}[\cite{C:GorVaiWee12}]\label{thm:skfe_from_owf}
If there exist OWFs, then there exists selectively indistinguishability-secure SKFE for $\Ppoly$.
\end{theorem}

\subsection{Advanced Tools}\label{sec:advenced_tools}

We present the definitions for reusable SKE with certified deletion introduced by Hiroka et al.~\cite{AC:HMNY21}
\begin{definition}[Reusable SKE with Certified Deletion (Syntax)]\label{def:reusable_sk_cert_del}
A secret key encryption scheme with certified deletion is a tuple of quantum algorithms $(\KG,\qencrypt,\qdecrypt,\qdelete,\Vrfy)$ with plaintext space $\Ms$ and key space $\Ks$.
\begin{description}
    \item[$\KG (1^\secp) \ra \sk$:] The key generation algorithm takes as input the security parameter $1^\secp$ and outputs a secret key $\sk \in \Ks$.
    \item[$\qencrypt(\sk,m) \ra (\vk,\qct)$:] The encryption algorithm takes as input $\sk$ and a plaintext $m\in\Ms$ and outputs a verification key $\vk$ and a ciphertext $\qct$.
    \item[$\qdecrypt(\sk,\qct) \ra m^\prime$:] The decryption algorithm takes as input $\sk$ and $\qct$ and outputs a plaintext $m^\prime \in \Ms$ or $\bot$.
    \item[$\qdelete(\qct) \ra \cert$:] The deletion algorithm takes as input $\qct$ and outputs a certification $\cert$.
    \item[$\Vrfy(\vk,\cert)\ra \top \mbox{ or }\bot$:] The verification algorithm takes $\vk$ and $\cert$ and outputs $\top$ or $\bot$.

\item[Decryption correctness:] There exists a negligible function $\negl$ such that for any $m\in\Ms$, 
\begin{align}
\Pr\left[
\qdecrypt(\sk,\qct)= m
\ \middle |
\begin{array}{ll}
\sk\lrun \KG(1^\secp)\\
(\vk,\qct) \lrun \qencrypt(\sk,m)
\end{array}
\right] 
=1-\negl(\secp).
\end{align}

\item[Verification correctness:] There exists a negligible function $\negl$ such that for any $m\in\Ms$, 
\begin{align}
\Pr\left[
\Vrfy(\vk,\cert)=\top
\ \middle |
\begin{array}{ll}
\sk\lrun \KG(1^\secp)\\
(\vk,\qct) \lrun \qencrypt(\sk,m)\\
\cert \lrun \qdelete(\qct)
\end{array}
\right] 
=1-\negl(\secp).
\end{align}
\end{description}
\end{definition}

\begin{definition}[IND-CPA-CD Security for Reusable SKE with Certified Deletion]\label{def:reusable_sk_certified_del}
Let $\Sigma=(\KG, \qencrypt, \qdecrypt, \qdelete, \Vrfy)$ be a secret key encryption with certified deletion.
We consider the following security experiment $\expc{\Sigma,\qA}{sk}{cert}{del}(\secp,b)$.

\begin{enumerate}
    \item The challenger computes $\sk \la \KG(1^\secp)$.
    \item $\qA$ sends an encryption query $m$ to the challenger. The challenger computes $(\vk,\qct)\lrun \qencrypt(\sk,m)$ to $\qA$ and returns $(\vk,\qct)$ to $\qA$. This process can be repeated polynomially many times.
    \item $\qA$ sends $(m_0,m_1)\in\cM^2$ to the challenger. 
    \item The challenger computes $(\vk_b,\qct_b) \la \qencrypt(\sk,m_b)$ and sends $\qct_b$ to $\qA$.
    \item Again, $\qA$ can send encryption queries.
    \item At some point, $\qA$ sends $\cert$ to the challenger.
    \item The challenger computes $\Vrfy(\vk_b,\cert)$. If the output is $\bot$, the challenger sends $\bot$ to $\qA$.
    If the output is $\top$, the challenger sends $\sk$ to $\qA$.
    \item If the challenger sends $\bot$ in the previous item, $\qA$ can send encryption queries again.
    \item $\qA$ outputs $b'\in \bit$.
\end{enumerate}
We say that $\Sigma$ is IND-CPA-CD secure if for any QPT $\qA$, it holds that
\begin{align}
\advc{\Sigma,\qA}{sk}{cert}{del}(\secp)\seteq \abs{\Pr[ \expc{\Sigma,\qA}{sk}{cert}{del}(\secp, 0)=1] - \Pr[ \expc{\Sigma,\qA}{sk}{cert}{del}(\secp, 1)=1] }\le \negl(\secp).
\end{align}
\end{definition}

We introduce an additional property for certificates.
\begin{definition}[Unique Certificate]\label{def:unique_cert_CDSKE}
We say that an SKE scheme with certified deletion has the unique certificate property if there is at most one certification $\cert$ for each $\vk$ output by $\Enc$ such that $\top = \Vrfy(\vk,\cert)$.
\end{definition}

Known constructions of SKE (and PKE) with certified deletion based on BB84 states has the unique certificate property.
\begin{theorem}[\cite{AC:HMNY21}]\label{thm:ske_cert_del_owf}
If there exist OWFs, there exists an IND-CPA-CD secure SKE scheme that has the unique certificate property.
\end{theorem}

We introduce a variant of IND-CPA-CD security where the adversary can send many verification queries, called indistinguishability against Chosen Verification Attacks (CVA). We use this security notion to achieve SKFE with secure key leasing in~\cref{sec:iSKFESKL}.
\begin{definition}[IND-CVA-CD Security for Reusable SKE with Certified Deletion]\label{def:reusable_sk-vo_certified_del}
Let $\Sigma=(\KG, \qencrypt, \qdecrypt, \qdelete, \Vrfy)$ be a secret key encryption with certified deletion.
We consider the following security experiment $\expc{\Sigma,\qA}{sk}{cert}{vo}(\secp,b)$.

\begin{enumerate}
    \item The challenger computes $\sk \la \KG(1^\secp)$.
    \item $\qA$ sends an encryption query $m$ to the challenger. The challenger computes $(\vk,\ct)\lrun \qencrypt(\sk,m)$ to $\qA$ and returns $(\vk,\ct)$ to $\qA$. This process can be repeated polynomially many times.
    \item $\qA$ sends $(m_0,m_1)\in\cM^2$ to the challenger. 
    \item The challenger computes $(\vk_b,\qct_b) \la \qencrypt(\sk,m_b)$ and sends $\qct_b$ to $\qA$.
    \item Again, $\qA$ can send encryption queries. $\qA$ can also send a verification query $\cert$ to the challenger. The challenger returns $\sk$ if $\top = \Vrfy(\vk_b,\cert)$, $\bot$ otherwise. This process can be repeated polynomially many times.
    \item $\qA$ outputs $b'\in \bit$.
\end{enumerate}
We say that $\Sigma$ is IND-CVA-CD secure if for any QPT $\qA$, it holds that
\begin{align}
\advc{\Sigma,\qA}{sk}{cert}{vo}(\secp)\seteq \abs{\Pr[ \expc{\Sigma,\qA}{sk}{cert}{vo}(\secp, 0)=1] - \Pr[ \expc{\Sigma,\qA}{sk}{cert}{vo}(\secp, 1)=1] }\le \negl(\secp).
\end{align}
\end{definition}

Known IND-CPA-CD secure SKE schemes satisfy IND-CVA-CD security thanks to the following theorem since they satisfy~\cref{def:unique_cert_CDSKE}.
\begin{theorem}\label{thm:cert_del_vo_equivalence_for_unique_cert}
If an SKE scheme is IND-CPA-CD secure and has the unique certificate property, then it also satisfies IND-CVA-CD security.
\end{theorem}

\begin{proof}
We construct an adversary $\qB$ for IND-CPA-CD by using an adversary $\qA$ for IND-CVA-CD as follows.
\begin{enumerate}
\item At the beginning, the adversary $\qA$ is fixed, so the number of the verification query by $\qA$ is also fixed. We denote the number by $q$.
\item When $\qA$ sends an encryption query $m$, $\qB$ sends $m$ to its challenger, receives $(\vk,\qct)\gets \qEnc(\sk,m)$, and passes $(\vk,\qct)$ to $\qA$.
\item When $\qA$ sends $(m_0,m_1)$, $\qB$ sends $(m_0,m_1)$ to its challenger, receives $\qct_b$, passes $\qct_b$ to $\qA$.
\item Then, $\qB$ guesses an index $i^\ast \chosen [q+1]$.
\item When $\qA$ sends $\cert_i$ as the $i$-th verification query, if $i<i^\ast$, $\qB$ returns $\bot$. When $\qA$ sends $\cert_{i^\ast}$ as the $i^\ast$-th verification query, if there exits $i$ such that $\cert_i=\cert_{i^\ast}$, $\qB$ aborts. Otherwise, $\qB$ forwards $\cert_{i^\ast}$ to its challenger. If $\qB$ receives $\bot$, $\qB$ aborts. If $\qB$ receives $\sk$, then $\qB$ passes $\sk$ to $\qA$. After $\qB$ receives $\sk$, if a query $\cert_i = \cert_{i^\ast}$, $\qB$ answers $\sk$, otherwise $\bot$.
\item $\qB$ outputs what $\qA$ outputs.
\end{enumerate}
Since we consider unique certificate SKE with certified deletion, if $\cB$ receives $\sk$ from its challenger, it is the only one correct certificate under $\vk_b$. Hence, $\qB$ can correctly simulate the IND-CVA-CD game with probability $1/(q+1)$ (Note that $i^\ast= q+1$ means $\qA$ did not output a correct certificate). Therefore, we obtain $\advc{\Sigma,\qA}{sk}{cert}{del}(\secp) = \frac{1}{q+1}\advc{\Sigma,\qA}{sk}{cert}{vo}(\secp)$, and the proof completes.
\end{proof}

\begin{definition}[Set Homomorphic Secret Sharing]\label{def:SetHSS}
A set homomorphic secret sharing scheme $\SetHSS$ consists of three algorithms $(\InpEncode,\FuncEncode,\Decode)$.
Below, we let $\cX$, $\cY$, and $\cF$ be the input, output, and function spaces of $\SetHSS$, respectively.

\begin{description}
\item[$\SetGen(1^\secp)\ra\params\seteq (p,\numall,(T_\instance)_{\instance\in[\numset]})$:]The set generation algorithm takes as input a security parameter, and outputs parameters $p$ and $\numall$, and a collection of $\numset$ sets $(T_\instance)_{\instance\in[\numset]}$, where each set $T_\instance\subseteq [\numall]$.
\item[$\InpEncode(\params,x)\ra (\share_\instance)_{\instance\in[\numset]}$:] The input encoding algorithm takes as input $\params$ output by $\SetGen$ and an input $x\in\cX$, and outputs a set of input shares $(\share_\instance)_{\instance\in[\numset]}$.
\item[$\FuncEncode(\params,f)\ra (f_{\instance})_{\instance\in[\numset]}$:] The function encoding algorithm takes as input $\params$ output by $\SetGen$ and a function $f\in\cF$, and outputs a set of function shares $(f_{\instance})_{\instance\in[\numset]}$.
\item[$\Decode((f_{\instance}(\share_{\instance}))_{\instance\in[\numset]})\ra y$:] The decoding algorithm takes as input a set of evaluations of function shares on their respective input shares $(f_{\instance}(\share_{\instance}))_{\instance\in[\numset]}$ and outputs a value $y\in\cY$.

\item[Correctness:] We require that for every $x \in \cX$, $f \in \calF$, we have that
\[
\Pr\left[
\Decode((f_{\instance}(\share_{\instance}))_{\instance\in[\numset]}) \ne f(x)
 \ \middle |
\begin{array}{rl}
&\params=(p,\numall,(T_\instance)_{\instance\in[\numset]})\la\SetGen(1^\secp),\\
 &(\share_{\instance})_{\instance\in[\numset]}\la\InpEncode(\params,x), \\
 &(f_{\instance})_{\instance\in[\numset]}\la\FuncEncode(\params,f)
\end{array}
\right]=\negl(\secp).
\]

\item[Existence of Unmarked Element:]Let $\params=(p,\numall,(T_\instance)_{\instance\in[\numset]})\la\SetGen(1^\secp)$. Suppose we randomly generate a set $S\subseteq[\numset]$ so that each $\instance\in[\numset]$ is independently included in $S$ with probability at most $p$.
Then, without negligible probability, there exists $\element\in[\numall]$ such that $\element\notin\bigcup_{\instance\in S}T_\instance$.

\item[Selective Indistinguishability-Security:]Consider the following experiment $\expb{\qA,\SetHSS}{sel}{ind}(1^\secp,\coin)$ between an adversary $\qA$ and a challenger:
        \begin{enumerate}
            \item The challenger generates $\params=(p,\numall,(T_\instance)_{\instance\in[\numset]})\la\SetGen(1^\secp)$ and sends $\params$ to $\qA$. $\qA$ sends $(\element^*,x_0^*,x_1^*)$ to the challenger, where $\element^*\in[\numall]$. The challenger runs $(\share_{\instance})_{\instance\in[\numset]}\la\InpEncode(\params,x_\coin)$ and sends $(\share_{\instance})_{\instance\in[\numset]_{\element^* \notin}}$ to $\qA$, where $[\numset]_{\element^* \notin}$ denotes the subset of $[\numset]$ consisting of $\instance$ such that $\element^* \notin T_\instance$.
            $\qA$ can access the following oracle.
            \begin{description}
            \item[$\Oracle{\FuncEncode}(f)$:] Given $f$, if $f(x_0^*)\ne f(x_1^*)$, it returns $\bot$. Otherwise, it generates $(f_{\instance})_{\instance\in[\numset]}\la\FuncEncode(\params,f)$ and returns $(f_{\instance},f_{\instance}(\share_{\instance}))_{\instance\in[\numset]}$.
            \end{description}
            \item $\qA$ outputs a guess $\coin^\prime$ for $\coin$. The challenger outputs $\coin'$ as the final output of the experiment.

        \end{enumerate}
        We say that $\SetHSS$ is selectively indistinguishability-secure if, for any QPT $\qA$, it holds that
\begin{align}
\advb{\SetHSS,\qA}{sel}{ind}(\secp) \seteq \abs{\Pr[\expb{\SetHSS,\qA}{sel}{ind} (1^\secp,0) \out 1] - \Pr[\expb{\SetHSS,\qA}{sel}{ind} (1^\secp,0) \out 1] } \le \negl(\secp).
\end{align}

\end{description}

\end{definition}

\begin{theorem}[\cite{C:JKMS20}]\label{thm:sethss_owf}
If there exist OWFs, there exists a set homomorphic secret sharing.
\end{theorem}

The above definition of SetHSS is different from the original one by Jain et al.~\cite{C:JKMS20} in that the original definition does not have the set generation algorithm and does not require the existence of unmarked element property.
Jain et al. defined SetHSS so that it works for any collection of sets.
Then, Jain et al. separately introduced the set generation algorithm and properties of it, which perfectly match the security bound amplification for FE.
We can obtain \cref{thm:sethss_owf} from ~\cite[Theorem 5.3, Lemma 6.1 and 6.2, in eprint ver.]{C:JKMS20} and the parmeter setting provided in the proof of ~\cite[Theorem 8.1 in eprint ver.]{C:JKMS20}.
In their parameter setting, $p=1/\poly(\secp)$, $\numall=\secp$, and $\numset=\polylog(\secp)$.
\else\fi


\section{Definition of SKFE with Secure Key Leasing}\label{sec:def_SKFE_SKL}

We introduce the definition of SKFE with secure key leasing and its variants.

\subsection{SKFE with Secure Key Leasing}
We first define SKFE with secure key leasing (SKFE-SKL).

\begin{definition}[SKFE with Secure Key Leasing]
An SKFE-SKL scheme $\SKFESKL$ is a tuple of six algorithms $(\Setup, \qKG, \Enc, \qDec,\qcert,\certvrfy)$. 
Below, let $\cX$, $\cY$, and $\cF$ be the plaintext, output, and function spaces of $\SKFESKL$, respectively.
\begin{description}
\item[$\Setup(1^\secp,1^{\numkey})\ra\msk$:] The setup algorithm takes a security parameter $1^\lambda$ and a collusion bound $1^{\numkey}$, and outputs a master secret key $\msk$.
\item[$\qKG(\msk,f,1^{\numct})\ra(\qfsk,\vk)$:] The key generation algorithm takes a master secret key $\msk$, a function $f \in \calF$, and an availability bound $1^{\numct}$, and outputs a functional decryption key $\qfsk$ and a verification key $\vk$.

\item[$\Enc(\msk,x)\ra\ct$:] The encryption algorithm takes a master secret key $\msk$ and a plaintext $x \in \cX$, and outputs a ciphertext $\ct$.

\item[$\qDec(\qfsk,\ct)\ra\tlx$:] The decryption algorithm takes a functional decryption key $\qfsk$ and a ciphertext $\ct$, and outputs a value $\tilde{x}$.

\item[$\qcert(\qfsk)\ra\cert$:] The certification algorithm takes a function decryption key $\qfsk$, and outputs a classical string $\cert$.

\item[$\Vrfy(\vk,\cert)\ra\top/\bot$:] The certification-verification algorithm takes a verification key $\vk$ and a string $\cert$, and outputs $\top$ or $\bot$.

\item[Decryption correctness:]For every $x \in \cX$, $f \in \calF$, and $\numkey,\numct\in\bbN$, we have
\begin{align}
\Pr\left[
\qDec(\qfsk, \ct) \allowbreak = f(x)
\ \middle |
\begin{array}{ll}
\msk \la \Setup(1^\lambda,1^{\numkey})\\
(\qfsk,\vk)\gets\qKG(\msk,f,1^{\numct})\\
\ct\gets\Enc(\msk,x)
\end{array}
\right] 
=1-\negl(\secp).
\end{align}

\item[Verification correctness:]For every $f \in \calF$ and $\numkey,\numct\in\bbN$, we have 
\begin{align}
\Pr\left[
\Vrfy(\vk,\cert)=\top
\ \middle |
\begin{array}{ll}
\msk \la \Setup(1^\lambda,1^{\numkey})\\
(\qfsk,\vk)\gets\qKG(\msk,f,1^{\numct})\\
\cert \lrun \qcert(\qfsk)
\end{array}
\right] 
=1-\negl(\secp).
\end{align}

\end{description}
\end{definition}

\begin{definition}[Selective Lessor Security]\label{def:sel_lessor_SKFESKL}
We say that $\SKFESKL$ is a selectively lessor secure SKFE-SKL scheme for $\Xs,\Ys$, and $\Fs$, if it satisfies the following requirement, formalized from the experiment $\expb{\qA,\SKFESKL}{sel}{lessor}(1^\secp,\coin)$ between an adversary $\qA$ and a challenger:
        \begin{enumerate}
            \item At the beginning, $\qA$ sends $(1^{\numkey},x_0^*,x_1^*)$ to the challenger. The challenger runs $\msk\gets\Setup(1^\secp,1^\numkey)$.
            Throughout the experiment, $\qA$ can access the following oracles.
            \begin{description}
            \item[$\Oracle{\Enc}(x)$:] Given $x$, it returns $\Enc(\msk,x)$.
            \item[$\Oracle{\qKG}(f,1^{\numct})$:] Given $(f,1^{\numct})$, it generates $(\qfsk,\vk)\la\qKG(\msk,f,1^{\numct})$, sends $\qfsk$ to $\qA$, and adds $(f,1^{\numct},\vk,\bot)$ to $\List{\qKG}$. $\qA$ can access this oracle at most $\numkey$ times.
            \item[$\Oracle{\Vrfy}(f,\cert)$:] Given $(f,\cert)$, it finds an entry $(f,1^\numct,\vk,M)$ from $\List{\qKG}$. (If there is no such entry, it returns $\bot$.) If $\top=\Vrfy(\vk,\cert)$ and the number of queries to $\Oracle{\Enc}$ at this point is less than $\numct$, it returns $\top$ and updates the entry into $(f,1^\numct,\vk,\top)$. Otherwise, it returns $\bot$.
            \end{description}
            \item When $\qA$ requests the challenge ciphertext, the challenger checks if for any entry $(f,1^\numct,\vk,M)$ in $\List{\qKG}$ such that $f(x_0^*)\ne f(x_1^*)$, it holds that $M=\top$. If so, the challenger generates $\ct^*\la\Enc(\msk,x_\coin^*)$ and sends $\ct^*$ to $\qA$. Otherwise, the challenger outputs $0$. Hereafter, $\qA$ is not allowed to sends a function $f$ such that $f(x_0^*)\ne f(x_1^*)$ to $\Oracle{\qKG}$.
            \item $\qA$ outputs a guess $\coin^\prime$ for $\coin$. The challenger outputs $\coin'$ as the final output of the experiment.

        \end{enumerate}
        For any QPT $\qA$, it holds that
\ifnum\llncs=0        
\begin{align}
\advb{\SKFESKL,\qA}{sel}{lessor}(\secp) \seteq \abs{\Pr[\expb{\SKFESKL,\qA}{sel}{lessor} (1^\secp,0) \out 1] - \Pr[\expb{\SKFESKL,\qA}{sel}{lessor} (1^\secp,1) \out 1] }\leq \negl(\secp).
\end{align}
\else
\begin{align}
\advb{\SKFESKL,\qA}{sel}{lessor}(\secp) 
&\seteq \abs{\Pr[\expb{\SKFESKL,\qA}{sel}{lessor} (1^\secp,0) \out 1] - \Pr[\expb{\SKFESKL,\qA}{sel}{lessor} (1^\secp,1) \out 1] }\\
&\leq \negl(\secp).
\end{align}
\fi
\end{definition}

\begin{remark}[On the adaptive security]
We can similarly define adaptive lessor security where we allow $\qA$ to adaptively choose the challenge plaintext pair $(x_0^*,x_1^*)$.
For standard FE, we can generically convert a selectively secure one into an adaptively secure one without any additional assumption~\cite{C:ABSV15}.
We observe that a similar transformation works for SKFE with secure key leasing.
Thus, for simplicity, we focus on selective lessor security in this work.
See~\cref{sec:adatpive_FESKL_security} for the definition and transformation.
\end{remark}

\begin{remark}[On the simulation-based security]
We can also define a simulation-based variant of selective/adaptive lessor security where a simulator simulates a challenge ciphertext without the challenge plaintext $x^\ast$ as the simulation-based security for standard FE~\cite{TCC:BonSahWat11,C:GorVaiWee12}.
We can generically convert indistinguishability-based lessor secure SKFE with secure key leasing into a simulation-based lessor secure one without any additional assumptions as standard FE~\cite{C:DIJOPP13}.
See~\cref{sec:simulation_FESKL_security} for the simulation-based definition and the transformation.
\end{remark}

\subsection{SKFE with Static-Bound Secure Key Leasing}\label{sec:def_SKFESKL_variants}

In this section, we define SKFE with static-bound secure key leasing (SKFE-sbSKL).
It is a weaker variant of SKFE-SKL in which a single availability bound $\numct$ applied to every decryption key is fixed at the setup time.
We design SKFE-sbSKL so that it can be transformed into SKFE-SKL in a generic way.
For this reason, we require an SKFE-sbSKL scheme to satisfy an efficiency requirement called weak optimal efficiency and slightly stronger variant of the lessor security notion.
\footnote{We borrow the term weak optimal efficiency from the paper by Garg, Goyal, Lu, and Waters~\cite{EPRINT:GGLW21}, which studies dynamic bounded collusion security for standard FE.}

Below, we first introduce the syntax of SKFE-sbSKL.
Then, before introducing the definition of (selective) lessor security for it, we provide the overview of the transformation to SKFE-SKL since we think the overview makes it easy to understand the security notion.

\begin{definition}[SKFE with Static-Bound Secure Key Leasing]
An SKFE-sbSKL scheme $\SKFEsbSKL$ is a tuple of six algorithms $(\Setup, \qKG, \Enc, \qDec,\qcert,\certvrfy)$.
The only difference from a normal SKFE scheme with secure key leasing is that $\qKG$ does not take as input the availability bound $\numct$, and instead, $\Setup$ takes it as an input. Moreover, $\Setup$ takes it in binary as $\Setup(1^\secp,1^{\numkey},\numct)$, and we require the following weak optimal efficiency.

\begin{description}
\item[Weak Optimal Efficiency:] We require that the running time of $\Setup$ and $\Enc$ is bounded by a fixed polynomial of $\lambda$, $\numkey$, and $\log \numct$.
\end{description}

\end{definition}

\paragraph{Overview of the transformation to SKFE-SKL.}
As seen above, $\Setup$ and $\Enc$ of an SKFE-sbSKL scheme $\SKFEsbSKL$ is required to run in time $\log \numct$.
This is because, in the transformation to SKFE-SKL, we use $\secp$ instances of $\SKFEsbSKL$ where the $k$-th instance is set up with the availability bound $2^k$ for every $k\in[\secp]$.
The weak optimal efficiency ensures that $\Setup$ and $\Enc$ of all $\secp$ instances run in polynomial time.
The details of the transformation are as follows.

We construct an SKFE-SKL scheme $\SKFESKL$ from an SKFE-sbSKL scheme  $\SKFEsbSKL=(\Setup,\qKG,\Enc,\qDec,\qcert,\Vrfy)$.
When generating a master secret key $\skl.\msk$ of $\SKFESKL$, we generate $\msk_k\la\Setup(1^\secp,1^\numkey,2^k)$ for every $k\in[\secp]$, and set $\skl.\msk\seteq (\msk_k)_{k\in[\secp]}$.
To encrypt $x$ by $\SKFESKL$, we encrypt it by all $\secp$ instances, that is, generate $\ct_k\la\Enc(\msk_k,x)$ for every $k\in[\secp]$.
The resulting ciphertext is $\skl.\ct\seteq (\ct_k)_{k\in[\secp]}$.
To generate a decryption key of $\SKFESKL$ for a function $f$ and an availability bound $\numct$, we first compute $k'\in[\secp]$ such that $2^{k'-1}\leq \numct \leq 2^{k'}$.
Then, we generate $(\qfsk_{k'},\vk_{k'})\la\qKG(\msk_{k'},f)$. The resulting decryption key is $\skl.\qfsk\seteq (k',\qfsk_{k'})$ and the corresponding verification key is $\vk\seteq \vk_{k'}$.
Decryption is performed by decrypting $\ct_{k'}$ included in $\skl.\ct\seteq (\ct_k)_{k\in[\secp]}$ by $\qfsk_{k'}$.
The certification generation and verification of $\SKFESKL$ are simply those of $\SKFESKL$.

We now consider the security proof of $\SKFESKL$.
In the experiment $\expb{\qA,\SKFESKL}{sel}{lessor}(1^\secp,0)$, an adversary $\qA$ is given the challenge ciphertext $\skl.\ct^*\seteq (\ct^*_k)_{k\in[\secp]}$, where $\ct^*_k\la\Enc(\msk_k,x_0^*)$ for every $k\in[\secp]$.
The proof is done if we can switch all of $\ct^*_k$ into $\Enc(\msk_k,x_1^*)$ without being detected by $\qA$.
To this end, the underlying $\SKFEsbSKL$ needs to satisfy a stronger variant of lessor security notion where an adversary is allowed to declare the availability bound such that $\qKG$ does not run in polynomial time, if the adversary does not make any query to the key generation oracle.
For example, to switch $\ct^*_\secp$, the reduction algorithm attacking $\SKFEsbSKL$ needs to declare the availability bound $2^\secp$, under which $\qKG$ might not run in polynomial time.
Note that $\Setup$ and $\Enc$ run in polynomial time even for such an availability bound due to the weak optimal efficiency.
Thus, we formalize the security notion of SKFE-sbSKL as follows.

\begin{definition}[Selective Strtong Lessor Security]\label{def:sel_lessor_SKFEsbSKL}
We define selective strong lessor security for SKFE-sbSKL in the same way as that for SKFE-SKL defined in \cref{def:sel_lessor_SKFESKL} , except the following changes for the security experiment.
\begin{itemize}
\item $\qA$ outputs $\numct$ at the beginning, and the challenger generates $\msk\la\Setup(1^\secp,1^\numkey,\numct)$. If $\qA$ makes a query to $\Oracle{\qKG}$ or $\Oracle{\Vrfy}$, $\qA$ is required to output $\numct$ such that $\qKG$ and $\Vrfy$ run in polynomial time.
\item $\Oracle{\qKG}$ does not take $1^{\numct}$ as an input.
\end{itemize}

\end{definition}

\begin{remark}[Insufficiency of existing bounded collusion techniques]\label{remark:GVW_not_work_for_dynamic}
In \cref{sec:tech_overview}, we stated that it is not clear how to use the existing bounded collusion techniques~\cite{C:GorVaiWee12,TCC:AnaVai19} for constructing SKFE-sbSKL.
We provide a more detailed discussion on this point.

The bounded collusion technique essentially enables us to increase the number of decryption keys that an adversary can obtain.
Thus, to try to use the bounded collusion technique in our context, imagine the following naive construction using standard SKFE $\SKFE$ and SKE with certified deletion $\CDSKE$. This construction is a flipped version of the naive construction provided in~\cref{sec:tech_overview}.
In the construction, we encrypt a ciphertext of $\SKFE$ by $\CDSKE$, and we include the key of $\CDSKE$ into the decryption key of the resulting scheme.
The construction can be seen as an SKFE scheme with certified deletion (for ciphertexts) that is secure if an adversary deletes the challenge ciphertext before seeing any decryption key.
The roles of ciphertexts and decryption keys are almost symmetric in SKFE~\cite{JC:BraSeg18}.
Thus, if we can amplify the security of this construction so that it is secure if an adversary sees some decryption keys before deleting the challenge ciphertext, it would lead to SKFE-sbSKL.
The question is whether we can perform such an amplification using the existing bounded collusion techniques~\cite{C:GorVaiWee12,TCC:AnaVai19}.
We observe that it is highly non-trivial to adapt the existing bounded collusion technique starting from ``$0$-bounded'' security.
Especially, it seems difficult to design such a transformation so that the resulting SKFE-sbSKL obtained by flipping the roles of ciphertexts and decryption keys satisfies weak optimal efficiency and security against unbounded number of encryption queries such as \cref{def:sel_lessor_SKFEsbSKL}.

We develop a different technique due to the above reason.
Namely, we reduce the task of amplifying the availability bound of SKFE-sbSKL into the task of amplifying the security bound of it.
In fact, our work implicitly shows that security bound amplification for FE can be used to achieve bounded collusion-resistance.
We see that we can construct bounded collusion secure FE from single-key FE by our parallelizing then security bound amplification technique.
\end{remark}

\subsection{Index-Based SKFE with Static-Bound Secure Key Leasing}

We define index-based SKFE-sbSKL.
Similarly to SKFE-sbSKL, it needs to satisfy weak optimal efficiency and (selective) strong lessor security.

\begin{definition}[Index-Base SKFE with Static-Bound Secure Key Leasing]
An index-base SKFE-sbSKL scheme $\iSKFESKL$ is a tuple of six algorithms $(\Setup, \qKG, \iEnc,\allowbreak \qDec,\qcert,\certvrfy)$. The only difference from an SKFE-sbSKL scheme is that the encryption algorithm $\iEnc$ additionally takes as input an index $\index\in[\numct]$.
\begin{description}
\item[Decryption correctness:]For every $x \in \cX$, $f \in \calF$, $\numkey,\numct\in\bbN$, and $\index\in[\numct]$, we have
\begin{align}
\Pr\left[
\qDec(\qfsk, \ct) \allowbreak = f(x)
\ \middle |
\begin{array}{ll}
\msk \la \Setup(1^\lambda,1^{\numkey},\numct)\\
(\qfsk,\vk)\gets\qKG(\msk,f)\\
\ct\gets\Enc(\msk,\index,x)
\end{array}
\right] 
=1-\negl(\secp).
\end{align}

\item[Verification correctness:]For every $f \in \calF$ and $\numkey,\numct\in\bbN$, we have 
\begin{align}
\Pr\left[
\Vrfy(\vk,\cert)=\top
\ \middle |
\begin{array}{ll}
\msk \la \Setup(1^\lambda,1^{\numkey},\numct)\\
(\qfsk,\vk)\gets\qKG(\msk,f)\\
\cert \lrun \qcert(\qfsk)
\end{array}
\right] 
=1-\negl(\secp).
\end{align}

\item[Weak Optimal Efficiency:] We require that the running time of $\Setup$ and $\Enc$ is bounded by a fixed polynomial of $\lambda$, $\numkey$, and $\log \numct$.

\end{description}

\end{definition}

\begin{definition}[Selective Strong Lessor Security]\label{def:sel_lessor_iSKFEsbSKL}
We say that $\iSKFESKL$ is a selectively strong lessor secure index-based SKFE-sbSKL scheme for $\Xs,\Ys$, and $\Fs$, if it satisfies the following requirement, formalized from the experiment $\expc{\qA,\iSKFESKL}{sel}{s}{lessor}(1^\secp,\coin)$ between an adversary $\qA$ and a challenger:
        \begin{enumerate}
            \item At the beginning, $\qA$ sends $(1^{\numkey},\numct,\index^*,x_0^*,x_1^*)$ to the challenger. If $\qA$ makes a query to $\Oracle{\qKG}$ or $\Oracle{\Vrfy}$, $\qA$ is required to output $\numct$ such that $\qKG$ and $\Vrfy$ run in polynomial time. The challenger runs $\msk\gets\Setup(1^\secp,1^\numkey,\numct)$.
            Throughout the experiment, $\qA$ can access the following oracles.
            \begin{description}
            \item[$\Oracle{\Enc}(\index,x)$:] Given $\index$ and $x$, it returns $\Enc(\msk,\index,x)$.
            \item[$\Oracle{\qKG}(f)$:] Given $f$, it generates $(\qfsk,\vk)\la\qKG(\msk,f)$, sends $\qfsk$ to $\qA$, and adds $(f,\vk,\bot)$ to $\List{\qKG}$. $\qA$ can access this oracle at most $\numkey$ times.
            \item[$\Oracle{\Vrfy}(f,\cert)$:] Given $(f,\cert)$, it finds an entry $(f,\vk,M)$ from $\List{\qKG}$. (If there is no such entry, it returns $\bot$.) If $\top=\Vrfy(\vk,\cert)$, it returns $\top$ and updates the entry into $(f,\vk,\top)$. Otherwise, it returns $\bot$.
            \end{description}
           \item When $\qA$ requests the challenge ciphertext, the challenger checks if for any entry $(f,\vk,M)$ in $\List{\qKG}$ such that $f(x_0^*)\ne f(x_1^*)$, it holds that $M=\top$, and $\qA$ does not make a query with $\index^*$ to $\Oracle{\Enc}$ at this point. If so, the challenger generates $\ct^*\la\Enc(\msk,j^*,x_\coin^*)$ and sends $\ct^*$ to $\qA$. Otherwise, the challenger outputs $0$. Hereafter, $\qA$ is not allowed to sends a function $f$ such that $f(x_0^*)\ne f(x_1^*)$ to $\Oracle{\qKG}$.
            \item $\qA$ outputs a guess $\coin^\prime$ for $\coin$. The challenger outputs $\coin'$ as the final output of the experiment.

        \end{enumerate}
        For any QPT $\qA$, it holds that
\ifnum\llncs=0        
\begin{align}
\advc{\iSKFESKL,\qA}{sel}{s}{lessor}(\secp) \seteq \abs{\Pr[\expc{\iSKFESKL,\qA}{sel}{s}{lessor} (1^\secp,0) \out 1] - \Pr[\expc{\iSKFESKL,\qA}{sel}{s}{lessor} (1^\secp,0) \out 1] }\leq \negl(\secp).
\end{align}
\else
\begin{align}
\advc{\iSKFESKL,\qA}{sel}{s}{lessor}(\secp) 
&\seteq \abs{\Pr[\expc{\iSKFESKL,\qA}{sel}{s}{lessor} (1^\secp,0) \out 1] - \Pr[\expc{\iSKFESKL,\qA}{sel}{s}{lessor} (1^\secp,0) \out 1] }\\
&\leq \negl(\secp).
\end{align}
\fi
\end{definition}


\section{Index-Base SKFE with Static-Bound Secure Key Leasing}\label{sec:iSKFESKL}
We construct an index-based SKFE-sbSKL scheme $\iSKFESKL=(\iSetup,\qiKG,\allowbreak\iEnc,\qiDec,\qicert,\iVrfy)$ using the following tools:

\begin{itemize}
\item An SKFE scheme $\SKFE=(\Setup, \KG, \Enc, \Dec)$.
\item An SKE scheme with Certified Deletion $\CDSKE=(\CD.\KG,\CD.\qencrypt,\CD.\qdecrypt,\allowbreak\CD.\qdelete,\CD.\Vrfy)$.
\item A PRF $\prf$.
\end{itemize}

The description of $\iSKFESKL$ is as follows. 
\begin{description}

 \item[$\iSetup(1^\secp,1^{\numkey},\numct)$:] $ $
 \begin{itemize}
 \item Generate $K \gets \bit^\secp$.
 \item Output $\skl.\msk:=(q,n,K)$.
 \end{itemize}
 \item[$\qiKG(\msk,f)$:] $ $
 \begin{itemize}
 \item Parse $(q,n,K)\gets\skl.\msk$.
 \item Compute $\rskfe_\index\|\rcd_\index\gets\prf_K(\index)$, $\msk_\index\gets \Setup(1^\secp,1^{\numkey};\rskfe_\index)$, and $\cd.\sk_\index\gets \CD.\KG(1^\secp;\rcd_\index)$ for every $\index\in[\numct]$.
 \item Generate $\fsk_\index\gets\KG(\msk_\index,f)$ for every $\index\in[\numct]$.
 \item Generate $(\cd.\qct_\index,\vk_\index)\gets\CD.\qencrypt(\cd.\sk_\index,\fsk_\index)$ for every $\index\in[\numct]$.
 \item Output $\skl.\qfsk:=(\cd.\qct_\index)_{\index\in[\numct]}$ and $\vk:=(\vk_\index)_{\index\in[\numct]}$.
 \end{itemize}
 \item[$\iEnc(\skl.\msk,\index,x)$:] $ $
 \begin{itemize}
  \item Parse $(q,n,K)\gets\skl.\msk$.
 \item Compute $\rskfe_\index\|\rcd_\index\gets\prf_K(\index)$, $\msk_\index\gets \Setup(1^\secp,1^{\numkey};\rskfe_\index)$, and $\cd.\sk_\index\gets \CD.\KG(1^\secp;\rcd_\index)$.
 \item Generate $\ct_\index\gets\Enc(\msk_\index,x)$.
 \item Output $\skl.\ct:=(\index,\ct_\index,\cd.\sk_\index)$.
 \end{itemize}
\item[$\qiDec(\skl.\qfsk,\skl.\ct)$:] $ $
\begin{itemize}
\item Parse $(\cd.\qct_\index)_{\index\in[\numct]}\gets\skl.\qfsk$ and $(\index,\ct_\index,\cd.\sk_\index)\gets\skl.\ct$.
\item Compute $\fsk_\index\gets\CD.\qdecrypt(\cd.\sk_\index,\skl.\qfsk_\index)$.
\item Output $y\gets\Dec(\fsk_\index,\ct_\index)$.
\end{itemize}
\item[$\qicert(\skl.\qfsk)$:] $ $
\begin{itemize}
\item Parse $(\cd.\qct_\index)_{\index\in[\numct]}\gets\skl.\qfsk$.
\item Compute $\cert_\index\gets\CD.\qdelete(\cd.\ct_\index)$ for every $\index\in[\numct]$.
\item Output $\cert:=(\cert_\index)_{\index\in[\numct]}$.
\end{itemize}
\item[$\iVrfy(\vk,\cert)$:] $ $
\begin{itemize}
\item Parse $(\vk_\index)_{\index\in[\numct]}\gets\vk$ and $(\cert_\index)_{\index\in[\numct]}\gets\cert$.
\item Output $\top$ if $\top=\CD.\Vrfy(\vk_\index,\cert_\index)$ for every $\index\in[\numct]$, and otherwise $\bot$.
\end{itemize}
\end{description}

It is clear that $\iSKFESKL$ satisfies correctness and weak optimal efficiency.
For security, we have the following theorem.



\begin{theorem}\label{thm:iSKFESKL_lessor}
If $\SKFE$ is selective indistinguishability-secure, $\CDSKE$ is IND-CVA-CD secure,\footnote{See~\cref{def:reusable_sk-vo_certified_del} for the defition of IND-CVA-CD.} and $\prf$ is a secure PRF, then $\iSKFESKL$ satisfies selective strong lessor security.
\end{theorem}



\begin{proof}[Proof of~\cref{thm:iSKFESKL_lessor}]
We define a sequence of hybrid games to prove the theorem.
\begin{description}
\item[$\hybi{0}$:] This is the same as $\expc{\qA,\iSKFESKL}{sel}{s}{lessor}(1^\secp,0)$.
        \begin{enumerate}
        \item At the beginning, $\qA$ sends $(1^{\numkey},\numct,\index^*,x_0^*,x_1^*)$ to the challenger. The challenger generates $K\la\bit^\secp$. Below, we let $\rskfe_\index\|\rcd_\index\gets\prf_K(\index)$, $\msk_\index\gets \Setup(1^\secp,1^{\numkey};\rskfe_\index)$, and $\cd.\sk_\index\gets \CD.\KG(1^\secp;\rcd_\index)$ for every $\index\in[\numct]$.
            Throughout the experiment, $\qA$ can access the following oracles.
            \begin{description}
            \item[$\Oracle{\Enc}(\index,x)$:] Given $\index$ and $x$, it generates $\ct_\index\gets\Enc(\msk_\index,x)$ and returns $\skl.\ct\seteq(\index,\ct_\index,\cd.\sk_\index)$.
                        \item[$\Oracle{\qKG}(f)$:] Given $f$, it does the following.
 \begin{itemize}
 \item Compute $\fsk_\index\gets\KG(\msk_\index,f)$ for every $\index\in[\numct]$.
\item Compute $(\cd.\qct_\index,\vk_\index)\gets\CD.\qencrypt(\cd.\sk_\index,\fsk_\index)$ for every $\index\in[\numct]$.
\item Sets $\skl.\qfsk\seteq(\cd.\qct_\index)_{\index\in[\numct]}$ and $\skl.\vk\seteq(\vk_\index)_{\index\in[\numct]}$.
 \end{itemize}
It sends $\skl.\qfsk$ to $\qA$ and adds $(f,\skl.\vk,\bot)$ to $\List{\qKG}$.
$\qA$ is allowed to make at most $\numkey$ queries to this oracle.
            \item[$\Oracle{\Vrfy}(f,\cert\seteq(\cert_j)_{j\in[\numct]})$:] Given $(f,\cert\seteq(\cert_j)_{j\in[\numct]})$, it finds an entry $(f,\vk,M)$ from $\List{\qKG}$. (If there is no such entry, it returns $\bot$.) If $\top=\Vrfy(\vk_j,\cert_j)$ for every $j\in[\numct]$, it returns $\top$ and updates the entry into $(f,\vk,\top)$. Otherwise, it returns $\bot$.
            \end{description}

 \item When $\qA$ requests the challenge ciphertext, the challenger checks if for any entry $(f,\vk,M)$ in $\List{\qKG}$ such that $f(x_0^*)\ne f(x_1^*)$, it holds that $M=\top$, and $\qA$ does not make a query with $\index^*$ to $\Oracle{\Enc}$ at this point. If so, the challenger generates $\ct^*_{\index^*}\la\Enc(\msk_{\index^*},x_0^*)$ and sends $\skl.\ct^*\seteq(\index^*,\ct^*_{\index^*},\cd.\sk_{\index^*})$ to $\qA$. Otherwise, the challenger outputs $0$. Hereafter, $\qA$ is not allowed to sends a function $f$ such that $f(x_0^*)\ne f(x_1^*)$ to $\Oracle{\qKG}$.

            \item $\qA$ outputs $\coin^\prime$. The challenger outputs $\coin'$ as the final output of the experiment.

        \end{enumerate}
        \end{description}

\begin{description}
\item[$\hybi{1}$:] This is the same as $\hybi{0}$ except that $\rskfe_\index\|\rcd_\index$ is generated as a uniformly random string for every $\index\in[\numct]$.
\end{description}

We have $\abs{\Pr[\hybi{0}=1]-\Pr[\hybi{1}=1]}=\negl(\secp)$ from the security of $\prf$.

\begin{description}
\item[$\hybi{2}$:] This hybrid is the same as $\hybi{1}$ except that when $\qA$ sends $f$ to $\Oracle{\qKG}$, if $f(x_0^*)\ne f(x_1^*)$, the challenger generates $\cd.\qct_{\index^*}$ included in $\skl.\qfsk\seteq(\cd.\qct_\index)_{\index\in[\numct]}$ as $(\cd.\qct_{\index^*},\vk_{\index^*})\la\CD.\qencrypt(\cd.\sk_{\index^*},\mv{0})$.
\end{description}

We can show that $\abs{\Pr[\hybi{1}=1]-\Pr[\hybi{2}=1]}=\negl(\secp)$ from the security of $\CDSKE$ as follows.
We say $\qA$ is valid if when $\qA$ requests the challenge ciphertext, for any entry $(f,\vk,M)$ in $\List{\qKG}$ such that $f(x_0^*)\ne f(x_1^*)$, it holds that $M=\top$, and $\qA$ does not make a query with $\index^*$ to $\Oracle{\Enc}$ at this point.
In the estimation of $\abs{\Pr[\hybi{1}=1]-\Pr[\hybi{2}=1]}$, we have to consider the case where $\qA$ is valid since if $\qA$ is not valid, the output of the experiment is $0$.
In this transition of experiments, we change a plaintext encrypted under $\cd.\sk_{\index^*}$.
If $\qA$ is valid, $\qA$ cannot obtain $\cd.\sk_{\index^*}$ before $\qA$ is given $\skl.\ct^*$, and $\qA$ returns all ciphertexts under $\cd.\sk_{\index^*}$ before it gets $\cd.\sk_{\index^*}$.
Although the reduction does not have $\vk_{\index^\ast}$ here, it can simulate $\Oracle{\Vrfy}$ by using the verification oracle in IND-CVA-CD game.
Then, we see that $\abs{\Pr[\hybi{1}=1]-\Pr[\hybi{2}=1]}=\negl(\secp)$ follows from the security of $\CDSKE$ under the key $\cd.\sk_{\index^*}$.

\begin{description}
\item[$\hybi{3}$:] This hybrid is the same as $\hybi{2}$ except that the challenger generates $\ct^*_{\index^*}$ included in $\skl.\ct^*$ as $\ct^*_{\index^*}\la\Enc(\msk_{\index^*},x_1^*)$.
\end{description}

By the previous transition, in $\hybi{2}$ and $\hybi{3}$, $\qA$ can obtain a decryption key under $\msk_{\index^*}$ for a function $f$ such that $f(x_0^*)=f(x_1^*)$.
Thus, $\abs{\Pr[\hybi{2}=1]-\Pr[\hybi{3}=1]}=\negl(\secp)$ holds from the security of $\SKFE$.

\begin{description}
\item[$\hybi{4}$:] This hybrid is the same as $\hybi{3}$ except that we undo the changes from $\hybi{0}$ to $\hybi{2}$. $\hybi{4}$ is the same as $\expc{\qA,\iSKFESKL}{sel}{s}{lessor}(1^\secp,1)$.
\end{description}

$\abs{\Pr[\hybi{3}=1]-\Pr[\hybi{4}=1]}=\negl(\secp)$ holds from the security of $\prf$ and $\CDSKE$.

From the above discussions, $\iSKFESKL$ satisfies selective lessor security.
\end{proof}


\section{SKFE with Static-Bound Secure Key Leasing}\label{sec:removing_index}
We construct an SKFE-sbSKL scheme $\SKFEsbSKL=(\sbSKL.\Setup,\sbSKL.\qKG,\allowbreak\sbSKL.\Enc,\sbSKL.\qDec,\allowbreak\sbSKL.\qcert,\sbSKL.\Vrfy)$ from the following tools:
\begin{itemize}
\item An index-based SKFE-sbSKL scheme $\iSKFESKL=(\iSetup, \qiKG, \iEnc,\allowbreak \qiDec,\allowbreak\qicert,\iVrfy)$.
\item A set homomorphic secret sharing $\SetHSS=(\SetGen, \InpEncode,\FuncEncode,\allowbreak\Decode)$.
\item An SKE scheme $\SKE=(\E,\D)$.
\end{itemize}
 The description of $\SKFEsbSKL$ is as follows.
\begin{description}

 \item[$\sbSKL.\Setup(1^\secp,1^{\numkey},\numct)$:] $ $
 \begin{itemize}
 \item Generate $\params\seteq (p,\numall,(T_\instance)_{\instance\in[\numset]})\la\SetGen(1^\secp)$.
 \item Generate $\msk_\instance\gets \iSetup(1^\secp,1^{\numkey},N)$ for every $\instance\in[\numset]$, where $N=\numct/p$.
 \item Generate $\Kske\la\bit^\secp$.
 \item Output $\sbskl.\msk\seteq (\params,N,(\msk)_{\instance\in[\numset]},\Kske)$.
 \end{itemize}
 \item[$\sbSKL.\qKG(\sbskl.\msk,f)$:] $ $
 \begin{itemize}
 \item Parse $(\params,N,(\msk)_{\instance\in[\numset]},\Kske)\la\sbskl.\msk$.
 \item Generate $\ctske_{\instance}\la\E(\Kske,\mv{0})$ for every $\instance\in[\numset]$.
 \item Generate $(f_{\instance})_{\instance\in[\numset]}\la\FuncEncode(\params,f)$.
 \item Generate $(\qfsk_{\instance},\vk_{\instance})\gets\qiKG(\msk_\instance,F[f_{\instance},\ctske_{\instance}])$ for every $\instance\in[\numset]$, where the circuit $F$ is described in \cref{fig:F}.
 \item Output $\sbskl.\qfsk\seteq (\qfsk_\instance)_{\instance\in[\numset]}$ and $\sbskl.\vk\seteq (\vk_\instance)_{\instance\in[\numset]}$.
 \end{itemize}
 \item[$\sbSKL.\Enc(\sbskl.\msk,x)$:] $ $
 \begin{itemize}
 \item Parse $(\params,N,(\msk)_{\instance\in[\numset]},\Kske)\la\sbskl.\msk$.
\item Generate $(\share_{\instance})_{\instance\in[\numset]}\la\InpEncode(\params,x)$.
\item Generate $\index_\instance \la [N]$ for every $\instance\in[\numset]$.
\item Generate $\ct_{\instance}\la\iEnc(\msk_\instance,\index_\instance,(\share_{\instance},\mv{0},0))$ for every $\instance\in[\numset]$.
\item Output $\sbskl.\ct\seteq (\ct_{\instance})_{\instance\in[\numset]}$.
 \end{itemize}
\item[$\sbSKL.\qDec(\sbskl.\qfsk,\sbskl.\ct)$:] $ $
\begin{itemize}
\item Parse $(\qfsk_{\instance})_{\instance\in[\numset]}\gets\sbskl.\qfsk$ and $(\ct_{\instance})_{\instance\in[\numset]}\gets\sbskl.\ct$.
\item Compute $y_{\bfinstance}\la\qiDec(\qfsk_{\instance},\ct_{\instance})$ for every $\instance\in[\numset]$.
\item Output $y\la\Decore((y_{\instance})_{\instance\in[\numset]})$.
\end{itemize}
\item[$\sbSKL.\qcert(\sbskl.\qfsk)$:] $ $
\begin{itemize}
\item Parse $(\qfsk_{\instance})_{\instance\in[\numset]}\gets\sbskl.\qfsk$.
\item Compute $\cert_{\instance}\gets\qicert(\qfsk_{\instance})$ for every $\instance\in[\numset]$.
\item Output $\sbskl.\cert\seteq (\cert_{\instance})_{\instance\in[\numset]}$.
\end{itemize}
\item[$\sbSKL.\Vrfy(\sbskl.\vk,\sbskl.\cert)$:] $ $
\begin{itemize}
\item Parse $(\vk_\instance)_{\instance\in[\numset]}\gets\sbskl.\vk$ and $(\cert_\instance)_{\instance\in[\numset]}\gets\sbskl.\cert$.
\item Output $\top$ if $\top=\iVrfy(\vk_{\instance},\cert_{\instance})$ for every $\instance\in[\numset]$, and otherwise $\bot$.
\end{itemize}
\end{description}

\protocol
{Circuit $F[f_\instance,\ctske_\instance](\share_\instance,\Kske,b)$}
{Description of $F[f_\instance,\ctske_\instance](\share_\instance,\Kske,b)$.}
{fig:F}
{
\ifnum\llncs=1
\scriptsize
\else
\fi
\begin{description}
\setlength{\parskip}{0.3mm} 
\setlength{\itemsep}{0.3mm} 
\item[Hardwired:] A function share $f_\instance$ and an $\SKE$'s ciphertext $\ctske_\instance$.
\item[Input:] an input share $\share_\instance$, an $\SKE$'s secret key $\Kske$, and a bit $b$.
\end{description}
\begin{enumerate}
\setlength{\parskip}{0.3mm} 
\setlength{\itemsep}{0.3mm} 
\item If $b=1$, output $\D(\Kske,\ctske_{\instance})$.
\item Otherwise, output $f_{\instance}(\share_{\instance})$.
\end{enumerate}
}

We show the correctness of $\SKFEsbSKL$.
Let $\sbskl.\qfsk\seteq (\qfsk_\instance)_{\instance\in[\numset]}$ be a decryption key for $f$ and let $\sbskl.\ct\seteq (\ct_{\instance})_{\instance\in[\numset]}$ be a ciphertext of $x$.
From the correctness of $\iSKFESKL$, we obtain $f_{\instance}(\share_{\instance})$ by decrypting $\ct_{\instance}$ with $\qfsk_{\instance}$ for every $\instance\in[\numset]$, where $(f_{\instance})_{\instance\in[\numset]}\la\FuncEncode(\params,f)$ and  $(\share_{\instance})_{\instance\in[\numset]}\la\InpEncode(\params,x)$.
Thus, we obtains  $f(x)\la\Decore((f_{\instance}(\share_{\instance}))_{\instance\in[\numset]})$ from the correctness of $\SetHSS$.
It is clear that $\SKFEsbSKL$ also satisfies verification correctness.

Also, the weak optimal efficiency of $\SKFEsbSKL$ easily follows from that of $\iSKFESKL$ since the running time of algorithms of $\SetHSS$ is independent of $\numct$. Note that $\sbSKL.\Enc$ samples indices from $[N]=[n/p]$, but it can be done in time $\log \numct$.

For security, we have the following theorems.


\begin{theorem}\label{thm:removing_index_lessor_SetHSS}
If $\iSKFESKL$ is a selectively strong lessor secure index-based SKFE-sbSKL scheme and $\SetHSS$ is a set homomorphic secret sharing scheme, and $\SKE$ is a CPA secure SKE scheme, then $\SKFEsbSKL$ is selectively strong lessor secure.
\end{theorem}



\begin{proof}[Proof of~\cref{thm:removing_index_lessor_SetHSS}]
We define a sequence of hybrid games to prove the theorem.
\begin{description}
\item[$\hybi{0}$:] This is the same as $\expc{\qA,\SKFEsbSKL}{sel}{s}{lessor}(1^\secp,0)$.
        \begin{enumerate}
            \item At the beginning, $\qA$ sends $(1^{\numkey},\numct,x_0^*,x_1^*)$ to the challenger. The challenger generates $\params\seteq (p,\numall,(T_\instance)_{\instance\in[\numset]})\la\SetGen(1^\secp)$, $\msk_\instance \gets \iSetup(1^\secp,1^{\numkey},N)$ for every $\instance\in[\numset]$, and $\Kske\la\bit^\secp$, where $N=n/p$. 
            Throughout the experiment, $\qA$ can access the following oracles.
            \begin{description}
            \item[$\Oracle{\Enc}(x^k)$:] Given the $k$-th query $x^k$, it returns $\sbskl.\ct^k$ generated as follows.
             \begin{itemize}
\item Generate $(\share^k_{\instance})_{\instance\in[\numset]}\la\InpEncode(\params,x^k)$.
\item Generate $\index^k_\instance \la [N]$ for every $\instance\in[\numset]$.
\item Generate $\ct^k_{\instance}\la\iEnc(\msk,\instance,\index^k_{\instance},(\share^k_{\instance},\mv{0},0))$ for every $\instance\in[\numset]$.
\item Set $\sbskl.\ct^k\seteq (\ct^k_{\instance})_{\instance\in[\numset]}$.
 \end{itemize}
            
            \item[$\Oracle{\qKG}(f)$:] Given $f$, it generates $\sbskl.\qfsk$ and $\sbskl.\vk$ as follows.
             \begin{itemize}
 \item Generate $(f_{\instance})_{\instance\in[\numset]}\la\FuncEncode(\params,f)$.
 \item Generate $\ctske_{\instance}\la\E(\Kske,\mv{0})$ for every $\instance\in[\numset]$.
 \item Generate $(\qfsk_{\instance},\vk_{\instance})\gets\qiKG(\msk_\instance,F[f_{\instance},\ctske_{\instance}])$ for every $\instance\in[\numset]$.
 \item Set $\sbskl.\qfsk\seteq (\qfsk_\instance)_{\instance\in[\numset]}$ and $\sbskl.\vk\seteq (\vk_\instance)_{\instance\in[\numset]}$.
 \end{itemize}
 It sends $\sbskl.\qfsk$ to $\qA$ and adds $(f,\sbskl.\vk,\bot)$ to $\List{\qKG}$.
 \item[$\Oracle{\Vrfy}(f,\cert\seteq (\cert_\instance)_{\instance\in[\numset]})$:] Given $(f,\cert\seteq (\cert_\instance)_{\instance\in[\numset]})$, it finds an entry $(f,\vk,M)$ from $\List{\qKG}$. (If there is no such entry, it returns $\bot$.) If $\top=\Vrfy(\vk_\instance,\cert_\instance)$ for every $\instance\in[\numset]$, it returns $\top$ and updates the entry into $(f,\vk,\top)$. Otherwise, it returns $\bot$.
            \end{description}
             \item When $\qA$ requests the challenge ciphertext, the challenger checks if for any entry $(f,\vk,M)$ in $\List{\qKG}$ such that $f(x_0^*)\ne f(x_1^*)$, it holds that $M=\top$, and the number of queries to $\Oracle{\Enc}$ at this point is less than $\numct$. If so, the challenger sends $\sbskl.\ct^*$ computed as follows to $\qA$.
\begin{itemize}
\item Generate $(\share^*_{\instance})_{\instance\in[\numset]}\la\InpEncode(\params,x_0^*)$.
\item Generate $\index^*_\instance \la [N]$ for every $\instance\in[\numset]$.
\item Generate $\ct^*_{\instance}\la\iEnc(\msk_\instance,(\share^*_{\instance},\mv{0},0))$ for every $\instance\in[\numset]$.
\item Set $\sbskl.\ct^*\seteq (\ct^*_{\instance})_{\instance\in[\numset]}$.
 \end{itemize}
Otherwise, the challenger outputs $0$. Hereafter, $\qA$ is not allowed to sends a function $f$ such that $f(x_0^*)\ne f(x_1^*)$ to $\Oracle{\qKG}$.
\item $\qA$ outputs $\coin^\prime$. The challenger outputs $\coin'$ as the final output of the experiment.

        \end{enumerate}
        \end{description}

Below, we call $\instance\in[\numset]$ \emph{a secure instance index} if $\index_\instance^*\neq \index_\instance^k$ holds for every $k\in[\numct]$.
We also call $\instance\in[\numset]$ \emph{an insecure instance index} if it is not a secure instance index.
Let $\Ssecure\subseteq[\numset]$ be the set of secure instance indices, and $\Sinsecure=S\setminus\Ssecure$.
Since each $\index_\instance^k$ is sampled from $[N]=[n/p]$, for  each $\instance\in[\numset]$, $\instance$ is independently included in $\Sinsecure$ with probability at most $n/N=p$.
Then, from the existence of unmarked element property of $\SetHSS$, without negligible probability, there exists $\element\in[\numall]$ such that $\element\notin\bigcup_{\instance\in \Sinsecure}T_\instance$.
Below, for simplicity, we assume that there always exists at least one such instance index, and we denote it as $\element^*$.

\begin{description}
\item[$\hybi{1}$:] This is the same as $\hybi{0}$ except that we generate $\index_{\instance}^k$ for every $\instance\in[\numset]$ and $k\in[\numct]$ and $\index_\instance^*$ for every $\instance\in[\numset]$ at the beginning of the experiment.
Note that by this change, secure instance indices and $\instance^*$ are determined at the beginning of the experiment.
\end{description}

$\abs{\Pr[\hybi{0}=1]- \Pr[\hybi{1}=1]}= 0$ holds since the change at this step is only conceptual.

\begin{description}
\item[$\hybi{2}$:] This is the same as $\hybi{1}$ except that when $\qA$ makes a query $f$ to $\Oracle{\qKG}$, if $f(x_0^*)=f(x_1^*)$, it generates $\ctske_{\instance}$ as $\ctske_{\instance}\la\E(\Kske,f_{\instance}(\share^*_{\instance}))$ for every $\instance\in\Ssecure$.
\end{description}

$\abs{\Pr[\hybi{1}=1]- \Pr[\hybi{2}=1]}= \negl(\secp)$ holds from the security of $\SKE$.

\begin{description}
\item[$\hybi{3}$:] This is the same as $\hybi{2}$ except that the challenger generates $\ct^*_{\instance}$ as $\ct^*_{\instance}\la\iEnc(\msk_\instance,\index^*_{\instance},(\mv{0},\Kske,1))$ for every $\instance\in\Ssecure$.
\end{description}

$\abs{\Pr[\hybi{2}=1]- \Pr[\hybi{3}=1]}= \negl(\secp)$ holds from the selective lessor security of $\iSKFESKL$.
We provide the proof of it in \cref{prop:SKFEsbSKL_23_setHSS}.

\begin{description}
\item[$\hybi{4}$:] This is the same as $\hybi{3}$ except that the challenger generates $(\share^*_{\instance})_{\instance\in[\numset]}$ as $(\share^*_{\instance})_{\instance\in[\numset]}\la\InpEncode(\params,x_1^*)$.
\end{description}

$\abs{\Pr[\hybi{3}=1]- \Pr[\hybi{4}=1]}= \negl(\secp)$ holds from the selective indistinguishability-security of $\SetHSS$.
We provide the proof of it in \cref{prop:SKFEsbSKL_34_setHSS}.

\begin{description}
\item[$\hybi{5}$:]This is the same as $\hybi{4}$ except that we undo the changes from $\hybi{0}$ to $\hybi{3}$. This is the same experiment as $\expc{\qA,\SKFEsbSKL}{sel}{s}{lessor}(1^\secp,1)$.
\end{description}

$\abs{\Pr[\hybi{4}=1]- \Pr[\hybi{5}=1]}= \negl(\secp)$ holds from the security of $\SKE$ and $\iSKFESKL$.

\begin{proposition}\label{prop:SKFEsbSKL_23_setHSS}
$\abs{\Pr[\hybi{2}=1]- \Pr[\hybi{3}=1]}= \negl(\secp)$ holds if $\iSKFESKL$ is selectively lessor secure.
\end{proposition}

\begin{proof}[Proof of~\cref{prop:SKFEsbSKL_23_setHSS}]
We define intermediate experiments $\hybi{2,\instance'}$ between $\hybi{2}$ and $\hybi{3}$ for $\instance'\in[\numset]$.
\begin{description}
\item[$\hybi{2,\instance'}$:] This is the same as $\hybi{2}$ except that the challenger generates $\ct^*_{\instance}$ as $\ct^*_{\instance}\la\iEnc(\msk_\instance,\index^*_{\instance},(\mv{0},\Kske,1))$ for every $\instance$ such that $\instance\in\Ssecure$ and $\instance\leq\instance'$.
\end{description}
Then, we have
\ifnum\llncs=0
\begin{align}
\abs{\Pr[\hybi{2}=1]- \Pr[\hybi{3}=1]}\leq\sum_{\instance'\in\numset}\abs{\Pr[\hybi{2,\instance'-1}=1\land \instance'\in\Ssecure]- \Pr[\hybi{2,\instance}=1\land\instance'\in\Ssecure]},\label{eqn:adv_removing_index}
\end{align}
\else
\begin{align}
&\abs{\Pr[\hybi{2}=1]- \Pr[\hybi{3}=1]}\\
\leq&\sum_{\instance'\in\numset}\abs{\Pr[\hybi{2,\instance'-1}=1\land \instance'\in\Ssecure]- \Pr[\hybi{2,\instance}=1\land\instance'\in\Ssecure]},\label{eqn:adv_removing_index}
\end{align}
\fi
where we define $\hybi{2,0}=\hybi{2}$ and $\hybi{2,\numset}=\hybi{3}$.
To estimate each term of \cref{eqn:adv_removing_index}, we construct the following adversary $\qB$ that attacks selective lessor security of $\iSKFESKL$.

\newcommand{\Ssecurele}{S_{\mathtt{secure},< \instance'}}

\begin{enumerate}
         \item $\qB$ executes $\qA$ and obtains $(1^{\numkey},\numct,x_0^*,x_1^*)$. $\qB$ generates $\params\seteq (p,\numall,(T_\instance)_{\instance\in[\numset]})\la\SetGen(1^\secp)$. $\qB$ generates $\index_\instance^k\la[N]$ for every $\instance\in[\numset]$ and $k\in[\numct]$ and $\index_{\instance}^*\la[N]$ for every $\instance\in[\numset]$, and identifies $\Ssecure$ and $\Sinsecure$, where $N=n/p$. If $\instance'\notin\Ssecure$, $\qB$ aborts with output $0$. Otherwise, $\qB$ behaves as follows. Below, we let $\Ssecurele=\Ssecure\cap[\instance'-1]$.
         $\qB$ computes $(\share_{\instance}^*)_{\instance\in[\numset]}\la\InpEncode(\params,x_0^*)$. $\qB$ also generates $\Kske\la\bit^\secp$. $\qB$ sends $(1^{\numkey},N,\index_{\instance'}^*,\allowbreak(\share^*_{\instance'},\mv{0},0),(\mv{0},\Kske,1))$. $\qB$ also generates $\msk_\instance\la\iSetup(1^\secp,1^\numkey,N)$ for every $\instance\in[\numset]\setminus\{\instance'\}$.
            $\qB$ simulates oracles for $\qA$ as follows.
            \begin{description}
            \item[$\Oracle{\Enc}(x^k)$:] Given the $k$-th query $x^k$, $\qB$ returns $\sbskl.\ct^k$ generated as follows.
             \begin{itemize}
\item Generate $(\share^k_{\instance})_{\instance\in[\numset]}\la\InpEncode(\params,x^k)$.
\item If $k\le \numct$, use $(\index_\instance^k)_{\instance\in[\numset]}$ generated at the beginning. Otherwise, Generate $\index^k_\instance \la [N]$ for every $\instance\in[\numset]$.
\item Query $(\index^k_{\instance'},(\share^k_{\instance'},\mv{0},0))$ to its encryption oracle and obtain $\ct^k_{\instance'}$.
\item Generate $\ct^k_{\instance}\la\iEnc(\msk_\instance,\index_\instance^k,(\share^k_\instance,\mv{0},0))$ for every $\instance\in[\numset]\setminus\{\instance'\}$.
\item Set $\sbskl.\ct^k\seteq (\ct^k_{\instance})_{\instance\in[\numset]}$.
 \end{itemize}
            
            \item[$\Oracle{\qKG}(f)$:] Given $f$, $\qB$ returns $\sbskl.\qfsk$ computed as follows.
             \begin{itemize}
              \item Generate $(f_{\instance})_{\instance\in[\numset]}\la\FuncEncode(\params,f)$.
 \item Generate $\ctske_{\instance}\la\E(\Kske,\mv{0})$ for every $\instance\in\Sinsecure$. Generate also $\ctske_{\instance}\la\E(\Kske,f_{\instance}(\share_{\instance}^*))$ for every $\instance\in\Ssecure$ if $f(x_0^*)=f(x_1^*)$, and otherwise generate $\ctske_{\instance}\la\E(\Kske,\mv{0})$ for every $\instance\in\Ssecure$.
 \item Query $F[f_{\instance'},\ctske_{\instance'}]$ to its key generation oracle and obtain $(\qfsk_{\instance'},\vk_{\instance'})$.
 \item Generate $(\qfsk_{\instance},\vk_{\instance})\la\qiKG(\msk_\instance,F[f_\instance,\ctske_\instance])$ for every $\instance\in[\numset]\setminus\{\instance'\}$.
 \item Set $\sbskl.\qfsk\seteq (\qfsk_\instance)_{\instance\in[\numset]}$.
 \end{itemize}
 Also, $\qB$ adds $(f,(\vk_\instance)_{i\in[\numset]\setminus\{\instance^\prime\}},\bot)$ to $\List{\qKG}$.
  \item[$\Oracle{\Vrfy}(f,\cert\seteq (\cert_\instance)_{\instance\in[\numset]})$:] Given $(f,\cert\seteq (\cert_\instance)_{\instance\in[\numset]})$, it finds an entry $(f,(\vk_\instance)_{i\in[\numset]\setminus\{\instance^\prime\}},\bot)$ from $\List{\qKG}$. (If there is no such entry, it returns $\bot$.) $\qB$ sends $(f,\cert_{\instance^\prime})$ to its verification oracle and obtains $M_{\instance^\prime}$.
If $M=\top$ and $\top=\Vrfy(\vk_\instance,\cert_\instance)$ for every $\instance\in[\numset]\setminus\{\instance^\prime\}$, $\qB$ returns $\top$ and updates the entry into $(f,(\vk_\instance)_{i\in[\numset]\setminus\{\instance^\prime\}},\top)$. Otherwise, $\qB$ returns $\bot$.
            \end{description}
            
            \item When $\qA$ requests the challenge ciphertext, $\qB$ checks if for any entry $(f,(\vk_\instance)_{i\in[\numset]\setminus\{\instance^\prime\}},M)$ in $\List{\qKG}$ such that $f(x_0^*)\ne f(x_1^*)$, it holds that $M=\top$. If so, $\qB$ requests the challenge ciphertext to its challenger and obtains $\ct_{\instance^\prime}^*$. $\qB$ also generates $\ct^*_{\instance}\la\iEnc(\msk_\instance,\index_\instance^*,(\mv{0},K,1))$ for every $\instance\in\Ssecurele$ and $\ct^*_{\instance}\la\iEnc(\msk_\instance,\index_\instance^*,(\share^*_\instance,\mv{0},0))$ for every $\instance\in[\numset]\setminus(\Ssecurele\cup\{\instance'\})$. $\qB$ sends $\sbskl.\ct\seteq (\ct^*_{\instance})_{\instance\in[\numset]}$ to $\qA$.
Hereafter, $\qB$ rejects $\qA$'s query $f$ to $\Oracle{\qKG}$ such that $f(x_0^*)\ne f(x_1^*)$.

            \item When $\qA$ outputs $\coin^\prime$, $\qB$ outputs $\coin'$.

        \end{enumerate}

$\qB$ simulates $\hybi{2,\instance'-1}$ (resp. $\hybi{2,\instance'}$) if $\qB$ runs in $\expc{\qB,\SKFEsbSKL}{sel}{s}{lessor}(1^\secp,0)$ (resp. $\expc{\qB,\SKFEsbSKL}{sel}{s}{lessor}(1^\secp,1)$) and $\instance'\in\Ssecure$.
This completes the proof.
\end{proof}

\begin{proposition}\label{prop:SKFEsbSKL_34_setHSS}
$\abs{\Pr[\hybi{3}=1]- \Pr[\hybi{4}=1]}= \negl(\secp)$ holds if $\SetHSS$ is a set homomorphic secret sharing.
\end{proposition}

\begin{proof}[Proof of~\cref{prop:SKFEsbSKL_34_setHSS}]
We construct the following adversary $\qB$ that attacks the selective indistinguishability-security of $\SetHSS$.

\begin{enumerate}
         \item Given $\params\seteq (p,\numall,(T_\instance)_{\instance\in[\numset]})$, $\qB$ executes $\qA$ and obtains $(1^{\numkey},\numct,x_0^*,x_1^*)$. $\qB$ generates $\index_\instance^k\la[N]$ for every $\instance\in[\numset]$ and $k\in[\numct]$ and $\index_{\instance}^*\la[N]$ for every $\instance\in[\numset]$, and identifies $\Ssecure$, $\Sinsecure$, and the unmarked element $\element^*$, where $N=n/p$.
         $\qB$ sends $(\element^*,x_0^*,x_1^*)$ to the challenger and obtains $(\share^*_{\instance})_{\instance\in[\numset]_{ \element^* \notin}}$, where $[\numset]_{\element^* \notin}$ denotes the subset of $[\numset]$ consisting of $\instance$ such that $\element^* \notin T_\instance$.
         $\qB$ also generates $\msk_\instance\la\iSetup(1^\secp,1^\numkey,N)$ for every $\instance\in[\numset]$ and $\Kske\la\bit^\secp$. 
            $\qB$ simulates oracles for $\qA$ as follows.
            \begin{description}
            \item[$\Oracle{\Enc}(x^k)$:] Given the $k$-th query $x^k$, $\qB$ returns $\sbskl.\ct^k$ generated as follows.
             \begin{itemize}
\item Generate $(\share^k_{\instance})_{\instance\in[\numset]}\la\InpEncode(\params,x^k)$.
\item If $k\le \numct$, use $(\index_\instance^k)_{\instance\in[\numset]}$ generated at the beginning. Otherwise, Generate $\index^k_\instance \la [N]$ for every $\instance\in[\numset]$.
\item Generate $\ct^k_{\instance}\la\iEnc(\msk_\instance,\index^k_{\instance},(\share^k_{\instance},\mv{0},0))$ for every $\instance\in[\numset]$.
\item Set $\sbskl.\ct^k\seteq (\ct^k_{\instance})_{\instance\in[\numset]}$.
 \end{itemize}

  \item[$\Oracle{\qKG}(f)$:] Given $f$, $\qB$ returns $\sbskl.\qfsk$ computed as follows.
             \begin{itemize}
               \item Queries $f$ to its function encode oracle and obtain $(f_{\instance},y_{\instance}\seteq f_{\instance}(\share^*_{\instance}))_{\instance\in[\numset]})$ if $f(x_0^*)=f(x_1^*)$. Otherwise, compute  $(f_{\instance})_{\instance\in[\numset]}\la\FuncEncode(\params,f)$.
 \item Generate $\ctske_{\instance}\la\E(\Kske,\mv{0})$ for every $\instance\in\Sinsecure$. Generate also $\ctske_{\instance}\la\E(\Kske,f_{\instance}(\share_{\instance}^*))$ for every $\instance\in\Ssecure$ if $f(x_0^*)=f(x_1^*)$, and otherwise generate $\ctske_{\instance}\la\E(\Kske,\mv{0})$ for every $\instance\in\Ssecure$.
\item Generate $(\qfsk_{\instance},\vk_{\instance})\la\qiKG(\msk_\instance,F[f_{\instance},\ctske_{\instance}])$ for every $\instance\in[\numset]$.
 \item Set $\sbskl.\qfsk\seteq (\qfsk_\instance)_{\instance\in[\numset]}$.
 \end{itemize}
 Also, $\qB$ adds $(f,(\vk_\instance)_{i\in[\numset]},\bot)$ to $\List{\qKG}$.
  \item[$\Oracle{\Vrfy}(f,\cert\seteq (\cert_\instance)_{\instance\in[\numset]})$:] Given $(f,\cert\seteq (\cert_\instance)_{\instance\in[\numset]})$, it finds an entry $(f,(\vk_\instance)_{i\in[\numset]},\bot)$ from $\List{\qKG}$. (If there is no such entry, it returns $\bot$.) If $\top=\Vrfy(\vk_\instance,\cert_\instance)$ for every $\instance\in[\numset]$ and the number of queries to $\Oracle{\Enc}$ at this point is less than $\numct$, $\qB$ returns $\top$ and updates the entry into $(f,(\vk_\instance)_{i\in[\numset]},\top)$. Otherwise, $\qB$ returns $\bot$.
  \end{description}

            \item  When $\qA$ requests the challenge ciphertext, $\qB$ checks if for any entry $(f,(\vk_\instance)_{i\in[\numset]\setminus\{\instance^\prime\}},M)$ in $\List{\qKG}$ such that $f(x_0^*)\ne f(x_1^*)$, it holds that $M=\top$. If so, $\qB$ generates $\ct^*_{\instance}\la\iEnc(\msk_\instance,\index^*_{\instance},(\mv{0},\Kske,1))$ for every $\instance\in\Ssecure$ and $\ct^*_{\instance}\la\iEnc(\msk_\instance,\index^*_{\instance},(\share^*_{\instance},\mv{0},0))$ for every $\instance\in\Sinsecure$, and $\qB$ sends $\sbskl.\ct\seteq (\ct^*_{\instance})_{\instance\in[\numset]}$ to $\qA$.
            Otherwise, $\qB$ outputs $0$ and terminates.
            Hereafter, $\qB$ rejects $\qA$'s query $f$ to $\Oracle{\qKG}$ such that $f(x_0^*)\ne f(x_1^*)$.
            \item When $\qA$ outputs $\coin^\prime$, $\qB$ outputs $\coin'$.

        \end{enumerate}

$\qB$ simulates $\hybi{3}$ (resp. $\hybi{4}$) if $\qB$ runs in $\expb{\SetHSS,\qB}{sel}{ind}(1^\secp,0)$ (resp. $\expb{\SetHSS,\qB}{sel}{ind}(1^\secp,1)$).
This completes the proof.
\end{proof}

From the above discussions, $\SKFEsbSKL$ satisfies selective strong lessor security.
\end{proof}

\begin{remark}[Difference from FE security amplification]\label{remark:diff_from_FE_amplification}
A savvy reader notices that although we use the technique used in the FE security amplification by Jain et al.~\cite{C:JKMS20}, we do not use their probabilistic replacement theorem~\cite[Theorem 7.1 in eprint ver.]{C:JKMS20} and the nested construction~\cite[Section 9 in eprint ver.]{C:JKMS20} in the proofs of~\cref{thm:removing_index_lessor_SetHSS}.
We do not need them for our purpose due to the following reason.

Jain et al. need the nested construction to achieve a secure FE scheme whose adversary's advantage is less than $1/6$ from one whose adversary's advantage is any constant $\epsilon\in (0,1)$.
We do not need the nested construction since we can start with a secure construction whose adversary's advantage is less than $1/6$ by setting a large index space in the index-based construction.

Jain et al. need the probabilistic replacement theorem due to the following reason. We do not know which FE instance is secure at the beginning of the FE security game in the security amplification context, while the adversary in set homomorphic secret sharing must declare the index of a secure instance at the beginning.
In our case, whether each index-based FE instance is secure or not depends on whether randomly sampled indices collide or not.
In addition, we can sample all indices used in the security game at the beginning of the game, and a secure FE instance is fixed at the beginning.
Thus, we can apply the security of set homomorphic secret sharing without the probabilistic replacement theorem.
\end{remark}


\section{SKFE with Secure Key Leasing}\label{sec:dynamic_SKFESKL}
We construct an SKFE-SKL scheme $\SKFESKL=(\SKL.\Setup,\SKL.\qKG,\SKL.\Enc,\allowbreak\SKL.\qDec,\allowbreak\SKL.\qcert,\allowbreak\SKL.\Vrfy)$ from
an SKFE-sbSKL scheme $\SKFEsbSKL\allowbreak=(\sbSKL.\Setup,\allowbreak \sbSKL.\qKG, \sbSKL.\Enc,\allowbreak \sbSKL.\qDec,\allowbreak\sbSKL.\qcert,\sbSKL.\Vrfy)$.
 The description of $\SKFEdbSKL$ is as follows.
\begin{description}

 \item[$\dbSKL.\Setup(1^\secp,1^{\numkey})$:] $ $
 \begin{itemize}
 \item Generate $\msk_k\gets \sbSKL.\Setup(1^\secp,1^{\numkey},2^k)$ for every $k\in[\secp]$.
 \item Output $\dbskl.\msk\seteq (\msk_k)_{k\in[\secp]}$.
 \end{itemize}
 \item[$\dbSKL.\qKG(\dbskl.\msk,f,1^\numct)$:] $ $
 \begin{itemize}
 \item Parse $(\msk_k)_{k\in[\secp]}\gets\dbskl.\msk$.
 \item Compute $k'$ such that $2^{k'-1}\leq \numct \leq 2^{k'}$.
 \item Generate $(\qfsk_{k'},\vk_{k'})\gets\sbSKL.\qKG(\msk_{k'},f)$.
 \item Output $\dbskl.\qfsk\seteq (k',\qfsk_{k'})$ and $\vk_{k'}$.
 \end{itemize}
 \item[$\dbSKL.\Enc(\dbskl.\msk,x)$:] $ $
 \begin{itemize}
 \item Parse $(\msk_k)_{k\in[\secp]}\gets\dbskl.\msk$.
 \item Generate $\ct_{k}\gets\sbSKL.\Enc(\msk_k,x)$ for every $k\in[\secp]$.
 \item Output $\skl.\ct\seteq (\ct_k)_{k\in[\secp]}$.
 \end{itemize}
\item[$\dbSKL.\qDec(\dbskl.\qsk_f,\dbskl.\ct)$:] $ $
\begin{itemize}
\item Parse $(k',\qfsk_{k'})\gets\dbskl.\qfsk$ and $(\ct_k)_{k\in[\secp]}\gets\dbskl.\ct$.
\item Output $y\gets\sbSKL.\qDec(\qfsk_{k'},\ct_{k'})$.
\end{itemize}
\item[$\dbSKL.\qcert(\dbskl.\qsk_f)$:] $ $
\begin{itemize}
\item Parse $(k',\qfsk_{k'})\gets\dbskl.\qsk_f$.
\item Output $\cert\gets\sbSKL.\qcert(\qfsk_{k'})$.
\end{itemize}
\item[$\dbSKL.\Vrfy(\vk,\cert)$:] $ $
\begin{itemize}
\item Output $\top/\bot\gets\sbSKL.\Vrfy(\vk,\cert)$.
\end{itemize}
\end{description}

The correctness of $\SKFEdbSKL$ follows from that of $\SKFEsbSKL$.
Also, we can confirm that all algorithms of $\SKFEdbSKL$ run in polynomial time since $\sbSKL.\Setup$ and $\sbSKL.\Enc$ of $\SKFEsbSKL$ run in polynomial time even for the availability bound $2^\secp$ due to its weak optimal efficiency.
For security, we have the following theorem.

\begin{theorem}\label{thm:SKFEdbSKL_lessor}
If $\SKFEsbSKL$ satisfies selective strong lessor security, then $\SKFEdbSKL$ satisfies selective lessor security.
\end{theorem}


\begin{proof}[Proof of~\cref{thm:SKFEdbSKL_lessor}]
We define a sequence of hybrid games to prove the theorem.
\begin{description}
\item[$\hybi{0}$:] This is the same as $\expb{\qA,\SKFESKL}{sel}{lessor}(1^\secp,0)$.
        \begin{enumerate}
            \item At the beginning, $\qA$ sends $(1^{\numkey},x_0^*,x_1^*)$ to the challenger. The challenger runs $\msk_k \gets \sbSKL.\Setup(1^\secp,1^{\numkey},2^k)$ for every $k\in[\secp]$. 
            Throughout the experiment, $\qA$ can access the following oracles.
            \begin{description}
            \item[$\Oracle{\Enc}(x)$:] Given $x$, it generates $\ct_k\gets\sbSKL.\Enc(\msk_k,x)$ for every $k\in[\secp]$ and returns $\dbskl.\ct\seteq (\ct_k)_{k\in[\secp]}$.
                        \item[$\Oracle{\qKG}(f,1^{\numct})$:] Given $(f,1^{\numct})$, it computes $k$ such that $2^{k-1}\leq \numct\leq 2^{k}$, generates $(\qfsk_{k},\vk_{k})\gets\sbSKL.\qKG(\msk_{k},f)$, and sets $\dbskl.\qfsk\seteq (k,\qfsk_{k})$. It returns $\dbskl.\qfsk$ to $\qA$ and adds $(f,1^{\numct},\vk_k,\bot)$ to $\List{\qKG}$. $\qA$ can access this oracle at most $\numkey$ times.
            \item[$\Oracle{\Vrfy}(f,\cert)$:] Given $(f,\cert)$, it finds an entry $(f,1^\numct,\vk,M)$ from $\List{\qKG}$. (If there is no such entry, it returns $\bot$.) If $\top=\Vrfy(\vk,\cert)$ and the number of queries to $\Oracle{\Enc}$ at this point is less than $\numct$, it returns $\top$ and updates the entry into $(f,1^\numct,\vk,\top)$. Otherwise, it returns $\bot$.
            \end{description}
            \item When $\qA$ requests the challenge ciphertext, the challenger checks if for any entry $(f,1^\numct,\vk,M)$ in $\List{\qKG}$ such that $f(x_0^*)\ne f(x_1^*)$, it holds that $M=\top$. If so, the challenger generates $\ct^*_k\la\sbSKL.\Enc(\msk_k,x_0^*)$ for every $k\in[\secp]$ and sends $\skl.\ct^*\seteq (\ct^*_k)_{k\in[\secp]}$ to $\qA$. Otherwise, the challenger outputs $0$. Hereafter, $\qA$ is not allowed to sends a function $f$ such that $f(x_0^*)\ne f(x_1^*)$ to $\Oracle{\qKG}$.
            \item $\qA$ outputs a guess $\coin^\prime$ for $\coin$. The challenger outputs $\coin'$ as the final output of the experiment.

        \end{enumerate}
        \end{description}

We define $\hybi{k'}$ for every $k'\in[\secp]$.

\begin{description}
\item[$\hybi{k'}$:] This hybrid is the same as $\hybi{k'-1}$ except that $\ct_{k'}^*$ is generated as $\ct_{k'}^*\la\Enc(\msk_{k'},x_1^*)$.

 
%
 
\end{description}
$\hybi{\secp}$ is exactly the same experiment as $\expb{\qA,\SKFESKL}{sel}{lessor}(1^\secp,1)$.

For every $k'\in[\secp]$, we let $\SUC_{k'}$ be the event that the output of the experiment $\hybi{k'}$ is $1$.
Then, we have
\begin{align}
\advb{\SKFESKL,\qA}{sel}{lessor}(\secp)=\abs{\Pr[\hybi{0}=1]-\Pr[\hybi{\secp}=1]}
\le
\sum_{k'=1}^\secp\abs{
\Pr[\SUC_{k'-1}]-\Pr[\SUC_{k'}]
}.
\end{align}

\begin{proposition}\label{prop:SKFEdbSKL_k-1k_k=}
It holds that $\abs{\Pr[\hybi{k'-1}=1]- \Pr[\hybi{k'}=1]}= \negl(\secp)$ for every $k'\in[\secp]$ if $\SKFEsbSKL$ is selectively lessor secure.
\end{proposition}

\begin{proof}[Proof of~\cref{prop:SKFEdbSKL_k-1k_k=}]
We construct the following adversary $\qB$ that attacks selective lessor security of $\SKFEsbSKL$ with respect to $\msk_{k'}$.
\begin{enumerate}
            \item $\qB$ executes $\qA$ and obtains $(1^{\numkey},x_0^*,x_1^*)$ from $\qA$. $\qB$ sends $(1^\numkey,x_0^*,x_1^*,2^{k'})$ to the challenger. 
            $\qB$ generates $\msk_k\la\sbSKL.\Setup(1^\secp,1^\numkey,2^k)$ for every $k\in[\secp]\setminus\{k'\}$. 
            $\qB$ simulates queries made by $\qA$ as follows.
            \begin{description}
            \item[$\Oracle{\Enc}(x)$:] Given $x$, $\qB$ generates $\ct_k\gets\sbSKL.\Enc(\msk_k,x)$ for every $k\in[\secp]\setminus\{k'\}$. $\qB$ also queries $x$ to its encryption oracle and obtains $\ct_{k'}$. $\qB$ returns $\dbskl.\ct\seteq (\ct_k)_{k\in[\secp]}$.
                        \item[$\Oracle{\qKG}(f,1^{\numct})$:] Given $(f,1^{\numct})$, $\qB$ computes $k$ such that $2^{k-1}\leq \numct\leq 2^{k}$. If $k\ne k'$, $\qB$ generates $(\qfsk_{k},\vk_{k})\gets\sbSKL.\qKG(\msk_{k},f)$, and otherwise $\qB$ queries $f$ to its key generation oracle and obtains $\qfsk_{k}$ and sets $\vk_k\seteq \bot$. $\qB$ returns $\dbskl.\qfsk\seteq \qfsk_{k}$. $\qB$ adds $(f,1^\numct,\vk_k,\bot)$ to $\List{\qKG}$.
                        \item[$\Oracle{\Vrfy}(f,\cert)$:] Given $(f,\cert)$, it finds an entry $(f,1^\numct,\vk,M)$ from $\List{\qKG}$. (If there is no such entry, it returns $\bot$.) If $\vk=\bot$, $\qB$ sends $\cert$ to its verification oracle and obtains $M$, and otherwise it computes $M=\Vrfy(\vk,\cert)$.
                        If $M=\top$ and the number of queries to $\Oracle{\Enc}$ at this point is less than $\numct$, it returns $\top$ and updates the entry into $(f,1^\numct,\vk,\top)$. Otherwise, it returns $\bot$.
                         
            \end{description}
            
            \item When $\qA$ requests the challenge ciphertext, the challenger checks if for any entry $(f,1^\numct,\vk,M)$ in $\List{\qKG}$ such that $f(x_0^*)\ne f(x_1^*)$, it holds that $M=\top$. If so, $\qB$ requests the challenge ciphertext to its challenger and obtains $\ct_{k'}^*$, generates $\ct^*_k\la\sbSKL.\Enc(\msk_k,x_1^*)$ for every $1\le k<k'$ and $\ct^*_k\la\sbSKL.\Enc(\msk_k,x_0^*)$ for every $k'<k\le \secp$, and sends $\skl.\ct^*\seteq (\ct^*_k)_{k\in[\secp]}$ to $\qA$.
             Otherwise, the challenger outputs $0$. Hereafter, $\qA$ is not allowed to sends a function $f$ such that $f(x_0^*)\ne f(x_1^*)$ to $\Oracle{\qKG}$.
            
            \item When $\qA$ outputs $\coin^\prime$, $\qB$ outputs $\coin^\prime$ and terminates.

        \end{enumerate}

$\qB$ simulates $\hybi{k'-1}$ (resp. $\hybi{k'}$) for $\qA$ if $\qB$ runs in $\expb{\qB,\SKFEsbSKL}{sel}{lessor}(1^\secp,0)$ (resp. $\expb{\qB,\SKFEsbSKL}{sel}{lessor}(1^\secp,1)$.).
This completes the proof.
\end{proof}

From the above discussions, $\SKFESKL$ satisfies selective lessor security.
\end{proof}

By~\cref{thm:SKFEdbSKL_lessor,thm:removing_index_lessor_SetHSS,thm:iSKFESKL_lessor,thm:prf-owf,thm:skfe_from_owf,thm:ske_cert_del_owf,thm:sethss_owf}, we obtain the following theorem.
\begin{theorem}
If there exist OWFs, there exists selectively lessor secure SKFE-SKL for $\Ppoly$ (in the sense of~\cref{def:sel_lessor_SKFESKL}).
\end{theorem}

Although we describe our results on SKFE-SKL in the bounded collusion-resistant setting, our transformation from standard SKFE to SKFE-SKL also works in the fully collusion-resistant setting. The fully collusion-resistance guarantees that the SKFE scheme is secure even if an adversary accesses the key generation oracle a-priori unbounded times.
Namely, if we start with fully collusion-resistant SKFE, we can obtain fully collusion-resistant SKFE-SKL by our transformations.

\ifnum\llncs=1
\else

\section{Single-Decryptor Functional Encryption}\label{sec:single_dec_FE}
This section introduces single-decryptor FE (SDFE), whose functional decryption keys are copy-protected.

\subsection{Preliminaries for SDFE}

\paragraph{Quantum information.}

We review some basics of qunatum information in this subsection.

Let $\cH$ be a finite-dimensional complex Hilbert space. A (pure) quantum state is a vector $\ket{\psi}\in \cH$.
Let $\cS(\cH)$ be the space of Hermitian operators on $\cH$. A density matrix is a Hermitian operator $\qstate{X} \in \cS(\cH)$ with $\Trace(\qstate{X})=1$, which is a probabilistic mixture of pure states.
A quantum state over $\cH=\bbC^2$ is called qubit, which can be represented by the linear combination of the standard basis $\setbk{\ket{0},\ket{1}}$. More generally, a quantum system over $(\bbC^2)^{\tensor n}$ is called an $n$-qubit quantum system for $n \in \bbN \setminus \setbk{0}$.

A Hilbert space is divided into registers $\cH= \cH^{\qreg{R}_1} \tensor \cH^{\qreg{R}_2} \tensor \cdots \tensor \cH^{\qreg{R}_n}$.
We sometimes write $\qstate{X}^{\qreg{R}_i}$ to emphasize that the operator $\qstate{X}$ acts on register $\cH^{\qreg{R}_i}$.\footnote{The superscript parts are gray colored.}
When we apply $\qstate{X}^{\qreg{R}_1}$ to registers $\cH^{\qreg{R}_1}$ and $\cH^{\qreg{R}_2}$, $\qstate{X}^{\qreg{R}_1}$ is identified with $\qstate{X}^{\qreg{R}_1} \tensor \mat{I}^{\qreg{R}_2}$.

A unitary operation is represented by a complex matrix $\mat{U}$ such that $\mat{U}\mat{U}^\dagger = \mat{I}$. The operation $\mat{U}$ transforms $\ket{\psi}$ and $\qstate{X}$ into $\mat{U}\ket{\psi}$ and $\mat{U}\qstate{X}\mat{U}^\dagger$, respectively.
A projector $\mat{P}$ is a Hermitian operator ($\mat{P}^\dagger =\mat{P}$) such that $\mat{P}^2 = \mat{P}$.

For a quantum state $\qstate{X}$ over two registers $\cH^{\qreg{R}_1}$ and $\cH^{\qreg{R}_2}$, we denote the state in $\cH^{\qreg{R}_1}$ as $\qstate{X}[\qreg{R}_1]$, where $\qstate{X}[\qreg{R}_1]= \Trace_2[\qstate{X}]$ is a partial trace of $\qstate{X}$ (trace out $\qreg{R}_2$).


\begin{definition}[Quantum Program with Classical Inputs and Outputs~\cite{C:ALLZZ21}]\label{def:Q_program_C_IO}
A quantum program with classical inputs is a pair of quantum state $\qstateq$ and unitaries $\setbk{\mat{U}_x}_{x\in[N]}$ where $[N]$ is the domain, such that the state of the program evaluated on input $x$ is equal to $\mat{U}_x \qstateq \mat{U}_x^\dagger$. We measure the first register of $\mat{U}_x \qstateq \mat{U}_x^\dagger$ to obtain an output. We say that $\setbk{\mat{U}_x}_{x\in[N]}$ has a compact classical description $\mat{U}$ when applying $\mat{U}_x$ can be efficiently computed given $\mat{U}$ and $x$.
\end{definition}

\begin{definition}[Positive Operator-Valued Measure]\label{def:POVM}
Let $\cI$ be a finite index set. A positive operator valued measure (POVM) $\cM$ is a collection $\setbk{\mat{M}_i}_{i\in\cI}$ of Hermitian positive semi-define matrices $\mat{M}_i$ such that $\sum_{i\in \cI}\mat{M}_i = \mat{I}$. When we apply POVM $\cM$ to a quantum state $\qstate{X}$, the measurement outcome is $i$ with probability $p_i =\Trace(\qstate{X}\mat{M}_i)$.
We denote by $\cM(\ket{\psi})$ the distribution obtained by applying $\cM$ to $\ket{\psi}$.
\end{definition}

\begin{definition}[Quantum Measurement]\label{def:quantum_measurement}
A quantum measurement $\cE$ is a collection $\setbk{\mat{E}_i}_{i\in\cI}$ of matrices $\mat{E}_i$ such that $\sum_{i\in\cI}\mat{E}_i^\dagger \mat{E}_i=\mat{I}$.
When we apply $\cE$ to a quantum state $\qstate{X}$, the measurement outcome is $i$ with probability $p_i =\Trace(\qstate{X}\mat{E}_i^\dagger \mat{E}_i)$. Conditioned on the outcome being $i$, the post-measurement state is $\mat{E}_i \qstate{X} \mat{E}_i^\dagger/p_i$.
\end{definition}
We can construct a POVM $\cM$ from any quantum measurement $\cE$ by setting $\mat{M}_i \seteq \mat{E}_i^\dagger \mat{E}_i$.
We say that $\cE$ is an implementation of $\cM$. The implementation of a POVM may not be unique.

\begin{definition}[Projective Measurement/POVM]\label{def:projective_measurement}
A quantum measurement $\cE=\setbk{\mat{E}_i}_{i\in \cI}$ is projective if for all $i \in \cI$, $\mat{E}_i$ is a projector.
This implies that $\mat{E}_i\mat{E}_j = \mat{0}$ for distinct $i,j\in \cI$.
In particular, two-outcome projective measurement is called a binary projective measurement, and is written as $\cE=(\mat{P},\mat{I}-\mat{P})$, where $\mat{P}$ is associated with the outcome $1$, and $\mat{I}-\mat{P}$ with the outcome $0$.
Similarly, a POVM $\cM$ is projective if for all $i\in \cI$, $\mat{M}_i$ is a projector. This also implies that $\mat{M}_i\mat{M}_j = \mat{0}$ for distinct $i,j\in \cI$.
\end{definition}


\begin{definition}[Projective Implementation]\label{def:projective_implementation}
Let:
\begin{itemize}
 \item $\cD$ be a finite set of distributions over an index set $\cI$.
 \item $\cP=\setbk{\mat{P}_i}_{i\in \cI}$ be a POVM
 \item $\cE = \setbk{\mat{E}_D}_{D\in\cD}$ be a projective measurement with index set $\cD$.
 \end{itemize}
 We consider the following measurement procedure.
 \begin{enumerate}
 \item Measure under the projective measurement $\cE$ and obtain a distribution $D$.
 \item Output a random sample from the distribution $D$.
 \end{enumerate}
 We say $\cE$ is the projective implementation of $\cP$, denoted by $\projimp(\cP)$, if the measurement process above is equivalent to $\cP$.
\end{definition}

\begin{theorem}[{\cite[Lemma 1]{TCC:Zhandry20}}]\label{lem:commutative_projective_implementation}
Any binary outcome POVM $\cP=(\mat{P},\mat{I}-\mat{P})$ has a unique projective implementation $\projimp(\cP)$.
\end{theorem}

\paragraph{Threshold implementation.}
We review the notion of threshold implementation and related notions since we need them for single decryptor (functional) encryption.
This part is mostly taken from the paper by Coladangelo et al.~\cite{C:CLLZ21}.

\begin{definition}[Threshold Implementation~\cite{TCC:Zhandry20,C:ALLZZ21}]\label{def:threshold_implementation}
Let
\begin{itemize}
\item $\cP=(\mat{P},\mat{I}-\mat{P})$ be a binary POVM
\item $\projimp(\cP)$ and $\cE$ be a projective implementation of $\cP$ and the projective measurement in the first step of $\projimp(\cP)$, respectively
\item  $\gamma>0$.
\end{itemize}
A threshold implementation of $\cP$, denoted by $\TI_\gamma(\cP)$, is the following measurement procedure.
\begin{itemize}
\item Apply $\cE$ to a quantum state and obtain $(p,1-p)$ as an outcome.
\item Output $1$ if $p\ge \gamma$, and $0$ otherwise.
\end{itemize}
For any quantum state $\qstateq$, we denote by $\Tr[\TI_{\gamma}(\cP)\qstateq]$ the probability that the threshold implementation applied to $\qstateq$ outputs $1$ as Coladangelo et al. did~\cite{C:CLLZ21}. This means that whenever $\TI_{\gamma}(\cP)$ appears inside a trace $\Tr$, we treat $\TI_{\gamma}(\cP)$ as a projection onto the $1$ outcome.
\end{definition}

\begin{lemma}[\cite{C:ALLZZ21}]
Any binary POVM $\cP=(\mat{P},\mat{I}-\mat{P})$ has a threshold implementation $\TI_\gamma(\cP)$ for any $\gamma$.
\end{lemma}

\begin{definition}[Mixture of Projetive Measurement~\cite{TCC:Zhandry20}]\label{def:mixture_projective_measurement}
Let $D: \cR \ra \cI$ where $\cR$ and $\cI$ are some sets.
Let $\setbk{(\mat{P}_i,\mat{Q}_i)}_{\in \cI}$ be a collection of binary projective measurement.
The mixture of projective measurements associated to $\cR$, $\cI$, $D$, and $\setbk{(\mat{P}_i,\mat{Q}_i)}_{\in \cI}$ is the binary POVM $\cP_D =(\mat{P}_D,\mat{Q}_D)$ defined as follows
\begin{align}
& \mat{P}_D = \sum_{i\in\cI}\Pr[i \chosen D(R)]\mat{P}_i && \mat{Q}_D = \sum_{i\in\cI}\Pr[i \chosen D(R)]\mat{Q}_i,
\end{align}
where $R$ is uniformly distributed in $\cR$.
\end{definition}

\begin{theorem}[\cite{TCC:Zhandry20,C:ALLZZ21}]\label{thm:ind_distribution_TI}
Let
\begin{itemize}
\item $\gamma>0$
\item $\cP$ be a collection of projective measurements indexed by some sets
\item $\qstateq$ be an efficiently constructible mixed state
\item $D_0$ and $D_1$ be two efficienctly samplable and computationally indistinguishable distributions over $\cI$.
\end{itemize}
For any inverse polynomial $\epsilon$, there exists a negligible function $\delta$ such that
\[
\Tr[\TI_{\gamma -\epsilon}(\cP_{D_1})\qstateq] \ge \Tr[\TI_\gamma(\cP_{D_0})\qstateq] - \delta,
\]
where $\cP_{D_\coin}$ is the mixture of projective measurements associated to $\cP$, $D_\coin$, and $\coin \in \bit$.
\end{theorem}

\paragraph{Cryptographic tools.}

Before we introduce definitions, we introduce a convention.
We say that a cryptographic scheme is sub-exponentially secure if there exists some constant $0<\alpha<1$ such that for every QPT $\qA$ the advantage of the security game is bounded by $O(2^{-\secp^\alpha})$.
\begin{definition}[Learning with Errors]\label{def:LWE}
Let $n,m,q \in \N$ be integer functions of the security parameter $\secp$. Let $\chi = \chi(\secp)$ be an error distribution over $\Z$.
The LWE problem $\LWE_{n,m,q,\chi}$ is to distinguish the following two distributions.
\ifnum\llncs=0
\[
D_0 \seteq \setbracket{(\mm{A},\mv{s}^{\intercal}\mm{A}+\mv{e}) \mid \mm{A} \chosen \Zq^{n\times m}, \mv{s}\chosen \Zq^{n}, \mv{e}\chosen \chi^m}  \text{ and } D_1 \seteq \setbracket{(\mm{A},\mv{u}) \mid \mm{A} \chosen \Zq^{n\times m}, \mv{u}\chosen \Zq^m}.
\]
\else
\begin{align}
&D_0 \seteq \setbracket{(\mm{A},\mv{s}^{\intercal}\mm{A}+\mv{e}) \mid \mm{A} \chosen \Zq^{n\times m}, \mv{s}\chosen \Zq^{n}, \mv{e}\chosen \chi^m}  \text{ and }\\
&D_1 \seteq \setbracket{(\mm{A},\mv{u}) \mid \mm{A} \chosen \Zq^{n\times m}, \mv{u}\chosen \Zq^m}.
\end{align}
\fi

When we say we assume the hardness of the LWE problem or the QLWE assumption holds, we assume that for any QPT adversary $\qalgo{A}$, it holds that
\[
\abs{ \Pr[\qA (D_0)\out  1 ] - \Pr[\qA (D_1)\out 1]} \le \negl(\secp).
\]
\end{definition}

\begin{definition}[Puncturable PRF]\label{def:pprf}
A puncturable PRF (PPRF) is a tuple of algorithms $\PuncPRF = (\prfgen, \prf,\Puncture)$ where $\{\prf_{K}: \bin^{\ell_1} \ra \zo{\ell_2} \mid K \in \zo{\secp}\}$ is a PRF family and satisfies the following two conditions. Note that $\ell_1$ and $\ell_2$ are polynomials of $\secp$.
   \begin{description}
       \item[Punctured correctness:] For any polynomial-size set $S \subseteq \zo{\ell_1}$ and any $x\in \zo{\ell_1} \setminus S$, it holds that
       \begin{align}
       \Pr[\prf_{K}(x) = \prf_{K_{\notin S}}(x)  \mid K \gets \prfgen(1^{\secp}),
       K_{\notin S} \gets \Puncture(K,S)]=1.
       \end{align}
       \item[Pseudorandom at punctured point:] For any polynomial-size set $S \subseteq\zo{\ell_1}$
       and any QPT distinguisher $\qA$, it holds that
       \begin{align}
       \vert
       \Pr[\qA(\prf_{K_{\notin S}},\{\prf_{K}(x_i)\}_{x_i\in S}) \out 1] -
       \Pr[\qA(\prf_{K_{\notin S}}, (\cU_{\ell_2})^{\abs{S}}) \out 1]
       \vert \leq \negl(\secp),
       \end{align}
       where $K\gets \prfgen(1^{\secp})$,
       $K_{\notin S} \gets \Puncture(K,S)$ and $\cU_{\ell_2}$ denotes the uniform distribution over $\zo{\ell_2}$.
   \end{description}
   If $S = \setbk{x^\ast}$ (i.e., puncturing a single point), we simply write $\prf_{\ne x^\ast}(\cdot)$ instead of $\prf_{K_{\notin S}}(\cdot)$ and consider $\prf_{\ne x^\ast}$ as a keyed function.
\end{definition}

It is easy to see that the Goldwasser-Goldreich-Micali tree-based construction of PRFs (GGM PRF)~\cite{JACM:GolGolMic86} from \rmOWF yield puncturable PRFs where the size of the punctured key grows polynomially with the size of the set $S$ being punctured~\cite{AC:BonWat13,PKC:BoyGolIva14,CCS:KPTZ13}. Thus, we have:
\begin{theorem}[\cite{JACM:GolGolMic86,AC:BonWat13,PKC:BoyGolIva14,CCS:KPTZ13}]\label{thm:pprf-owf} If OWFs exist, then for any polynomials $\ell_1(\secp)$ and $\ell_2(\secp)$, there exists a PPRF that maps $\ell_1$-bits to $\ell_2$-bits.
\end{theorem}

\begin{definition}[Public-Key Functional Encryption]\label{def:PKFE}
A PKFE scheme $\PKFE$ is a tuple of four PPT algorithms $(\Setup, \KG, \Enc, \Dec)$. 
Below, let $\cX$, $\cY$, and $\cF$ be the plaintext, output, and function spaces of $\PKFE$, respectively.
\begin{description}
\item[$\Setup(1^\secp)\ra(\pk,\msk)$:] The setup algorithm takes a security parameter $1^\secp$ and outputs a public key $\pk$ and master secret key $\msk$.
\item[$\KG(\msk,f)\ra\sk_f$:] The key generation algorithm $\KG$ takes a master secret key $\msk$ and a function $f \in \cF$, and outputs a functional decryption key $\sk_f$.

\item[$\Enc(\pk,x)\ra\ct$:] The encryption algorithm takes a public key $\pk$ and a message $x \in \cX$, and outputs a ciphertext $\ct$.

\item[$\Dec(\sk_f,\ct)\ra y$:] The decryption algorithm takes a functional decryption key $\sk_f$ and a ciphertext $\ct$, and outputs $y \in \{ \bot \} \cup \cY$.

\item[Correctness:] We require 
we have that
\[
\Pr\left[
\Dec(\sk_f, \ct) = f(x)
 \ \middle |
\begin{array}{rl}
 &(\pk,\msk) \la \Setup(1^\secp),\\
 & \sk_f \gets \KG(\msk,f), \\
 &\ct \gets \Enc(\pk,x)
\end{array}
\right]=1 -\negl(\secp).
\]
\end{description}
\end{definition}
Note that $\Setup$ does not take collusion bound $1^q$ unlike SKFE in~\cref{def:SKFE} since we consider only single-key and collusion-resistant PKFE in this work.

\begin{definition}[Adaptive Indistinguishability-Security for PKFE]\label{def:ad_ind_PKFE}
We say that $\PKFE$ is an \emph{adaptively indistinguishability-secure} SDFE scheme for $\Xs,\Ys$, and $\Fs$, if it satisfies the following requirement, formalized from the experiment $\expb{\qA}{ada}{ind}(1^\secp,\coin)$ between an adversary $\qA$ and a challenger:
        \begin{enumerate}
            \item The challenger runs $(\pk,\msk)\gets\Setup(1^\secp)$ and sends $\pk$ to $\qA$.
            \item $\qA$ sends arbitrary key queries. That is, $\qA$ sends function $f_{i}\in\Fs$ to the challenger and the challenger responds with $\sk_{f_i}\gets \KG(\msk,f_i)$ for the $i$-th query $f_{i}$.
            \item At some point, $\qA$ sends $(x_0,x_1)$ to the challenger. If $f_i(x_0)=f_i(x_1)$ for all $i$, the challenger generates a ciphertext $\ct^*\gets\Enc(\pk,x_\coin)$. The challenger sends $\ct^*$ to $\qA$.
            \item Again, $\qA$ can sends function queries $f_i$ such that $f_i(x_0)=f_i(x_1)$.
            \item $\qA$ outputs a guess $\coin^\prime$ for $\coin$.
            \item The experiment outputs $\coin^\prime$.
        \end{enumerate}
        We say that $\PKFE$ is adaptively indistinguishability-secure if, for any QPT $\qA$, it holds that
\begin{align}
\advb{\PKFE,\qA}{ada}{ind}(\secp) \seteq \abs{\Pr[\expb{\PKFE,\qA}{ada}{ind} (1^\secp,0) \out 1] - \Pr[\expb{\PKFE,\qA}{ada}{ind} (1^\secp,1) \out 1] }\leq \negl(\secp).
\end{align}
If $\qA$ can send only one key query during the experiment, we say $\PKFE$ is adaptively single-key indistinguishability-secure.
\end{definition}

\begin{theorem}[\cite{C:GorVaiWee12}]\label{thm:pkfe_from_pke}
If there exists PKE, there exists adaptively single-key indistinguishability-secure PKFE for $\Ppoly$.
\end{theorem}


\begin{definition}[Indistinguishability Obfuscator~\cite{JACM:BGIRSVY12}]\label{def:io}
A PPT algorithm $\iO$ is a secure \rmIO for a classical circuit class $\{\cC_\secp\}_{\secp \in \bbN}$ if it satisfies the following two conditions.

\begin{description}
\item[Functionality:] For any security parameter $\secp \in \bbN$, circuit $C \in \cC_\secp$, and input $x$, we have that
\begin{align}
\Pr[C^\prime(x)=C(x) \mid C^\prime \lrun \iO(C)] = 1\enspace.
\end{align}

\item[Indistinguishability:] For any PPT $\Sampler$ and QPT distinguisher $\qD$, the following holds:

If $\Pr[\forall x\ C_0(x)=C_1(x) \land \abs{C_0}=\abs{C_1} \mid (C_0,C_1,\aux)\gets \Sampler(1^\secp)]> 1 - \negl(\secp)$, then we have
 \begin{align}
\adva{\iO,\qD}{io}(\secp) &\seteq \left|
 \Pr\left[\qD(\iO(C_0),\aux) = 1 \mid (C_0,C_1,\aux)\gets \Sampler(1^\secp) \right] \right.\\
&~~~~~~~ \left. - \Pr\left[\qD(\iO(C_1),\aux)= 1\mid (C_0,C_1,\aux)\gets \Sampler(1^\secp)  \right] \right|
  \leq \negl(\secp).
 \end{align}
\end{description}
\end{definition}

There are a few candidates of secure IO for polynomial-size classical circuits against quantum adversaries~\cite{TCC:BGMZ18,TCC:CHVW19,EC:AgrPel20,TCC:DQVWW21}.

\subsection{Single-Decryptor Encryption}\label{sec:single_dec_PKE}

We review the notion of single-decryptor encryption (SDE)~\cite{EPRINT:GeoZha20,C:CLLZ21} since it is a crucial building block of our SDFE schemes.
We also extend the existing definitions for SDE.
Some definitions are taken from the paper by Coladangelo et al.~\cite{C:CLLZ21}.
\begin{definition}[Single-Decryptor Encryption~\cite{C:CLLZ21}]\label{def:single_dec_PKE}
A single-decryptor encryption scheme $\SDE$ is a tuple of four algorithms $(\Setup, \QKeyGen, \Enc, \qDec)$. 
Below, let $\calM$ be the message space of $\SDE$.
\begin{description}
\item[$\Setup(1^\secp)\ra(\pk,\sk)$:] The setup algorithm takes a security parameter $1^\secp$, and outputs a public key $\pk$ and a secret key $\sk$.
\item[$\QKeyGen(\sk)\ra\qsk$:] The key generation algorithm $\QKeyGen$ takes a secret key $\sk$, and outputs a quantum decryption key $\qsk$.
\item[$\Enc(\pk,m)\ra\ct$:] The encryption algorithm takes a public key $\pk$ and a message $m \in \calM$, and outputs a ciphertext $\ct$.
\item[$\qDec(\qsk,\ct)\ra\tilde{m}$:] The decryption algorithm takes a quantum decryption key $\qsk$ and a ciphertext $\ct$, and outputs a message $\tilde{m} \in \{ \bot \} \cup \calM$.

\item[Correctness:] We require $\qDec(\qsk, \ct) \allowbreak = m$ for every $m \in \calM$, $(\sk,\pk) \la \Setup(1^\secp)$, $\qsk\gets\QKeyGen(\sk)$, and $\ct\gets\Enc(\pk,m)$.
\end{description}
\end{definition}
Note that $\Setup$ and $\Enc$ are classical algorithms.
Although they could be quantum algorithms, we select the definition above since the construction by Coladangelo et al.~\cite{C:CLLZ21} satisfies it.
This property is crucial in our schemes.

\begin{definition}[CPA Security]\label{def:single_dec_PKE_CPA}
An SDE scheme satisfies CPA security if for any (stateful) QPT $\qA$, it holds that
\[
2\abs{\Pr\left[
\qA(\ct_b)=b
 \ \middle |
\begin{array}{rl}
 &(\pk,\sk)\lrun \Setup(1^\secp),\\
 &(m_0,m_1) \lrun \qA(\pk), b \chosen \zo{}, \\
 &\ct_b \lrun \Enc(\pk,m_b)
 \end{array}
\right] -\frac{1}{2}} \le  \negl(\secp).
\]
\end{definition}

\begin{definition}[Quantum Decryptor~\cite{C:CLLZ21}]\label{def:qunatum_decryptor}
Let $\qstate{p}$ and $\mat{U}$ be a quantum state and a general quantum circuit acting on $n+m$ qubits, where $m$ is the number of qubits of $\qstate{p}$ and $n$ is the length of ciphertexts. A quantum decryptor for ciphertexts is a pair $(\qstate{p},\mat{U})$.

When we say that we run the quantum decryptor $(\qstate{p},\mat{U})$ on ciphertext $\ct$, we execute the circuit $\mat{U}$ on inputs $\ket{\ct}$ and $\qstate{p}$.
\end{definition}

Coladangelo et al.~\cite{C:CLLZ21} introduced a few security notions for SDE.
We introduce the stronger security notion among them, $\gamma$-anti-piracy security.

\begin{definition}[Testing a Quantum Decryptor~\cite{C:CLLZ21}]\label{def:test_quantum_decryptor}
Let $\gamma \in [0,1]$. Let $\pk$ and $(m_0,m_1)$ be a public key and a pair of messages, respectively.
A test for a $\gamma$-good quantum decryptor with respect to $\pk$ and $(m_0,m_1)$ is the following procedure.
\begin{itemize}
\item The procedure takes as input a quantum decryptor $\pirateD =(\qstateq,\mat{U})$.
\item Let $\cP = (\mat{P}, \mat{I} - \mat{P})$ be the following mixture of projective measurements acting on some quantum state $\qstateq^\prime$:
\begin{itemize}
\item Sample a uniformly random $\coin\chosen \bit$ and compute $\ct \gets \Enc(\pk,m_\coin)$.
\item Run $m^\prime \gets \pirateD(\ct)$. If $m^\prime = m_\coin$, output $1$, otherwise output $0$.
\end{itemize}
Let $\TI_{1/2+\gamma}(\cP)$ be the threshold implementation of $\cP$ with threshold $\frac{1}{2}+\gamma$.
Apply $\TI_{1/2+\gamma}(\cP)$ to $\qstateq$. If the outcome is $1$, we say that the test passed, otherwise the test failed.
\end{itemize}
\end{definition}


\begin{definition}[Strong Anti-Piracy Security~\cite{C:CLLZ21}]\label{def:strong_anti-piracy_PKE}
Let $\gamma \in [0,1]$.
We consider the strong $\gamma$-anti-piracy game $\expc{\SDE,\qA}{strong}{anti}{piracy}(\secp,\gamma(\secp))$ between the challenger and an adversary $\qA$ below.
\begin{enumerate}
\item The challenger generates $(\pk,\sk)\gets \Setup(1^\secp)$.
\item The challenger generates $\qsk \gets \QKeyGen(\sk)$ and sends $(\pk,\qsk)$ to $\qA$.
\item $\qA$ outputs $(m_0,m_1)$ and two (possibly entangled) quantum decryptors $\pirateD_1 = (\qstateq[\qreg{R}_1],\mat{U}_1)$ and $\pirateD_2 = (\qstateq[\qreg{R}_2],\mat{U}_2)$, where $m_0 \ne m_1$, $\abs{m_0}=\abs{m_1}$, $\qstateq$ is a quantum state over registers $\qreg{R}_1$ and $\qreg{R}_2$, and $\mat{U}_1$ and $\mat{U}_2$ are general quantum circuits.
\item The challenger runs the test for a $\gamma$-good decryptor with respect to $(m_0,m_1)$ on $\pirateD_1$ and $\pirateD_2$.
The challenger outputs $1$ if both tests pass; otherwise outputs $0$.
\end{enumerate}

We say that $\SDE$ is strong $\gamma$-anti-piracy secure if for any QPT adversary $\qA$, it satisfies that
\[
\Pr[\expc{\SDE,\qA}{strong}{anti}{piracy}(\secp,\gamma(\secp))=1]  \le  \negl(\sep).
\]
\end{definition}

\begin{theorem}[\cite{C:CLLZ21,ARXIV:CulVid21}]
Assuming the existence of sub-exponentially secure IO for $\Ppoly$ and OWFs, the hardness of QLWE, there exists an SDE scheme that satisfies strong $\gamma$-anti-piracy security for any inverse polynomial $\gamma$.
\end{theorem}

\paragraph{Extension of SDE.}
We define a more liberal security notion for $\gamma$-anti-piracy security of SDE to use an SDE scheme as a building block of some cryptographic primitive.

We define a slightly stronger version of quantum decryptor than that in~\cref{def:qunatum_decryptor}.
\begin{definition}[Testing a Quantum Distinguisher for Randomized Message]\label{def:test_quantum_decryptor_rand_msg}
Let $\gamma \in [0,1]$. Let $\pk$ and $(m_0,m_1)$ be a public key and a pair of messages, respectively. Let $\cG=\setbk{\cG_\secp}_{\secp \in \N}$ be a function family.
A test for a $\gamma$-good quantum distinguisher with respect to $\pk$ and $(m_0,m_1,g \in \cG_\secp)$ is the following procedure.
\begin{itemize}
\item The procedure takes as input a quantum decryptor $\pirateD =(\qstateq,\mat{U})$.
\item Let $\cP = (\mat{P}, \mat{I} - \mat{P})$ be the following mixture of projective measurements acting on some quantum state $\qstateq^\prime$:
\begin{itemize}
\item Sample a uniformly random $\coin\chosen \bit$ and $r \chosen \cR$, and compute $\ct \gets \Enc(\pk,g(m_\coin;r))$.
\item Run $\coin^\prime \gets \pirateD(\ct)$. If $\coin^\prime = \coin$, output $1$, otherwise output $0$.
\end{itemize}
Let $\TI_{1/2+\gamma}(\cP)$ be the threshold implementation of $\cP$ with threshold $\frac{1}{2}+\gamma$.
Apply $\TI_{1/2+\gamma}(\cP)$ to $\qstateq$. If the outcome is $1$, we say that the test passed, otherwise the test failed.
\end{itemize}
\end{definition}
The definition above is different from~\cref{def:test_quantum_decryptor}. First, we apply a randomized function to $m_\coin$ and encrypt $g(m_\coin;r)$. Second, $\pirateD$ distinguish whether $\ct$ is a ciphertext of $g(m_0;r)$ or $g(m_1;r)$ rather than computing $g(m_\coin;r)$ (or $m_\coin$).
These differences are crucial when using SDE as a building block of some cryptographic primitive.


We finished preparation for defining a new security notion for SDE.
In the security game defined below, the adversary can send a randomized function applied to challenge messages by the challenger.
\begin{definition}[Strong Anti-Piracy Security for Randomized Message]\label{def:strong_anti-piracy_rand_msg_PKE}
Let $\gamma \in [0,1]$.
We consider the strong $\gamma$-anti-piracy with randomized function family $\cG=\setbk{\cG_\secp}_{\secp \in \N}$ game $\expd{\SDE,\qA,\cG}{strong}{anti}{piracy}{rand}(\secp,\gamma(\secp))$ between the challenger and an adversary $\qA$ below.
\begin{enumerate}
\item The challenger generates $(\pk,\sk)\gets \Setup(1^\secp)$.
\item The challenger generates $\qsk \gets \QKeyGen(\sk)$ and sends $(\pk,\qsk)$ to $\qA$.
\item $\qA$ outputs $(m_0,m_1)$, $g\in \cG_\secp$, and two (possibly entangled) quantum decryptors $\pirateD_1 = (\qstateq[\qreg{R}_1],\mat{U}_1)$ and $\pirateD_2 = (\qstateq[\qreg{R}_2],\mat{U}_2)$, where $\abs{m_0}=\abs{m_1}$, $\qstateq$ is a quantum state over registers $\qreg{R}_1$ and $\qreg{R}_2$, and $\mat{U}_1$ and $\mat{U}_2$ are general quantum circuits.
\item The challenger runs the test for a $\gamma$-good distinguisher with respect to $\pk$ and $(m_0,m_1,g)$ on $\pirateD_1$ and $\pirateD_2$.
The challenger outputs $1$ if both tests pass; otherwise outputs $0$.
\end{enumerate}

We say that $\SDE$ is strong $\gamma$-anti-piracy secure with randomized function family $\cG$ if for any QPT adversary $\qA$, it satisfies that
\[
\Pr[\expd{\SDE,\qA,\cG}{strong}{anti}{piracy}{rand}(\secp,\gamma(\secp))=1]  \le  \negl(\sep).
\]
\end{definition}
Note that we remove the restriction $m_0\ne m_1$, which is in~\cref{def:strong_anti-piracy_PKE}, since the adversary has no advantage even if it sets $m_0 =m_1$.

We do not know how to prove that an SDE scheme that satisfies strong $\gamma$-anti-piracy secure also satisfies strong $\gamma$-anti-piracy secure with randomized function in a black-box way.
This is because the adversary $\qA$ does not receive a challenge ciphertext, but quantum decryptor $\pirateD_1$ and $\pirateD_2$ receives it. That is, a reduction that plays the role of $\qA$ does not receive the challenge ciphertext.
\footnote{In the standard PKE setting, it is easy to see that the standard CPA security implies the liberal CPA security variant, where the adversary select not only $(m_0,m_1)$ but also $g\in \cG_\secp$.}

However, it is easy to see that the SDE scheme by Coladangelo et al.~\cite{C:CLLZ21} satisfies strong $\gamma$-anti-piracy secure with randomized function family $\cG$.
Intuitively, selecting $g$ does not give more power to adversaries since they can select \emph{any} $(m_0,m_1)$ in the original strong anti-piracy security.
We can easily obtain the following theorem.
\begin{theorem}[\cite{C:CLLZ21,ARXIV:CulVid21}]\label{thm:SDE_rand_from_IO_LWE}
Assuming the existence of sub-exponentially secure IO for $\Ppoly$ and OWFs, the hardness of QLWE, and $\cG$ is a randomized function family, there exists an SDE scheme that satisfies strong $\gamma$-anti-piracy security with randomize function family $\cG$ for any inverse polynomial $\gamma$.
\end{theorem}

The proof of~\cref{thm:SDE_rand_from_IO_LWE} is almost the same as that of the strong anti-piracy security of the SDE by Coladangelo et al.~\cite[Theorem 6.13 in the eprint ver.]{C:CLLZ21}.
We explain only the differences to avoid replication.
\paragraph{On randomized message.}
To obtain this theorem, we simply replace $m_b$ with $g(m_b;r_b)$ in the proof by Coladangelo et al.~\cite[Section 6.4 in the ePrint ver.]{C:CLLZ21}, where $r_b$ is randomness chosen by the challenger.
The intuition is as follows.
Even if the challenger applies the adversarially chosen randomized function $g$ to $m_0$ or $m_1$, it does not give more power to the adversary for breaking strong $\gamma$-anti-piracy security. This is because the adversary can send \emph{any} two messages $(m_0,m_1)$ as a challenge message pair in the original strong $\gamma$-anti-piracy security game. The proof by Coladangelo et al.~\cite{C:CLLZ21} does not use any specific property of challenge messages.
The one issue is whether $g(m_0;r_0)\ne g(m_1;r_1)$ holds for any $r_0,r_1 \in \cR$ since the original strong $\gamma$-anti-piracy game requires that challenge messages must be different.
However, this restriction is just a convention. If the adversary sets $m_0=m_1$, it just loses how to distinguish ciphertexts.
That is, even if $g(m_0;r_0)= g(m_1;r_1)$ happens, it does not give more power to the adversary, and is not a problem.

\paragraph{On distinguish-based definition.}
We require a pirate to be a distinguisher rather than a decryptor (computing entire $m_\coin$).
In general, security against distinguisher-based pirates is not implied by security against decryptor-based pirates since adversaries do not need to compute the entire $m_\coin$ in the distinguisher-based definition (the requirement on adversaries is milder).

In the original proof by Coladangelo et al.~\cite{C:CLLZ21}, they show that if pirate decryptors compute $m_\coin$, the reduction can distinguish whether a ciphertext is a real compute-and-compare obfuscation or a simulated compute-and-compare obfuscation.\footnote{The SDE scheme by Coladangelo et al.~\cite{C:CLLZ21} is constructed from IO and compute-and-compare obfuscation~\cite{FOCS:WicZir17,FOCS:GoyKopWat17}.} This breaks the unpredictability of the supported distribution in compute-and-compare obfuscation. Thus, pirate decryptors cannot compute $m_\coin$.
Here, even if pirate decryptors only distinguish $m_0$ from $m_1$ (that is, the output of $\pirateD$ is only one-bit information) instead of computing entire $m_\coin$, the reduction can distinguish real or simulation for compute-and-compare obfuscation. This is because the security of compute-and-compare obfuscation is the decision-type. The reduction in the original proof use only information about $\coin$ and \emph{does not use $m_\coin$}. Thus, we can prove strong $\gamma$-anti-piracy security with any randomized function family $\cG$ defined in~\cref{def:strong_anti-piracy_rand_msg_PKE}.

\paragraph{Other security notions for SDE.}
Coladangelo et al. introduced a few variants of anti-piracy security~\cite{C:CLLZ21}.
One is the CPA-style anti-piracy security, and the other is the anti-piracy security with random challenge plaintexts.
Both are weaker than the strong anti-piracy security in the SDE setting.
See~\cref{def:anti-piracy_CPA_PKE} for the CPA-style definition. We omit the anti-piracy security with random challenge plaintexts since we do not use it in this work.

\subsection{Definitions for Single-Decryptor Functional Encryption}\label{sec:def_single_dec_FE}
We define the notion of single-decryptor functional encryption (SDFE).

\begin{definition}[Single-Decryptor Functional Encryption]\label{def:single_dec_FE}
A single-decryptor functional encryption scheme $\SDFE$ is a tuple of five algorithms $(\Setup,\allowbreak\KG, \QKeyGen, \Enc, \qDec)$. 
Below, let $\cX$, $\cY$, and $\cF$ be the plaintext, output, and function spaces of $\SDFE$.
\begin{description}
\item[$\Setup(1^\secp)\ra(\pk,\msk)$:] The setup algorithm takes a security parameter $1^\secp$, and outputs a public key $\pk$ and a master secret key $\msk$.
\item[$\QKeyGen(\msk,f)\ra\qsk_f$:] The key generation algorithm takes a master secret key $\msk$ and a function $f\in\cF$, and outputs a quantum functional decryption key $\qsk_f$.
\item[$\Enc(\pk,x)\ra\ct$:] The encryption algorithm takes a public key $\pk$ and a plaintext $x \in \cX$, and outputs a ciphertext $\ct$.
\item[$\qDec(\qsk_f,\ct)\ra y$:] The decryption algorithm takes a quantum functional decryption key $\qsk_f$ and a ciphertext $\ct$, and outputs $y \in \{ \bot \} \cup \cY$.

\item[Correctness:] For every $\secp\in \bbN$, $x \in \cX$, $f \in \cF$, $(\pk,\msk)$ in the support of $\Setup(1^\secp)$,
we require that
\[\qDec(\QKeyGen(\msk,f), \Enc(\pk,x)) \allowbreak = f(x).\]
\end{description}
\end{definition}

Note that $\Setup$ and $\Enc$ are classical algorithms as~\cref{def:single_dec_PKE}.

\begin{definition}[Adaptive Security for SDFE]\label{def:adaptive_SDFE}
We say that $\SDFE$ is an \emph{adaptively secure} SDFE scheme for $\Xs,\Ys$, and $\Fs$, if it satisfies the following requirement, formalized from the experiment $\expb{\SDFE,\qA}{ada}{ind}(1^\secp,\coin)$ between an adversary $\qA$ and a challenger:
        \begin{enumerate}
            \item The challenger runs $(\pk,\msk)\gets\Setup(1^\secp)$ and sends $\pk$ to $\qA$.
            \item $\qA$ sends arbitrary key queries. That is, $\qA$ sends function $f_{i}\in\Fs$ to the challenger and the challenger responds with $\qsk_{f_{i}}\gets\QKeyGen(\msk,f_i)$ for the $i$-th query $f_{i}$.
            \item At some point, $\qA$ sends $(x_0,x_1)$ to the challenger. If $f_i(x_0)=f_i(x_1)$ for all $i$, the challenger generates a ciphertext $\ct^*\gets\Enc(\pk,x_\coin)$. The challenger sends $\ct^*$ to $\qA$.
            \item Again, $\qA$ can sends function queries $f_i$ such that $f_i(x_0)=f_i(x_1)$.
            \item $\qA$ outputs a guess $\coin^\prime$ for $\coin$.
            \item The experiment outputs $\coin^\prime$.
        \end{enumerate}
        We say that $\SDFE$ is adaptively secure if, for any QPT $\qA$, it holds that
\begin{align}
\advb{\SDFE,\qA}{ada}{ind}(\secp) \seteq \abs{\Pr[\expb{\SDFE,\qA}{ada}{ind} (1^\secp,0) \out 1] - \Pr[\expb{\SDFE,\qA}{ada}{ind} (1^\secp,1) \out 1] }\leq \negl(\secp).
\end{align}
If $\qA$ can send only one key query during the experiment, we say $\SDFE$ is adaptively single-key secure.
\end{definition}

\begin{definition}[Testing a Quantum FE Distinguisher]\label{def:test_quantum_decryptor_FE}
Let $\gamma \in [0,1]$. Let $\pk$, $(x_0,x_1)$, and $f^\ast$ be a public key, a pair of plaintexts, and a function respectively such that $f^\ast(x_0)\ne f^\ast(x_1)$.
A test for a $\gamma$-good quantum FE distinguisher with respect to $\pk$, $(x_0,x_1)$, and $f^\ast$ is the following procedure.
\begin{itemize}
\item The procedure takes as input a quantum FE decryptor $\pirateD =(\qstateq,\mat{U})$.
\item Let $\cP = (\mat{P}, \mat{I} - \mat{P})$ be the following mixture of projective measurements acting on some quantum state $\qstateq^\prime$:
\begin{itemize}
\item Sample a uniformly random $\coin\chosen \bit$ and compute $\ct \gets \Enc(\pk,x_\coin)$.
\item Run $\coin^\prime \gets \pirateD(\ct)$. If $\coin^\prime = \coin$, output $1$; otherwise output $0$.
\end{itemize}
Let $\TI_{1/2+\gamma}(\cP)$ be the threshold implementation of $\cP$ with threshold $\frac{1}{2}+\gamma$.
Apply $\TI_{1/2+\gamma}(\cP)$ to $\qstateq$. If the outcome is $1$, we say that the test passed, otherwise the test failed.
\end{itemize}
\end{definition}
We follow the spirit of~\cref{def:test_quantum_decryptor_rand_msg}. That is, we consider a quantum distinguisher rather than a quantum decryptor that computes $f^\ast(x_\coin)$.


\begin{definition}[Strong Anti-Piracy Security for FE]\label{def:strong_anti-piracy_FE}
Let $\gamma \in [0,1]$.
We consider the strong $\gamma$-anti-piracy game $\expc{\SDFE,\qA}{strong}{anti}{piracy}(\secp,\gamma(\secp))$ between the challenger and an adversary $\qA$ below.
\begin{enumerate}
\item The challenger generates $(\pk,\msk)\gets \Setup(1^\secp)$ and sends $\pk$ to $\qA$.
\item $\qA$ sends key queries $f_i$ to the challenger and receives $\qsk_{f_i} \lrun \QKeyGen(\msk,f_i)$.
\item At some point $\qA$ sends a challenge query $f^*$ to the challenger and receives $\qsk_{f^*} \lrun \QKeyGen(\msk,f^*)$.
\item Again, $\qA$ sends $f_i$ to the challenger and receives $\qsk_{f_i} \lrun \QKeyGen(\msk,f_i)$.
\item $\qA$ outputs $(x_0,x_1)$ and two (possibly entangled) quantum decryptors $\pirateD_1 = (\qstateq[\qreg{R}_1],\mat{U}_1)$ and $\pirateD_2 = (\qstateq[\qreg{R}_2],\mat{U}_2)$, where $\forall i\ f_i(x_0)=f_i(x_1)$, $f^*(x_0)\ne f^*(x_1)$, $\qstateq$ is a quantum state over registers $\qreg{R}_1$ and $\qreg{R}_2$, and $\mat{U}_1$ and $\mat{U}_2$ are general quantum circuits.
\item The challenger runs the test for a $\gamma$-good FE distinguisher with respect to $\pk$, $(x_0,x_1)$, and $f^\ast$ on $\pirateD_1$ and $\pirateD_2$.
The challenger outputs $1$ if both tests pass, otherwise outputs $0$.
\end{enumerate}

We say that $\SDFE$ is strong $\gamma$-anti-piracy secure if for any QPT adversary $\qA$, it satisfies that
\[
\advc{\SDFE,\qA}{strong}{anti}{piracy}(\secp,\gamma(\secp)) \seteq \Pr[\expc{\SDFE,\qA}{strong}{anti}{piracy}(\secp,\gamma(\secp))=1]  \le  \negl(\sep).
\]
If $\qA$ can send only the challenge query $f^\ast$ during the experiment, we say that $\SDFE$ is challenge-only strong $\gamma$-anti-piracy secure.
\end{definition}
If a pirate has $\qsk_{f^\ast}$, it can easily compute $f^\ast(x_\coin)$ and is a good FE distinguisher since $f^\ast(x_0)\ne f^\ast(x_1)$.
The definition says either $\pirateD_1$ or $\pirateD_2$ do not have the power of $\qsk_{f^\ast}$.

\paragraph{Other security notions for SDFE.}
We can define CPA-style anti-piracy security for SDFE as SDE. We provide the CPA-style anti-piracy security for SDFE in~\cref{def:anti-piracy_CPA_FE}.
We can also define anti-piracy security with random challenge plaintexts for SDFE. However, if we allow adversaries to select a constant-valued function as $f^\ast$, the security is trivially broken. It might be plausible to define security with random challenge plaintexts for functions with high min-entropy. It limits supported classes of functions. However, we focus on FE for all circuits in this work. We omit the definition of anti-piracy with random challenge plaintexts.

\paragraph{Implication to FE with secure key leasing.}
Loosely speaking, SDFE implies FE with secure key leasing. However, there are subtle issues for this implication. Concretely, the implication holds if
\begin{itemize}
\item we allow the deletion certificate to be a quantum state, and
\item restrict our attention to the setting where an adversary is given single decryption key that can be used to detect the challenge bit. We consider this setting for SDFE by default as we can see in~\cref{def:strong_anti-piracy_FE}. However, we consider more powerful adversaries who can obtain multiple such decryption keys for secure key leasing as we can see in~\cref{def:sel_lessor_SKFESKL}. (Although we consider PKFE for SDFE and SKFE for secure key leasing in this work, we ignore the difference between public key and secret key for simplicity.)
\end{itemize}

Under these conditions, we can see the implication as follows. We require an adversary to send back the original quantum state encoding the decryption key as the deletion certificate. We can check the validity of this (quantum) certificate by estimating its success probability using the threshold implementation defined in~\cref{def:threshold_implementation}. Then, we see that an adversary who breaks the secure key leasing security of this construction clearly violates the single-decryptor security of the underlying scheme.

We can remove the second condition above if SDFE is also secure against adversaries who can obtain multiple decryption keys that can be used to detect the challenge bit. Such scheme is called collusion-resistant SDFE. Currently, we do not even have a construction of collusion-resistant SDE (that is, single decryptor PKE).

\subsection{Single-Key Secure Single-Decryptor Functional Encryption}\label{sec:const_1key_SDFE}

We use the following tools:
\begin{itemize}
\item A single-decryptor encryption $\SDE=(\SDE.\Setup, \SDE.\QKeyGen, \SDE.\Enc, \SDE.\qDec)$.
\item An adaptively secure single-key PKFE scheme $\PKFE=(\PKFE.\Setup,\PKFE.\KG,\allowbreak\PKFE.\Enc,\PKFE.\Dec)$.
\end{itemize}

The description of $\oneSDFE$ is as follows.
\begin{description}

 \item[$\oneSDFE.\Setup(1^\secp)$:] $ $
 \begin{itemize}
 \item Generate $(\fe.\pk,\fe.\msk)\gets \PKFE.\Setup(1^\secp)$.
 \item Generate $(\sde.\pk,\sde.\dk)\gets \SDE.\Setup(1^\secp)$.
 \item Output $\pk \seteq (\fe.\pk,\sde.\pk)$ and $\msk \seteq (\fe.\msk,\sde.\dk)$.
 \end{itemize}
 \item[$\oneSDFE.\QKeyGen(\msk,f)$:] $ $
 \begin{itemize}
 \item Parse $\msk =(\fe.\msk,\sde.\dk)$.
 \item Generate $\fe.\sk_f \gets \PKFE.\KG(\msk,f)$.
 \item Generate $\sde.\qdk \gets\SDE.\QKeyGen(\sde.\dk)$.
 \item Output $\qsk_f \seteq (\fe.\sk_f,\sde.\qdk)$.
 \end{itemize}
 \item[$\oneSDFE.\Enc(\pk,x)$:] $ $
 \begin{itemize}
 \item Parse $\pk = (\fe.\pk,\sde.\pk)$.
 \item Generate $\fe.\ct_x \gets \PKFE.\Enc(\fe.\pk,x)$.
 \item Generate $\sde.\ct \gets \SDE.\Enc(\sde.\pk,\fe.\ct_x)$.
 \item Output $\ct \seteq \sde.\ct$.
 \end{itemize}
\item[$\oneSDFE.\qDec(\qsk_f,\ct)$:] $ $
\begin{itemize}
\item Parse $\qsk_f =(\fe.\sk_f,\sde.\qdk)$ and $\ct = \sde.\ct$.
\item Compute $\fe.\ct_x^\prime\gets\SDE.\qDec(\sde.\qdk,\sde.\ct)$.
\item Output $y\gets\PKFE.\Dec(\fe.\ct_x^\prime)$.
\end{itemize}
\end{description}

\begin{theorem}\label{thm:1SDFE_from_1FE}
If $\PKFE$ is adaptively single-key indistinguishability-secure, $\oneSDFE$ is adaptively single-key secure.
\end{theorem}
\begin{proof}
Suppose that $\advb{\oneSDFE,\qA}{ada}{ind}(\secp)$ is non-negligible for a contradiction.
We construct a QPT algorithm $\qB$ for the adaptive single-key security of $\PKFE$ by using the adversary $\qA$ of the adaptive single-key security game of $\oneSDFE$ as follows.
\begin{enumerate}
\item $\qB$ receives $\fe.\pk$ from its challenger, generates $(\sde.\pk,\sde.\dk)\gets \SDE.\Setup(1^\secp)$, and sends $\pk \seteq(\fe.\pk,\sde.\pk)$ to $\qA$.
\item When $\qA$ sends a key query $f$, $\qB$ sends it to its challenger, receives $\fe.\sk_{f} \gets \PKFE.\KG(\fe.\msk,f)$, computes $\sde.\qdk \gets \SDE.\QKeyGen(\sde.\dk)$, and passes $\qsk_{f} \seteq (\fe.\sk_{f},\sde.\qdk)$ to $\qA$.
\item When $\qA$ sends a pair $(x_0,x_1)$, $\qB$ passes $(x_0,x_1)$ to its challenger and receives $\fe.\ct^\ast \gets \PKFE.\Enc(\fe.\pk,x_\coin)$. $\qB$ generates $\sde.\ct \gets \SDE.\Enc(\sde.\pk,\fe.\ct^\ast)$ and passes $\sde.\ct$ to $\qA$.
\item If $f(x_0)\ne f(x_1)$, $\qB$ aborts. Otherwise, go to the next step.
\item $\qB$ outputs whatever $\qA$ outputs.
\end{enumerate}
It is easy to see that $\qB$ perfectly simulates $\expb{\oneSDFE,\qA}{ada}{ind}(\secp)$ since $\qA$ must send $(f,x_0,x_1)$ such that $f(x_0)=f(x_1)$ due to the condition of $\expb{\oneSDFE,\qA}{ada}{ind}(\secp)$. Thus, $\qB$ breaks the adaptive single-key security of $\PKFE$ by using the distinguishing power of $\qA$. This is a contradiction and we finished the proof.
\end{proof}

\begin{theorem}\label{thm:1SDFE_anti-piracy_from_SDE}
If $\SDE$ is strong $\gamma/4$-anti-piracy secure for randomized function family, $\oneSDFE$ is challenge-only strong $\gamma$-anti-piracy secure.
\end{theorem}

\begin{proof}
We recall the original game.
\begin{description}
\item[$\hybi{0}$:] This is the same as $\expc{\qA,\oneSDFE}{strong}{anti}{piracy}(\secp,\gamma(\secp))$. The detailed description is as follows.
\begin{enumerate}
\item Compute $(\fe.\pk,\fe.\msk) \gets \PKFE.\Setup(1^\secp)$ and $(\sde.\pk,\sde.\dk)\gets\SDE.\Setup(1^\secp)$. Set $\pk \seteq (\fe.\pk,\sde.\pk)$ and $\msk \seteq (\fe.\msk,\sde.\dk)$.
\item Compute $\fe.\sk_{f^\ast} \gets \PKFE.\KG(\fe.\msk,f^\ast)$ and $\sde.\qdk \gets \SDE.\QKeyGen(\sde.\dk)$, and  send $\qsk_{f^\ast} \seteq (\fe.\sk_{f^\ast},\sde.\qdk)$ to $\qA$.
\item Receive $(x_0,x_1)$, $\pirateD_1 =(\qstateq[\qreg{R}_1],\mat{U}_1)$, and $\pirateD_2 =(\qstateq[\qreg{R}_2],\mat{U}_2)$, where $f^\ast(x_0)\ne f^\ast(x_1)$.
\item For $i\in \setbk{1,2}$, let $\cP_{i,D}$ be the following mixture of projective measurements acting on some quantum state $\qstateq^\prime$:
\begin{itemize}
\item Sample a uniform $\coin \chosen \bit$. Compute $\fe.\ct \gets \PKFE.\Enc(\fe.\pk,x_\coin)$ and $\sde.\ct \gets \SDE.\Enc(\sde.\pk,\fe.\ct)$ and set $\ct \seteq \sde.\ct$.
\item Run the quantum decryptor $(\qstateq^\prime,\mat{U}_i)$ on input $\ct$. If the outcome is 
$\coin$, output $1$. Otherwise, output $0$.
\end{itemize}
Let $D$ be the distribution over pairs $(\coin,\ct)$ defined in the first item above, and let $\cE_i = \setbk{\mat{E}^i_{(\coin,\ct)}}_{\coin,\ct}$ be a collection of projective measurements where $\mat{E}^i_{(\coin,\ct)}$ is the projective measurement described in the second item above.
$\cP_{i,D}$ is the mixture of projective measurements associated to $D$ and $\cE_i$.
\item Run $\TI_{\frac{1}{2}+\gamma}(\cP_{i,D})$ for $i\in \setbk{1,2}$ on quantum decryptor $\pirateD_1=(\qstateq[\qreg{R}_1],\mat{U}_1)$ and $\pirateD_2=(\qstateq[\qreg{R}_2],\mat{U}_2)$.
Output $1$ if both tests pass, otherwise output $0$.
\end{enumerate}
\end{description}

We prove the following.
\begin{lemma}\label{lem:1SDFE_hyb1}
If $\SDE$ is strong $\gamma$-anti-piracy secure with randomized function family $\setbk{\PKFE.\Enc(\fe.\pk,\cdot;\cdot)}_{\secp}$, it holds that $\advc{\SDFE,\qA}{strong}{anti}{piracy}(\secp,\gamma(\secp)) \le \negl(\secp)$.
\end{lemma}
\begin{proof}[Proof of~\cref{lem:1SDFE_hyb1}]
Suppose that $\advc{\SDFE,\qA}{strong}{anti}{piracy}(\secp,\gamma(\secp))$ is non-negligible for a contradiction.
We construct a QPT algorithm $\qB$ for the strong $\gamma$-anti-piracy game with randomized function $\setbk{\PKFE.\Enc(\fe.\pk,\cdot;\cdot)}_{\secp}$ of $\SDE$ by using the adversary $\qA$ of the strong $\gamma$-anti-piracy game of $\SDFE$ as follows.
\begin{enumerate}
\item $\qB$ receives $\sde.\pk$ and $\sde.\qdk$ from its challenger.
\item $\qB$ generates $(\fe.\pk,\fe.\msk) \gets \PKFE.\Setup(1^\secp)$, sets $\pk\seteq (\fe.\pk,\sde.\pk)$, and sends it to $\qA$.
\item When $\qA$ sends the challenge query $f^\ast$, $\qB$ generates $\fe.\sk_{f^\ast} \gets \PKFE.\KG(\fe.\msk,f^\ast)$, sets $\qsk_{f^\ast} \seteq (\fe.\sk_{f^\ast},\sde.\qdk)$, and sends $\qsk_{f^\ast}$ to $\qA$.
\item At some point, $\qB$ receives $(x_0,x_1)$ and two (possibly entangled) quantum decryptors $\pirateD_1 = (\qstateq[\qreg{R}_1],\mat{U}_1)$ and $\pirateD_2 = (\qstateq[\qreg{R}_2],\mat{U}_2)$, where $f^*(x_0)\ne f^*(x_1)$, from $\qA$.
\item $\qB$ sets quantum decryptors $\pirateD_1^\ast$ and $\pirateD_2^\ast$ as follows.
$\qB$ sets $\pirateD_1^\ast \seteq (\qstateq[\qreg{R}_1],\mat{U}_1)$ and $\pirateD_2^\ast \seteq (\qstateq[\qreg{R}_2],\mat{U}_2)$.
$\qB$ sets a randomized function $g(\cdot;\cdot)\seteq \PKFE.\Enc(\fe.\pk,\cdot;\cdot)$ and sends $(x_0,x_1,g))$ and $\pirateD_1^\ast = (\qstateq[\qreg{R}_1],\mat{U}_1)$ and $\pirateD_2^\ast = (\qstateq[\qreg{R}_2],\mat{U}_2)$ to its challenger.
\end{enumerate}
Recall that, for testing quantum decryptors $\pirateD_1$ and $\pirateD_2$ in $\advc{\SDFE,\qA}{strong}{anti}{piracy}(\secp,\gamma(\secp))$, we use the following: For $i\in \setbk{1,2}$, let $\cP_{i,D}$ be the following mixture of projective measurements acting on some quantum state $\qstateq^\prime$:
\begin{itemize}
\item Sample a uniform $\coin \chosen \bit$. Compute $\fe.\ct \gets \PKFE.\Enc(\fe.\pk,x_\coin)$ and $\sde.\ct \gets \SDE.\Enc(\sde.\pk,\fe.\ct)$ and set $\ct \seteq \sde.\ct$.
\item Run the quantum decryptor $(\qstateq^\prime,\mat{U}_i)$ on input $\ct$. If the outcome is 
$\coin$, output $1$. Otherwise, output $0$.
\end{itemize}

The challenger of $\SDE$ runs the test by using the following.
Let $\cP_{i,D^\prime}$ be the following mixture of projective measurements acting on some quantum state $\qstateq_{\SDE}$:
\begin{itemize}
\item Sample a uniformly random $\coin\chosen \bit$ and $r \chosen \cR$, and compute $m_\coin^\prime \seteq \PKFE.\Enc(\fe.\pk,x_\coin;r)$. Note that $\qB$ sets $g(\cdot,\cdot) \seteq \PKFE.\Enc(\fe.\pk,\cdot;\cdot)$. Then, compute $\ct \gets \SDE.\Enc(\sde.\pk,m_\coin^\prime)$.
\item Run $\coin^\prime \gets \pirateD(\ct)$. If $\coin^\prime = \coin$, output $1$, otherwise output $0$.
\end{itemize}
The distribution is the same as $D$, which $\qB$ simulates. 

We assumed that $\advc{\SDFE,\qA}{strong}{anti}{piracy}(\secp,\gamma(\secp))$ is non-negligible at the beginning of this proof.
That is, applying $\TI_{\frac{1}{2}+\gamma}(\cP_{i,D})$ on $\qstateq[\qreg{R}_i]$ results in two outcomes $1$ with non-negligible probability:
\begin{align}
\Tr[\TI_{\frac{1}{2}+\gamma }(\cP_{1,D}) \tensor \TI_{\frac{1}{2}+\gamma }(\cP_{2,D})\qstateq] > \negl(\secp),
\end{align}
where $\mu >\negl(\secp)$ means $\mu$ is non-negligible.
This means that $\qstateq[\qreg{R}_i]$ is a $\gamma$-good distinguisher with respect to ciphertexts generated according to $D$.
It is easy to see that $(\pirateD_1^\ast,\pirateD_2^\ast)$ is a $\gamma$-good distinguisher for randomized function $\PKFE.\Enc(\fe.\pk,\cdot;\cdot)$ since $D^\prime$ and $D$ are the same distribution.
Thus, if $\qA$ outputs $\gamma$-good distinguisher, $\qB$ can also outputs $\gamma$-good distinguisher for randomized function $\PKFE.\Enc(\fe.\pk,\cdot;\cdot)$. This completes the proof of~\cref{lem:1SDFE_hyb1}.
\end{proof}
By~\cref{lem:1SDFE_hyb1}, we complete the proof of~\cref{thm:1SDFE_anti-piracy_from_SDE}.
\end{proof}

We obtain the following corollary from~\cref{thm:SDE_rand_from_IO_LWE,thm:1SDFE_from_1FE,thm:1SDFE_anti-piracy_from_SDE}.

\begin{corollary}\label{cor:strong_1SDFE_from_IO_LWE}
Assuming the existence of sub-exponentially secure IO for $\Ppoly$ and OWFs, and the hardness of QLWE, there exists an SDFE scheme for $\Ppoly$ that satisfies adaptive single-key security and challenge-only strong $\gamma$-anti-piracy security for any inverse polynomial $\gamma$.
\end{corollary}

\begin{remark}[On approximation version of threshold implementation]\label{remark:on_ATI}
We do not need to use approximate version of $\TI$ that runs in polynomial time (denoted by $\ATI$ by previous works~\cite{C:ALLZZ21,C:CLLZ21}) because we do not extract information from pirate decryptors $\pirateD_1$ and $\pirateD_2$.
We use the decision bits output by them for breaking the security of cryptographic primitives in security reductions.
This is the same for the construction in~\cref{sec:const_multi-key_SDFE}.
\end{remark}
\subsection{Collusion-Resistant Single-Decryptor Functional Encryption}\label{sec:const_multi-key_SDFE}
\begin{construction}
We use the following tools:
\begin{itemize}
\item A single-key SDFE $\oneSDFE=(\oneSDFE.\Setup,\oneSDFE.\KG, \oneSDFE.\QKeyGen, \oneSDFE.\Enc,\allowbreak \oneSDFE.\qDec)$.
\item An IO $\iO$.
\item A puncturable PRF $\PuncPRF=(\prfgen,\prf,\Puncture)$, where $\prf: \zo{\secp}\times \zo{\ell} \ra \cR_{\Setup}$ and $\cR_{\Setup}$ is the randomness space of $\oneSDFE.\Setup$.
\item A puncturable PRF $\PuncPRF^\prime=(\prfgen^\prime,\prf^\prime,\Puncture^\prime)$, where $\prf^\prime: \zo{\secp}\times \zo{\ell} \ra \cR_{\Enc}$ and $\cR_{\Enc}$ is the randomness space of $\oneSDFE.\Enc$.
\item A (standard) PRF $\prf^{(1)}: \zo{\secp}\times \cC_{\oneSDFE} \ra \cR_{\prfgen^\prime}$, $\cC_{\oneSDFE}$ is the ciphertext space of $\oneSDFE$ and $\cR_{\prfgen^\prime}$ is the randomness space of $\prfgen^\prime$.
\item A (standard) PRF $\prf^{(2)}: \zo{\secp}\times \cC_{\oneSDFE} \ra \cR_{\iO}$, $\cC_{\oneSDFE}$ is the ciphertext space of $\oneSDFE$ and $\cR_{\iO}$ is the randomness space of $\iO$.
\end{itemize}
Note that we use $\prf^{(1)}$ and $\prf^{(2)}$ only in security proofs.

\begin{description}
\item[$\SDFE.\Setup(1^\secp)$:] $ $
\begin{itemize}
\item Generate $\prfkey \gets \prfgen(1^\secp)$ and $\sfS_{\sfonefe}[\prfkey]$ defined in~\cref{fig:setup-pkfe}.
\item Return $(\pkhat ,\mskhat) \seteq (\iO_1(\sfS_{\sfonefe}),\prfkey)$.
\end{itemize}

\item[$\SDFE.\QKeyGen(\mskhat,f)$:] $ $
\begin{itemize}
\item Parse $\prfkey = \mskhat$ and choose $\gtag \chosen \zo{\ell}$.
\item Compute $r_{\gtag} \gets \prf_{\prfkey}(\gtag)$ and $(\msk_{\gtag}, \pk_{\gtag}) \gets \oneSDFE.\Setup(1^\secp; r_{\gtag})$.
\item Generate $\qsk_{f,\gtag} \gets \oneSDFE.\QKeyGen(\msk_{\gtag},f)$.
\item Output $\qskhat_{f,\gtag} \seteq (\gtag,\qsk_{f,\gtag})$.
\end{itemize}
\item[$\SDFE.\Enc(\pkhat, x)$:] $ $
\begin{itemize}
\item Generate $\prfkey^\prime \gets \prfgen^\prime(1^\secp)$ and $\sfE_\sfonefe[\pkhat,\prfkey^\prime,x]$ defined in~\cref{fig:enc-one-pkfe}.
\item Return $\cthat \gets \iO_2(\sfE_\sfonefe[\pkhat,\prfkey^\prime,x])$.
\end{itemize}

\item[$\SDFE.\qDec(\qskhat_f, \cthat)$:] $ $
\begin{itemize}
\item Parse $(\gtag, \qsk_{f,\gtag}) = \qskhat_f$.
\item Evaluate the circuit $\cthat$ on input $\gtag$, that is $\ct_{\gtag} \gets \cthat(\gtag)$.
\item Return $y \gets \oneSDFE.\qDec(\qsk_{f,\gtag}, \ct_{\gtag})$.
\end{itemize}
\end{description}
 \end{construction}

\protocol
{Setup Circuit $\sfS_{\sfonefe} [\prfkey](\gtag)$}
{Description of $\sfS_{\sfonefe} [\prfkey]$.}
{fig:setup-pkfe}
{
\ifnum\llncs=1
\scriptsize
\else
\fi
\begin{description}
\setlength{\parskip}{0.3mm} 
\setlength{\itemsep}{0.3mm} 
\item[Hardwired:] puncturable PRF key $\prfkey$.
\item[Input:] tag $\gtag \in \zo{\ell}$.
\item[Padding:] circuit is padded to size $\padsize_\sfS \seteq \padsize_{\sfS} (\secp,n,s,\ell)$, which is determined in analysis.
\end{description}
\begin{enumerate}
\setlength{\parskip}{0.3mm} 
\setlength{\itemsep}{0.3mm} 
\item Compute $r_\gtag \gets \prf_{\prfkey}(\gtag)$.
\item Compute $(\pk_{\gtag}, \msk_{\gtag}) \gets \oneSDFE.\Setup(1^\secp;r_{\gtag})$ and output $\pk_{\gtag}$.
\end{enumerate}
}
\ifnum\llncs=1
\vspace{-5mm}
\else
\fi

\protocol
{Encryption Circuit $\sfE_\sfonefe[\pkhat,\prfkey^\prime, x](\gtag)$}
{Description of $\sfE_\sfonefe[\pkhat,\prfkey^\prime, x]$.}
{fig:enc-one-pkfe}
{
\ifnum\llncs=1
\scriptsize
\else
\fi
\begin{description}
\setlength{\parskip}{0.3mm} 
\setlength{\itemsep}{0.3mm} 
\item[Hardwired:] circuit $\pkhat$, puncturable PRF key $\prfkey^\prime$, and message $x$.
\item[Input:] tag $\gtag \in \zo{\ell}$.
\item[Padding:] circuit is padded to size $\padsize_{\sfE} \seteq \padsize_{\sfE}(\secp,n,s,\ell)$, which is determined in analysis.
\end{description}
\begin{enumerate}
\setlength{\parskip}{0.3mm} 
\setlength{\itemsep}{0.3mm} 
\item Evaluate the circuit $\pkhat$ on input $\gtag$, that is $\pk_{\gtag} \gets \pkhat(\gtag)$.
\item Compute $r_{\gtag}^\prime \gets \prf^\prime_{\prfkey^\prime}(\gtag)$ and output $\ct_{\gtag} \gets \oneSDFE.\Enc(\pk_{\gtag},x;r^\prime_{\gtag})$.
\end{enumerate}
}

\begin{theorem}\label{thm:SDFE_from_1FE_IO}
If $\oneSDFE$ is sub-exponentially adaptively single-key secure, $\iO$ is a sub-exponentially secure IO for $\Ppoly$, and $\PuncPRF$ is a sub-exponentially secure PPRF, $\SDFE$ is adaptively secure.
\end{theorem}
\begin{proof}
We define a sequence of hybrid games.
For the hybrid games, we present some definitions.
Let $q$ be the total number of key queries by $\qA$. Note that $q$ could be any polynomial.
When we start the adaptive security game, the adversary $\qA$ is fixed, and $q$ is also fixed. We choose tags $\gtag_1,\ldots,\gtag_q \chosen \zo{\ell}$ for $q$ key queries at the beginning of the game.
We can interpret $\ell$ bit strings as integers and assume that there is no $i,j$ such that $i\ne j$ and $\gtag_i = \gtag_j$ without loss of generality.

\begin{description}
\item[$\hybi{0}(\coin)$:] This is the original adaptive security game.
\begin{enumerate}
\item The challenger generates $(\pkhat,\mskhat)=(\iO_1(\sfS_\sfonefe),\prfkey)\gets \SDFE.\Setup(1^\secp)$ and sends $\pkhat$ to $\qA$.
\item $\qA$ sends key queries $f_k$ to the challenger and the challenger generates $r_{\gtag_k} \gets \prf_\prfkey(\gtag_k)$ and $(\pk_{\gtag_k},\msk_{\gtag_k}) \gets \oneSDFE.\Setup(1^\secp;r_{\gtag_k})$, and returns $\qskhat_{f_k} = (\gtag_k,\qsk_{f_k,\gtag_k})$ where $\qsk_{f_k,\gtag_k}\lrun \oneSDFE.\QKeyGen(\msk_{\gtag_k},f_k)$.
\item At some point $\qA$ sends a pair $(x_0,x_1)$ to the challenger.
If $f_i(x_0)=f_i(x_1)$, the challenger generates $\prfkey^\prime \gets \prfgen^\prime(1^\secp)$ and $\sfE_\sfonefe[\pkhat,\prfkey^\prime,x_{\coin}]$ defined in~\cref{fig:enc-one-pkfe} and sends $\cthat \gets \iO_2(\sfE_\sfonefe[\pkhat,\prfkey^\prime,x_{\coin}])$ to $\qA$.
\item Again, $\qA$ sends $f_k$ to the challenger and the challenger returns $\qskhat_{f_k} \lrun \SDFE.\QKeyGen(\mskhat,f_k)$ if $f_k(x_0)=f_k(x_1)$.
\item When $\qA$ outputs $\coin^\prime$, the game outputs $\coin^\prime$.
\end{enumerate}

Let $i \in [2^\ell]$.
In the following hybrid games, we gradually change $\sfS_\sfonefe$ and $\sfE_\sfonefe$ for the $i$-th tag $\gtag^\prime_i$, which is the $\ell$-bit string representation of $i \in [2^\ell]$.
So, it could happen that $\gtag^\prime_i = \gtag_j$ for some $i\in [2^\ell]$ and $j \in[q]$.
Let $\Adv_{0}(\coin)$ and $\Adv_{x}^{i}$ be the advantage of $\qA$ in $\hybi{0}(\coin)$ and $\hybij{x}{i}$, respectively.

 \item[$\hybij{1}{i}$:] We generate $\pkhat$ as obfuscated $\sfS^*_\sfonefe$ described in~\cref{fig:setup-pkfe-punc}. In this hybrid game, we set $r_{i} \la \prf_{\prfkey}(\gtag^\prime_i)$, $\prf_{\ne \gtag^\prime_i} \la \Puncture(\prfkey,\gtag^\prime_i)$, $(\pk_{\gtag^\prime_i},\msk_{\gtag^\prime_i}) \la \oneSDFE.\Setup(1^\secp ;r_{i})$, and $\pkhat \gets \iO_1(\sfS^\ast_\sfonefe [\gtag^\prime_i,\prf_{\ne \gtag^\prime_i},\pk_{\gtag^\prime_i}])$.

When $i=1$, the behavior of $\sfS^*_\sfonefe$ is the same as that of $\sfS_\sfonefe$ since the hard-wired $\pk_{\gtag^\prime_1}$ in $\sfS^*_\sfonefe$ is the same as the output of $\sfS_\sfonefe$ on the input $\gtag^\prime_1$. Their size is also the same since we pad circuit $\sfS_\sfonefe$ to have the same size as $\sfS_\sfonefe^*$. Then, we can use the indistinguishability of $\iO_1$ and it holds that $\abs{\Adv_{0}(0) - \Adv_{1}^{1}} \le \negl(\secp)$.

\item[$\hybij{2}{i}$:] The challenge ciphertext is generated by obfuscating $\sfE^*_\sfonefe$ described in~\cref{fig:enc-one-pkfe-punc}. In this hybrid game, we set $r^\prime_{i} \la \prf^\prime_{\prfkey^\prime}(\gtag^\prime_i)$, $\prf^\prime_{\ne \gtag^\prime_i} \la \Puncture(\prfkey^\prime,\gtag^\prime_i)$, $\ct_{i} \la \oneSDFE.\Enc(\pk_{\gtag^\prime_i},x_0 ;r'_{i})$, $\pk_{\gtag^\prime_i} \la \pkhat(\gtag^\prime_i)$, and $\cthat \gets \iO_2(\sfE_\sfonefe[\gtag^\prime,\pkhat,\prf^\prime_{\ne \gtag^\prime},\allowbreak x_0,x_1,\ct_{\gtag^\prime}])$.

When $i=1$, the behavior of $\sfE^*_\sfonefe$ is the same as that of $\sfE_\sfonefe$ since the hard-wired $\ct_1$ in $\sfE^*_\sfonefe$ is the same as the output of $\sfE_\sfonefe$ on the input $1$. Both circuits have the same size by padding $\pad_\sfE$.

In addition, for $i \geq 2$, the behavior of $\sfE^*_\sfonefe$ does not change between $\hybij{1}{i}$ and $\hybij{2}{i}$.
Thus, it holds that $\abs{\Adv_{2}^{i}- \Adv_{1}^{i} }\le \negl(\secp)$ for every $i \in [2^\ell]$ due to the indistinguishability of $\iO_2$.

\protocol
{Setup Circuit $\sfS^*_{\sfonefe} [\gtag^\prime,\prf_{\ne \gtag^\prime},\pk_{\gtag^\prime}](\gtag)$}
{Description of $\sfS^*_{\sfonefe} [\gtag^\prime,\prf_{\ne \gtag^\prime},\pk_{\gtag^\prime}]$.
}
{fig:setup-pkfe-punc}
{
\ifnum\llncs=1
\scriptsize
\else
\fi
\begin{description}
\setlength{\parskip}{0.3mm} 
\setlength{\itemsep}{0.3mm} 
\item[Hardwired:] tag $\gtag^\prime$, puncturable PRF key $\prf_{\ne \gtag^\prime}$, and $\oneSDFE$ public-key $\pk_{\gtag^\prime}$.
\item[Input:] tag $\gtag \in \zo{\ell}$.
\item[Padding:] circuit is padded to size $\padsize_{\sfS}\seteq\padsize_{\sfS}(\secp,n,s)$, which is determined in analysis.
\end{description}
\begin{enumerate}
\setlength{\parskip}{0.3mm} 
\setlength{\itemsep}{0.3mm} 
\item If $\gtag^\prime = \gtag$, output $\pk_{\gtag^\prime}$.
\item Else, compute $r \gets \prf_{\ne \gtag^\prime}(\gtag)$.
\item Compute $(\pk_{\gtag}, \msk_{\gtag}) \gets \oneSDFE.\Setup(1^\secp;r)$ and output $\pk_{\gtag}$.
\end{enumerate}
}

\protocol
{Encryption Circuit $\sfE^*_\sfonefe[\gtag^\prime,\pkhat,\prf^\prime_{\ne \gtag^\prime}, x_0,x_1,\ct_{\gtag^\prime}](\gtag)$}
{Description of $\sfE^*_\sfonefe[\gtag^\prime,\pkhat,\prf^\prime_{\ne \gtag^\prime}, x_0,x_1,\ct_{\gtag^\prime}]$.
}
{fig:enc-one-pkfe-punc}
{
\ifnum\llncs=1
\scriptsize
\else
\fi
\begin{description}
\setlength{\parskip}{0.3mm} 
\setlength{\itemsep}{0.3mm} 
\item[Hardwired:] tag $\gtag^\prime$, public key $\pkhat$ (this is an obfuscated circuit), puncturable PRF key $\prf^\prime_{\ne \gtag^\prime}$, plaintexts $x_0,x_1$, and ciphertext $\ct_{\gtag^\prime}$.
\item[Input:] tag $\gtag \in \zo{\ell}$.
\item[Padding:] circuit is padded to size $\padsize_{\sfE}\seteq\padsize_{\sfE}(\secp,n,s)$, which is determined in analysis.
\end{description}
\begin{enumerate}
\setlength{\parskip}{0.3mm} 
\setlength{\itemsep}{0.3mm} 
\item If $\gtag^\prime = \gtag$, output $\ct_{\gtag^\prime}$.
\item Else, compute $r^\prime_{\gtag} \gets \prf^\prime_{\ne \gtag^\prime}(\gtag)$ and the circuit $\pkhat$ on input $\gtag$, that is, $\pk_{\gtag} \gets \pkhat(\gtag)$,
\begin{description}
\item[If $\gtag > \gtag^\prime$:] Output $\ct_\gtag \gets \oneSDFE.\Enc(\pk_{\gtag},x_0;r^\prime_{\gtag})$.
\item[If $\gtag < \gtag^\prime$:] Output $\ct_\gtag \gets \oneSDFE.\Enc(\pk_{\gtag},x_1;r^\prime_{\gtag})$.
\end{description}
\end{enumerate}
}

\item[$\hybij{3}{i}$:] We change $r_{i} = \prf_{\prfkey}(\gtag^\prime_i)$ and $r^\prime_{i} = \prf^\prime_{\prfkey^\prime}(\gtag^\prime_i)$ into uniformly random $r_{i}$ and $r^\prime_{i}$. Due to the pseudorandomness at punctured points of puncturable PRF, it holds that $\abs{\Adv_{3}^{i}- \Adv_{2}^{i} }\le \negl(\secp)$ for every $i \in [2^\ell]$.

\item[$\hybij{4}{i}$:] We change $\ct_{\gtag^\prime_i}$ from $\oneSDFE.\Enc(\pk_{\gtag^\prime_i},x_0)$ to $\oneSDFE.\Enc(\pk_{\gtag^\prime_i},x_1)$.
In $\hybij{3}{i}$ and $\hybij{4}{i}$, we do not need randomness to generate $\pk_{\gtag^\prime_i}$ and $\ct_{\gtag^\prime_i}$.
We just hardwire $\pk_{\gtag^\prime_i}$ and $\ct_{\gtag^\prime_i}$ into $\sfS^*_\sfonefe$ and $\sfE^*_\sfonefe$, respectively. Therefore, for every $i \in [2^\ell]$, $\abs{\Adv_{4}^{i}- \Adv_{3}^{i} }\le \negl(\secp)$ follows from the adaptive security of $\oneSDFE$ under the master public key $\pk_{\gtag^\prime_i}$.
\begin{lemma}\label{lem:ind-by-one-pkfe}
It holds that $\abs{\Adv_{4}^{i}- \Adv_{3}^{i} }\le \negl(\secp)$ for all $i \in [2^\ell]$ if $\oneSDFE$ is adaptively single-key secure.
\end{lemma}
\begin{proof}[Proof of~\cref{lem:ind-by-one-pkfe}]
We construct an adversary $\qB$ in the selective security game of $\oneSDFE$ as follows. To simulate the adaptive security game of $\SDFE$, $\qB$ runs $\qA$ attacking $\oneSDFE$. $\qA$ adaptively sends key queries $f_1,\cdots,f_q$. $\qB$ simulates the game of $\SDFE$ as follows.
\begin{description}
\item[Setup:] $\qB$ receives a public key $\pk_{\gtag^\prime_i}$. Then, $\qB$ chooses $\gtag_1,\ldots,\gtag_q \in \zo{\ell}$ and generates $\prfkey$ and $\prfkey^\prime$ by using $\prfgen(1^\secp)$ and $\prfgen^\prime(1^\secp)$, $\prf_{\ne \gtag^\prime_i} \gets \Puncture(\prfkey,\gtag^\prime_i)$, $\prf^\prime_{\ne \gtag^\prime_i} \gets \Puncture(\prfkey^\prime,\gtag^\prime_i)$, and the public key $\pkhat \seteq \iO_1(\sfS^*_\sfonefe[\gtag^\prime_i,\prf_{\ne \gtag^\prime_i},\pk_{\gtag^\prime_i}])$ for $\SDFE$ according to \cref{fig:setup-pkfe-punc} by using the given $\pk_{\gtag^\prime_i}$. $\qB$ sends $\pkhat$ to $\qA$.
\item[Key Generation:] When $f_k$ (i.e., $k$-th query) is queried, $\qB$ checks $f_k(x_0)=f_k(x_1)$ and outputs $\bot$ if it does not hold. Otherwise, $\qB$ checks whether $\gtag^\prime_{i}=\gtag_{k}$ (i.e., $i$-th tag in $[2^\ell]$ is the same as the tag for $k$-th key query). $\qB$ passes $f_{k}$ to its challenger, receives $\qsk_{f_{k}}$, and sends $(\gtag^\prime_i,\qsk_{f_{k}})$ to $\qA$. If $\gtag^{\prime}_i \ne \gtag_{k} $, $\qB$ generates $\qsk_{f_k} \gets \oneSDFE.\QKeyGen (\msk_{\gtag_{k}},f_k)$ by using $(\pk_{\gtag_{k}},\msk_{\gtag_{k}}) \gets \oneSDFE.\Setup(1^\secp;\prf_{\ne \gtag^\prime_{i}}(\gtag_{k}))$ and returns it. Note that we do not need $\msk_{\gtag^\prime_i}$ for this simulation since $\qB$ receives $\qsk_{f_{k}}$ from its challenger if $\gtag^\prime_i = \gtag_k$. 
\item[Encryption:] When $\qA$ sends $(x_0,x_1)$ to $\qB$, $\qB$ passes $(x_0,x_1)$ to its challenger and receives $\ct^\ast_{\gtag^\prime_i}$. Then, $\qB$ generates the challenge ciphertext $\cthat \gets \iO_2(\sfE^\ast[\gtag^\prime_i,\pkhat,\prf^\prime_{\ne \gtag^\prime_{i}},x_0,x_1,\ct^\ast_{\gtag^\prime_i}])$ (obfuscated $\sfE^*_\sfonefe$~\cref{fig:enc-one-pkfe-punc}).
\end{description}
The simulation is completed. If $\qB$ receives $\ct^\ast_{\gtag^\prime_{i}} = \oneSDFE.\Enc(\pk_{\gtag^\prime_i},x_0)$, it perfectly simulates $\hybij{3}{i}$. If $\qB$ receives $\ct^\ast_{\gtag^\prime_i} = \oneSDFE.\Enc(\pk_{\gtag^\prime_i},x_1)$, it perfectly simulates $\hybij{4}{i}$. This completes the proof of the lemma.
\end{proof}
\item[$\hybij{5}{i}$:] We change $r_i$ and $r^\prime_{i}$ into $r_{i} = \prf_{\prfkey}(\gtag^\prime_i)$ and $r^\prime_{i} = \prf^\prime_{\prfkey^\prime}(\gtag^\prime_i)$. We can show $\abs{\Adv_{5}^{i}- \Adv_{4}^{i} }\le \negl(\secp)$ for every $i \in [2^\ell]$ by using the pseudorandomness at punctured point of puncturable PRF.
\end{description}
From the definition of $\sfS^*_{\sfonefe}$, $\sfE^*_{\sfonefe}$, and $\hybij{1}{i}$, the behaviors of $\sfS^*_{\sfonefe}$ and $\sfE^*_{\sfonefe}$ in $\hybij{5}{i}$ and $\hybij{1}{i+1}$ are the same. Thus, $\abs{\Adv_{1}^{i+1}- \Adv_{5}^{i} }\le \negl(\secp)$ holds for every $i \in [2^\ell -1]$ due to the indistinguishability of $\iO_1$ and $\iO_2$.
It also holds that $\abs{\Adv_{0}(1)- \Adv_{5}^{2^\ell} }\le \negl(\secp)$ based on the security guarantee of $\iO_1$ and $\iO_2$.

There are $O(2^\ell)$ hybrid games. However, if $\negl(\secp)$ is sub-exponentially small for all transitions, $\abs{\Adv_{0}(0)- \Adv_{0}(1)}$ is negligible.
Note that it is sufficient to set $\ell = \polylog(\secp)$.

\paragraph{Padding Parameter.} The proof of security relies on the indistinguishability of obfuscated $\sfS_{\sfonefe}$ and $\sfS^*_{\sfonefe}$ defined in~\cref{fig:setup-pkfe,fig:setup-pkfe-punc},
and that of obfuscated $\sfE_\sfonefe$ and $\sfE_\sfonefe^\ast$ defined in~\cref{fig:enc-one-pkfe,fig:enc-one-pkfe-punc}. Accordingly, we set $\padsize_{\sfS} \seteq \max(\abs{\sfS_{\sfonefe}},\abs{\sfS_{\sfonefe}^*})$ and $\padsize_{\sfE} \seteq \max(\abs{\sfE_{\sfonefe}},\abs{\sfE_{\sfonefe}^*})$.

The circuits $\sfS_\sfonefe$ and $\sfS_\sfonefe^*$ compute a puncturable PRF over domain $\zo{\ell}$ and a key pair of $\oneFE$, and may have punctured PRF keys and a public key hardwired. The circuits $\sfE_\sfonefe$ and $\sfE_\sfonefe^*$ run the circuit $\pkhat$ and compute a puncturable PRF over domain $\zo{\ell}$ and a ciphertext of $\oneSDFE$, and may have punctured PRF keys, tags, plaintexts, and a hard-wired ciphertext. Note that $\ell$ is a polynomial of $\secp$. Thus, it holds that
 \begin{align}
 \pad_{\sfS} & \le \poly(\secp,n,s),\\
 \pad_{\sfE} & \le \poly(\secp,n,s,\abs{\pkhat}).
 \end{align}
 Therefore, we complete the proof of~\cref{thm:SDFE_from_1FE_IO}.
\end{proof}

\begin{theorem}\label{thm:SDFE_anti-piracy_from_1SDFE_IO}
If $\oneSDFE$ is challenge-only strong $\gamma/2$-anti-piracy secure and sub-exponentially adaptively single-key secure, $\iO$ is a sub-exponentially secure IO for $\Ppoly$, and $\PuncPRF$ and $\PuncPRF^\prime$ are sub-exponentially secure puncturable PRFs, then $\SDFE$ is strong $\gamma$-anti-piracy secure.
\end{theorem}

\begin{proof}[Proof of~\cref{thm:SDFE_anti-piracy_from_1SDFE_IO}]
Let $\qA$ be an adversary attacking the strong $\gamma$-anti-piracy of $\SDFE$.
We define a sequence of hybrid games.
For the hybrid games, we define the following values.
Let $q$ be the total number of key queries by $\qA$ except for the challenge query $f^\ast$. Note that $q$ could be any polynomial.
When we start the strong $\gamma$-anti-piracy security game, the adversary $\qA$ is fixed, and $q$ is also fixed. We choose tags $\gtag_1,\ldots,\gtag_q \chosen \zo{\ell}$ for $q$ key queries and $\gtag^\ast$ for the challenge query at the beginning of the game.
We can interpret $\ell$ bit strings as integers and assume that there is no $i,j$ such that $i\ne j$ and $\gtag_i = \gtag_j$ without loss of generality.

\begin{description}
\item[$\hybi{0}$:] The first game is the original strong $\gamma$-anti-piracy experiment, that is, $\expc{\SDFE,\qA}{strong}{anti}{piracy}(1^\secp,\gamma)$.
\begin{enumerate}
\item The challenger generates $(\pkhat,\mskhat)=(\iO_1(\sfS_\sfonefe),\prfkey)\gets \SDFE.\Setup(1^\secp)$ and sends $\pkhat$ to $\qA$.
\item $\qA$ sends key queries $f_k$ to the challenger and the challenger returns $\qskhat_{f_k} = (\gtag_k,\qsk_{f_k,\gtag_k}) \lrun \SDFE.\QKeyGen(\mskhat,f_k)$ where $\qsk_{f_k,\gtag_k}\lrun \oneSDFE.\QKeyGen(\msk_{\gtag_k},f_k)$.
\item At some point $\qA$ sends a challenge query $f^*$ to the challenger and the challenger returns $\qskhat =(\gtag^\ast,\qsk_{f^\ast}) \lrun \SDFE.\QKeyGen(\mskhat,f^\ast)$ where $\qsk_{f^\ast}\lrun \oneSDFE.\QKeyGen(\msk_{\gtag^\ast},f^*)$.
\item Again, $\qA$ sends $f_k$ to the challenger and the challenger returns $\qskhat_{f_k} \lrun \SDFE.\QKeyGen(\mskhat,f_k)$.
\item $\qA$ outputs $(x_0,x_1)$ and two (possibly entangled) quantum decryptors $\pirateD_1 = (\qstateq[\qreg{R}_1],\mat{U}_1)$ and $\pirateD_2 = (\qstateq[\qreg{R}_2],\mat{U}_2)$, where $\forall i\ f_i(x_0)=f_i(x_1)$, $f^*(x_0)\ne f^*(x_1)$, $\qstateq$ is a quantum state over registers $\qreg{R}_1$ and $\qreg{R}_2$, and $\mat{U}_1$ and $\mat{U}_2$ are general quantum circuits.
\item For $i\in \setbk{1,2}$, let $\cP_{i,D}$ be the following mixture of projective measurements acting on some quantum state $\qstateq^\prime$:
\begin{itemize}
\item Sample a uniform $\coin \chosen \bit$. Compute $\cthat \gets \SDFE.\Enc(\pkhat,x_\coin)$ and set $\ct \seteq \cthat$.
\item Run the quantum decryptor $(\qstateq^\prime,\mat{U}_i)$ on input $\ct$. If the outcome is 
$\coin$, output $1$. Otherwise, output $0$.
\end{itemize}
Let $D$ be the distribution over pairs $(\coin,\ct)$ defined in the first item above, and let $\cE_i = \setbk{\mat{E}^i_{(\coin,\ct)}}_{\coin,\ct}$ be a collection of projective measurements where $\mat{E}^i_{(\coin,\ct)}$ is the projective measurement described in the second item above.
$\cP_{i,D}$ is the mixture of projective measurements associated to $D$ and $\cE_i$.
\item The challenger runs the test for a $\gamma$-good FE decryptor with respect to $\pk$, $(x_0,x_1)$, and $f^\ast$ on $\pirateD_1$ and $\pirateD_2$.
The challenger outputs $1$ if both tests pass, otherwise outputs $0$.
\end{enumerate}
In the following hybrid games, we gradually change $\sfS_\sfonefe$ and $\sfE_\sfonefe$ for the $i$-th tag $\gtag^\prime_i$, which is the $\ell$-bit string representation of $i \in [2^\ell]$, with the following exception.
Let $i^\ast \in [2^\ell]$ is the integer representation of $\gtag^\ast \in \zo{\ell}$.
Instead of going over indices $\setbk{1,\ldots, i^\ast -1,i^\ast,i^\ast+1,\ldots, 2^\ell}$, we go over indices $I_{i^\ast} \seteq \setbk{1,\ldots,i^\ast -1,i^\ast+1,\ldots,2^\ell,i^\ast}$. That is, we skip the tag $\gtag_{i^\ast}$ and go to $\gtag_{i^\ast}$ after we finish the $2^\ell$-th tag $\gtag_{2^\ell}$.
It could happen that $\gtag^\prime_i = \gtag_j$ for some $i\in [2^\ell]$ and $j \in[q]$.
Let $\Adv_{x}^{i}$ be the advantage of $\qA$ in $\hybij{x}{i}$.

\item[$\hybij{1}{i}$:] This is defined for $i \in [I_{i^\ast}]$ and the same as $\hybij{5}{i-1}$ except that we generate $\pkhat$ as obfuscated $\sfS^*_\sfonefe$ described in~\cref{fig:setup-sdfe-punc}. Note that we define $\hybij{5}{0}=\hybi{0}$. In this hybrid game, we set $r_{\gtag^\prime_i} \la \prf_{\prfkey}(\gtag^\prime_i)$, $\prf_{\ne \gtag^\prime_i} \la \Puncture(\prfkey,\gtag^\prime_i)$, and $(\pk_{\gtag^\prime_i},\msk_{\gtag^\prime_i}) \la \oneSDFE.\Setup(1^\secp ;r_{\gtag^\prime_i})$.

The behavior of $\sfS^*_\sfonefe$ is the same as that of $\sfS_\sfonefe$ since the hard-wired $\pk_{\gtag^\prime_1}$ in $\sfS^*_\sfonefe$ is the same as the output of $\sfS_\sfonefe$ on the input $\gtag^\prime_1$. Their size is also the same since we pad circuit $\sfS_\sfonefe$ to have the same size as $\sfS_\sfonefe^*$. Then, we can use the indistinguishability guarantee of $\iO_1$ and it holds that $\iO_1(\sfS_\sfonefe[\prfkey])$ is computationally indistinguishable from $\iO_1(\sfS^\ast_\sfonefe[\gtag^\prime_i,\prf_{\ne \gtag^\prime_i},\pk_{\gtag^\prime_i}])$. Let $D^{(1,i)}$ be the distribution over pairs $(\coin,\ct)$ defined in $\hybi{0}$ except that $\iO_1(\sfS^\ast_\sfonefe[\gtag^\prime_i,\prf_{\ne \gtag^\prime_i},\allowbreak\pk_{\gtag^\prime_i}])$ as $\pkhat$. Then, $D^{(1,1)}$ is computationally indistinguishable from $D$.


\item[$\hybij{2}{i}$:] This is defined for $i \in [I_{i^\ast}]$ and the same as $\hybij{1}{i}$ except that we change the distribution $D^{(1,i)}$ over pairs $(\coin,\ct)$ into $D^{(2,i)}$ as follows.
\begin{itemize}
 \item Sample a uniform $\coin \chosen \bit$.
 \item Compute $r^\prime_{\gtag^\prime_i} \la \prf^\prime_{\prfkey^\prime}(\gtag^\prime_i)$, $\prf^\prime_{\ne \gtag^\prime_i} \la \Puncture^\prime(\prfkey^\prime,\gtag^\prime_i)$, $\ct_{\gtag^\prime_i} \la \oneSDFE.\Enc(\pk_{\gtag^\prime_i},x_\coin ;r^\prime_{\gtag^\prime_i})$, and $\cthat \gets \iO_2(\sfE^\ast_\sfonefe[\gtag^\prime_i,\gtag^\ast,\pkhat,\prf^\prime_{\ne \gtag^\prime_i},x_\coin,x_1,\ct_{\gtag^\prime_i}])$ where $\sfE^\ast_\sfonefe$ is described in~\cref{fig:enc-one-sdfe-punc}. Set $\ct \seteq \cthat$.
 \end{itemize} 

\protocol
{Setup Circuit $\sfS^*_{\sfonefe} [\gtag^\prime,\prf_{\ne \gtag^\prime},\pk_{\gtag^\prime}](\gtag)$}
{Circuit $\sfS^*_{\sfonefe} [\gtag^\prime,\prf_{\ne \gtag^\prime},\pk_{\gtag^\prime}]$.
}
{fig:setup-sdfe-punc}
{
\ifnum\llncs=1
\scriptsize
\else
\fi
\begin{description}
\setlength{\parskip}{0.3mm} 
\setlength{\itemsep}{0.3mm} 
\item[Hardwired:] tag $\gtag^\prime$, puncturable PRF key $\prf_{\ne \gtag^\prime}$, and $\oneSDFE$ public-key $\pk_{\gtag^\prime}$.
\item[Input:] tag $\gtag \in \zo{\ell}$.
\item[Padding:] circuit is padded to size $\padsize_{\sfS}\seteq\padsize_{\sfS}(\secp,n,s)$, which is determined in analysis.
\end{description}
\begin{enumerate}
\setlength{\parskip}{0.3mm} 
\setlength{\itemsep}{0.3mm} 
\item If $\gtag^\prime = \gtag$, output $\pk_{\gtag^\prime}$.
\item Else, compute $r \gets \prf_{\ne \gtag^\prime}(\gtag)$.
\item Compute $(\pk_{\gtag}, \msk_{\gtag}) \gets \oneSDFE.\Setup(1^\secp;r_i)$ and output $\pk_{\gtag}$.
\end{enumerate}
}

\protocol
{Encryption Circuit $\sfE^*_\sfonefe[\gtag^\prime,\gtag^\ast,\pkhat,\prf^\prime_{\ne \gtag^\prime},x_\coin, x_1,\ct_{\gtag^\prime}](\gtag)$}
{Circuit $\sfE^*_\sfonefe[\gtag^\prime,\gtag^\ast,\pkhat,\prf^\prime_{\ne \gtag^\prime},x_\coin,x_1,\ct_{\gtag^\prime}]$.
}
{fig:enc-one-sdfe-punc}
{
\ifnum\llncs=1
\scriptsize
\else
\fi
\begin{description}
\setlength{\parskip}{0.3mm} 
\setlength{\itemsep}{0.3mm} 
\item[Hardwired:] tags $\gtag^\prime$, $\gtag^\ast$, public key $\pkhat$ (this is an obfuscated circuit), puncturable PRF key $\prf^\prime_{\ne \gtag^\prime}$, plaintexts $x_\coin,x_1$, and ciphertext $\ct_{\gtag^\prime}$.
\item[Input:] tag $\gtag \in \zo{\ell}$.
\item[Padding:] circuit is padded to size $\padsize_{\sfE}\seteq\padsize_{\sfE}(\secp,n,s)$, which is determined in analysis.
\end{description}
\begin{enumerate}
\setlength{\parskip}{0.3mm} 
\setlength{\itemsep}{0.3mm} 
\item If $\gtag^\prime = \gtag$, output $\ct_{\gtag^\prime}$.
\item Else, compute $r^\prime_{\gtag} \gets \prf^\prime_{\ne \gtag^\prime_i}(\gtag)$ and the circuit $\pkhat$ on input $\gtag$, that is, $\pk_{\gtag} \gets \pkhat(\gtag)$,
\begin{description}
    \item[If $\gtag = \gtag^\ast$:] Output $\ct_{\gtag^\ast} \gets \oneSDFE.\Enc(\pk_{\gtag^\ast},x_\coin;r^\prime_{\gtag^\ast})$
\item[If $\gtag > \gtag^\prime$:] Output $\ct_\gtag \gets \oneSDFE.\Enc(\pk_{\gtag},x_\coin;r^\prime_{\gtag})$.
\item[If $\gtag < \gtag^\prime$:] Output $\ct_\gtag \gets \oneSDFE.\Enc(\pk_{\gtag},x_1;r^\prime_{\gtag})$.
\end{description}
\end{enumerate}
}

The behavior of $\sfE^*_\sfonefe$ is the same as that of $\sfE_\sfonefe$ since the hard-wired $\ct_{\gtag^\prime_i}$ in $\sfE^*_\sfonefe$ is the same as the output of $\sfE_\sfonefe$ on the input $\gtag^\prime_i$.
Moreover, both circuits have the same size by padding $\pad_\sfE$. Then, we can use the indistinguishability of $\iO_2$, and it holds that $D^{(2,i)}$ is computationally indistinguishable from $D^{(1,i)}$.


\item[$\hybij{3}{i}$:] This is defined for $i \in [I_{i^\ast}]$ and the same as $\hybij{2}{i}$ except that we use distribution $D^{(3,i)}$, which is the same as $D^{(2,i)}$ except that setup randomness $r_{\gtag^\ast}\chosen \cR_{\Setup}$ and encryption randomness $r^\prime_{\gtag^\ast}\chosen \cR_{\Enc}$ are uniformly random. By the punctured pseudorandomness of $\prf$ and $\prf^\prime$, $D^{(3,i)}$ is computationally indistinguishable from $D^{(2,i)}$.


\item[$\hybij{4}{i}$:] This is defined for $i \in [I_{i^\ast}]\setminus \setbk{i^\ast}$ and the same as $\hybij{3}{i}$ except that we use distribution $D^{(4,i)}$, which is the same as $D^{(3,i)}$ except that the ciphertext $\ct_{\gtag^\prime_i}$ is $\oneSDFE.\Enc(\pk_{\gtag^\prime_i},x_1)$ instead of $\oneSDFE.\Enc(\pk_{\gtag^\prime_i},x_\coin)$.
In $\hybij{3}{i}$ and $\hybij{4}{i}$, we do not need randomness to generate $\pk_{\gtag^\prime_i}$ and $\ct_{\gtag^\prime_i}$.
We just hardwire $\pk_{\gtag^\prime_i}$ and $\ct_{\gtag^\prime_i}$ into $\sfS^*_\sfonefe$ and $\sfE^*_\sfonefe$, respectively.
In addition, it holds that $f_k(x_0)=f_k(x_1)$ for all $k\in[q]$ by the requirement of strong $\gamma$-anti-piracy security. For every $i \in [I_{i^\ast}] \setminus \setbk{i^\ast}$, by the adaptive security of $\oneSDFE$ under the public key $\pk_{\gtag^\prime_i}$, $D^{(4,i)}$ is computationally indistinguishable from $D^{(3,i)}$. Note that we do not need to generate a functional decryption key for $i\in[I_{i^\ast}]$ such that $\gtag^\prime_{i} \ne \gtag_{k}$.
The detail of this indistinguishability is almost the same as~\cref{lem:ind-by-one-pkfe}, and we omit it.


\item[$\hybij{5}{i}$:] This is defined for $i\in [I_{i^\ast}]\setminus \setbk{i^\ast}$ and the same as $\hybij{4}{i}$ except that we use distribution $D^{(5,i)}$, which is the same as $D^{(4,i)}$ except that we use $r_{i} = \prf_{\prfkey}(\gtag^\prime_i)$ and $r^\prime_{i} = \prf^\prime_{\prfkey^\prime}(\gtag^\prime_i)$.
By the punctured pseudorandomness of $\prf$ and $\prf^\prime$, $D^{(5,i)}$ is computationally indistinguishable from $D^{(4,i)}$.


We finish describing one cycle of our hybrid games.
We move from $\hybij{5}{i}$ to $\hybij{1}{i+1}$.
From the definition of $\sfS^*_{\sfonefe}$ and $\hybij{1}{i}$, the behaviors of $\sfS^*_{\sfonefe}$ in $\hybij{5}{i}$ and $\hybij{1}{i+1}$ are the same.
It holds that $D^{(1,i+1)}$ is computationally indistinguishable from $D^{(5,i)}$.
So, we can arrive at $\hybi{3}^{i^\ast}$ by iterating the transitions so far.
\item[$\hybij{4}{i^\ast}$:] This is the same as $\hybij{3}{i^\ast}$ except that we change the distribution $D^{(3,i^\ast)}$ over pairs $(\coin,\ct)$ into $D^{(4,i^\ast)}$ as follows.
\begin{itemize}
 \item Sample a uniform $\coin \chosen \bit$.
 \item Sample $r^\prime_{\gtag^\ast} \chosen \cR_{\Enc}$ and compute $\ct_{\gtag^\ast} \la \oneSDFE.\Enc(\pk_{\gtag^\ast},x_\coin ;r^\prime_{\gtag^\ast})$, $\prfkey^\prime \la \prfgen^\prime(1^\secp)$, $\prf^\prime_{\ne \gtag^\ast} \la \Puncture^\prime(\prfkey^\prime,\gtag^\ast)$, and $\cthat \gets \iO_2(\sfE^{\ast\ast}_\sfonefe[\gtag^\ast,\pkhat,\allowbreak\prf^\prime_{\ne \gtag^\ast},x_1,\ct_{\gtag^\ast}])$ where $\sfE^{\ast\ast}_\sfonefe$ is described in~\cref{fig:enc-one-sdfe-punc-last}. Set $\ct \seteq \cthat$.
 \end{itemize}

\protocol
{Encryption Circuit $\sfE^{\ast\ast}_\sfonefe[\gtag^\ast,\pkhat,\prf^\prime_{\ne \gtag^\ast}, x_1,\ct_{\gtag^\ast}](\gtag)$}
{Circuit $\sfE^{\ast\ast}_\sfonefe[\gtag^\ast,\pkhat,\prf^\prime_{\ne \gtag^\ast},x_1,\ct_{\gtag^\ast}]$.
}
{fig:enc-one-sdfe-punc-last}
{
\ifnum\llncs=1
\scriptsize
\else
\fi
\begin{description}
\setlength{\parskip}{0.3mm} 
\setlength{\itemsep}{0.3mm} 
\item[Hardwired:] tag $\gtag^\ast$, public key $\pkhat$ (this is an obfuscated circuit), puncturable PRF key $\prf^\prime_{\ne \gtag^\ast}$, plaintext $x_1$, and ciphertext $\ct_{\gtag^\ast}$.
\item[Input:] tag $\gtag \in \zo{\ell}$.
\item[Padding:] circuit is padded to size $\padsize_{\sfE}\seteq\padsize_{\sfE}(\secp,n,s)$, which is determined in analysis.
\end{description}
\begin{enumerate}
\setlength{\parskip}{0.3mm} 
\setlength{\itemsep}{0.3mm} 
\item If $\gtag^\ast = \gtag$, output $\ct_{\gtag^\ast}$.
\item Else, compute $r^\prime_{\gtag} \gets \prf^\prime_{\ne \gtag^\ast}(\gtag)$ and the circuit $\pkhat$ on input $\gtag$, that is, $\pk_{\gtag} \gets \pkhat(\gtag)$,
\item Output $\ct_{\gtag} \gets \oneSDFE.\Enc(\pk_{\gtag},x_1;r^\prime_{\gtag})$.
\end{enumerate}
}

From the definitions of $\sfE^\ast_{\sfonefe}$, $\sfE^{\ast\ast}_{\sfonefe}$, and $\hybij{3}{i^\ast}$, the behaviors of $\sfE^{\ast}_{\sfonefe}$ in $\hybij{3}{i^\ast}$ is the same as that of $\sfE^{\ast\ast}_{\sfonefe}$ in $\hybij{4}{i^\ast}$.
Moreover, both circuits have the same size by padding $\pad_\sfE$. Then, we can use the indistinguishability of $\iO_2$, and it holds that $D^{(4,i^\ast)}$ is computationally indistinguishable from $D^{(3,i^\ast)}$.
This change is for erasing information about $\coin$ from $\sfE^{\ast}_{\sfonefe}$.


\item[$\hybi{6}$:] This is the same as $\hybij{4}{i^\ast}$ except that we choose PRF keys $\prfkey_1,\prfkey_2 \chosen \zo{\secp}$ at the beginning of the game and change the distribution $D^{(4,i^\ast)}$ over pairs $(\coin,\ct)$ into $D^{(6)}$ as follows.
\begin{itemize}
 \item Sample a uniform $\coin \chosen \bit$.
 \item Sample $r^\prime_{\gtag^\ast} \chosen \cR_{\Enc}$, compute $\ct_{\gtag^\ast} \seteq \oneSDFE.\Enc(\pk_{\gtag^\ast},x_\coin ;r^\prime_{\gtag^\ast})$, $r_1\gets \prf^{(1)}_{\prfkey_1}(\ct_{\gtag^\ast})$, $r_2\gets \prf^{(2)}_{\prfkey_2}(\ct_{\gtag^\ast})$, $\prfkey^\prime \la \PRF^\prime.\Gen(1^\secp;r_1)$, $\prf^\prime_{\ne \gtag^\ast} \la \Puncture^\prime(\prfkey^\prime,\gtag^\ast)$, and $\cthat \seteq \iO_2(\sfE^{\ast\ast}_\sfonefe[\gtag^\ast,\pkhat,\allowbreak \prf^\prime_{\ne \gtag^\ast},x_1,\ct_{\gtag^\ast}];r_2)$ where $\sfE^{\ast\ast}_\sfonefe$ is described in~\cref{fig:enc-one-sdfe-punc-last}. Set $\ct \seteq \cthat$.
 \end{itemize} 
 That is, we de-randomize generation of $\prfkey^\prime$ and $\cthat$. Note that the puncturing algorithm $\Puncture^\prime$ is deterministic in the GGM construction based puncturable PRF~\cite{JACM:GolGolMic86,AC:BonWat13,CCS:KPTZ13,PKC:BoyGolIva14}.
By the pseudorandomness of $\prf^{(1)}$ and $\prf^{(2)}$, $D^{(6)}$ is computationally indistinguishable from $D^{(4,i^\ast)}$.

There are $O(2^\ell)$ hybrid games between $\hybi{0}$ and $\hybi{6}$. However, if all building blocks are sub-exponentially secure, $D^{(6)}$ and $D$ (the original distribution in $\hybi{0}$) are computationally indistinguishable since we can set $\ell=\polylog(\secp)$.
Note that we do not need the sub-exponential security for the challenge-only strong $\gamma/2$-anti-piracy security of $\oneSDFE$ since we do not use the anti-piracy security so far.

Thus, by~\cref{thm:ind_distribution_TI}, it holds that
\begin{align}
& \Tr[(\TI_{\frac{1}{2}+\gamma -\epsilon}(\cP_{1,D^{(6)}}) \tensor \TI_{\frac{1}{2} +\gamma - \epsilon }(\cP_{2,D^{(6)}}))\qstateq] \\
& \ge \Tr[(\TI_{\frac{1}{2}+\gamma }(\cP_{1,D^{}})\tensor \TI_{\frac{1}{2} +\gamma }(\cP_{2,D^{}}))\qstateq] - \delta,
\end{align}
where $\delta$ is some negligible function.
That is, it holds that $\Adv_{6} \ge  \Adv_{0} - \negl(\secp)$.
\end{description}

We prove the following.
\begin{lemma}\label{lem:1SDFE_reduction}
If $\oneSDFE$ is challenge-only strong $\gamma/2$-anti-piracy secure, it holds that $\Adv_{6} \le \negl(\secp)$.
\end{lemma}
\begin{proof}[Proof of~\cref{lem:1SDFE_reduction}]
Suppose that $\Adv_{6}$ is non-negligible for a contradiction.
We construct a QPT algorithm $\qB$ for the strong $\gamma/2$-anti-piracy game of $\oneSDFE$ by using the adversary $\qA$ of the strong $\gamma$-anti-piracy game of $\SDFE$ as follows.
\begin{enumerate}
\item  $\qB$ chooses $\gtag_1,\ldots,\gtag_q, \gtag^\ast \chosen \zo{\ell}$ and $\prfkey_1,\prfkey_2 \chosen \zo{\secp}$, and generates $\prfkey \gets \prfgen(1^\secp)$.
\item $\qB$ generates $\prf_{\ne \gtag^\ast} \la \Puncture(\prfkey,\gtag^\ast)$. 
\item $\qB$ receives $\pk^\ast$ from its challenger, sets $\pk_{\gtag^\ast} \seteq \pk^\ast$, and constructs $\sfS^*_\sfonefe[\gtag^\ast,\prf_{\ne \gtag^\ast},\allowbreak\pk_{\gtag^\ast}]$ described in~\cref{fig:setup-sdfe-punc} by choosing $\gtag^\ast$ for the challenge query $f^\ast$.
$\qB$ sets $\pkhat \seteq \iO_1(\sfS^*_\sfonefe[\gtag^\ast,\prf_{\ne \gtag^\ast},\pk_{\gtag^\ast}])$ and sends it to $\qA$.
\item When $\qA$ sends a key query $f_i$, $\qB$ generates $(\pk_{\gtag_i},\msk_{\gtag_i})\gets \oneSDFE.\Setup(1^\secp;\allowbreak\PRF_{\prfkey_{\ne \gtag^\ast}}(\gtag_i))$ and $\qsk_{f_i,\gtag_i} \gets \oneSDFE.\QKeyGen(\msk_{\gtag_i},f_i)$, sets $\qskhat_{f_i,\gtag_i} \seteq (\gtag_i,\qsk_{f_i,\gtag_i})$, and sends $\qskhat_{f_i}$ to $\qA$.
\item When $\qA$ sends the challenge query $f^\ast$, $\qB$ passes it to its challenger and receives $\qsk_{f^\ast} \gets \oneSDFE.\QKeyGen(\msk^\ast,f^\ast)$, and sends $\qskhat_{f^\ast} \seteq (\gtag^\ast,\qsk_{f^\ast})$ to $\qA$.
\item Again, $\qB$ can answer key queries $f_i$ from $\qA$ as the fourth item.
\item At some point, $\qB$ receives $(x_0,x_1)$ and two (possibly entangled) quantum decryptors $\pirateD_1 = (\qstateq[\qreg{R}_1],\mat{U}_1)$ and $\pirateD_2 = (\qstateq[\qreg{R}_2],\mat{U}_2)$, where $f^*(x_0)\ne f^*(x_1)$, from $\qA$. For $i\in \setbk{1,2}$, let $\cP_{i,D^{(6)}}$ be the following mixture of projective measurements acting on some quantum state $\qstateq^\prime$:
\begin{itemize}
 \item Sample a uniform $\coin \chosen \bit$.
 \item Sample $r^\prime_{\gtag^\ast} \chosen \cR_{\Enc}$, $\ct_{\gtag^\ast} \la \oneSDFE.\Enc(\pk_{\gtag^\ast},x_\coin ;r^\prime_{\gtag^\ast})$.
\item Generate $r_1\gets \prf^{(1)}_{\prfkey_1}(\ct_{\gtag^\ast})$, $r_2\gets \prf^{(2)}_{\prfkey_2}(\ct_{\gtag^\ast})$, $\prfkey^\prime \gets \prfgen^\prime(1^\secp;r_1)$, and $\prf^\prime_{\ne \gtag^\ast}=\Puncture^\prime(\prfkey^\prime,\gtag^\ast)$.
\item Generate $\cthat \gets \iO_2(\sfE^{\ast\ast}_\sfonefe[\gtag^\ast,\pkhat,\prf^\prime_{\ne \gtag^\ast},x_1,\ct_{\gtag^\ast}];r_2)$ where $\sfE^{\ast\ast}_\sfonefe$ is described in~\cref{fig:enc-one-sdfe-punc-last}.
 \item Run the quantum decryptor $(\qstateq^\prime,\mat{U}_i)$ on input $\cthat$. If the outcome is 
 $\coin$, output $1$. Otherwise, output $0$.
\end{itemize}
\item $\qB$ sets quantum decryptors $\pirateD_1^\ast$ and $\pirateD_2^\ast$ as follows.

In this item, we denote pure state $\ket{x}\bra{x}$ by $\ket{x}$ for ease of notation.
First, $\qB$ constructs $(\qstateq^\ast[\qreg{R}_i],\mat{U}_i^\ast)$ such that
\begin{itemize}
\item Set  $\qstateq^\ast[\qreg{R}_i] \seteq \qstateq[\qreg{R}_i] \tensor \ket{0^{\abs{\cthat}}} \tensor \ket{\gtag^\ast, \pkhat, 0^\ell, x_1, \prfkey_1,\prfkey_2,0^{\abs{r_1}},0^{\abs{r_2}},0^{\ell}} \tensor \ket{\ct_{\gtag^\ast}}$. Note that $\ket{\ct_{\gtag^\ast}}$ is the input for $(\qstateq^\ast[\qreg{R}_i],\mat{U}_i^\ast)$.
\item Unitary $\mat{U}^\ast_i$ acts on $\qstateq^\ast[\qreg{R}_i]$ and the result is
\[
\qstateq^\prime[\qreg{R}_i] \seteq \qstateq[\qreg{R}_i] \tensor \ket{\cthat} \tensor \ket{\gtag^\ast, \pkhat, \prf^\prime_{\ne \gtag^\ast}, x_1,\prfkey_1,\prfkey_2,r_1,r_2,\prfkey^\prime}\tensor \ket{\ct_{\gtag^\ast}},
\]
where $r_1 \seteq \prf^{(1)}_{\prfkey_1}(\ct_{\gtag^\ast})$, $r_2 \seteq \prf^{(2)}_{\prfkey_2}(\ct_{\gtag^\ast})$, $\prfkey^\prime \seteq \prfgen^\prime(1^\secp;r_1)$, $\prf^\prime_{\ne \gtag^\ast} \seteq \Puncture^\prime(\prfkey^\prime,\gtag^\ast)$, and $\cthat \seteq  \iO_2(\sfE^{\ast\ast}_\sfonefe[\gtag^\ast,\pkhat,\prf^\prime_{\ne \gtag^\ast},x_1,\ct_{\gtag^\ast}];r_2)$. That is, $\mat{U}_i^\ast$ constructs $\sfE^{\ast\ast}_\sfonefe[\gtag^\ast,\pkhat,\prf^\prime_{\ne \gtag^\ast},x_1,\ct_{\gtag^\ast}]$ described in~\cref{fig:enc-one-sdfe-punc-last} and generates $\cthat \seteq \iO_2(\sfE^{\ast\ast}_\sfonefe[\gtag^\ast,\pkhat,\prf^\prime_{\ne \gtag^\ast},x_1,\ct_{\gtag^\ast}];r_2)$. Note that $\qB$ can construct it since $\qB$ has $\gtag^\ast$, $\pkhat$, $\prfkey_1$, $\prfkey_2$, and $x_1$.
\end{itemize}
It is easy to see that if we apply $(\mat{U}_i \tensor \mat{I})$ to $\qstateq^\prime[\qreg{R}_i]$, we obtain an output of $\pirateD_i$ for input $\cthat$.
Thus, $\qB$ sets $\pirateD_1^\ast \seteq (\qstateq^\ast[\qreg{R}_1],(\mat{U}_1\tensor \mat{I})\mat{U}_1^\ast)$ and $\pirateD_2^\ast \seteq (\qstateq^\ast[\qreg{R}_2],(\mat{U}_2 \tensor \mat{I})\mat{U}_2^\ast)$.
$\qB$ sends $(x_0,x_1)$ and $\pirateD_1^\ast$ and $\pirateD_2^\ast$ to its challenger.
\end{enumerate}
The challenger of $\oneSDFE$ runs the test by using the following.
Let $\cP_{i,D^\prime}$ be the following mixture of projective measurements acting on some quantum state $\qstateq_{\oneSDFE}$:
\begin{itemize}
\item Sample a uniformly random $\coin\chosen \bit$ and compute $\ct_{\gtag^\ast} \gets \oneSDFE.\Enc(\pk_{\gtag^\ast},\allowbreak x_\coin)$.
\item Run $\coin^\prime \gets \pirateD^\ast(\ct)$. If $\coin^\prime = \coin$, output $1$, otherwise output $0$.
\end{itemize}

By the construction of $\mat{U}_i^\ast$, quantum decryptor $\pirateD_i^\ast$ takes a ciphertext $\ct_{\gtag^\ast}$ under $\pk_{\gtag^\ast}$ as input, converts it into a ciphertext $\cthat$ of $\SDFE$, and apply the quantum decryptor $\pirateD_i$ from $\qA$. The converted ciphertext $\cthat$ perfectly simulates a ciphertext in $\hybi{6}$ if an input to $\pirateD_i^\ast$ is a ciphertext under $\pk_{\gtag^\ast}$.

We assumed that $\Adv_{6}$ is non-negligible at the beginning of this proof.
That is, applying $\TI_{\frac{1}{2}+\gamma -\epsilon}(\cP_{i,D^{(6)}})$ on $\qstateq^\ast[\qreg{R}_i]$ results in two outcomes $1$ with non-negligible probability.
That is, it holds that
\begin{align}
\Tr[\TI_{\frac{1}{2}+\gamma -\epsilon}(\cP_{1,D^{(6)}}) \tensor \TI_{\frac{1}{2}+\gamma - \epsilon}(\cP_{2,D^{(6)}})\qstateq^\ast] > \negl(\secp).
\end{align}
This means that for $i\in \setbk{1,2}$, $\qstateq^\ast[\qreg{R}_i]$ is a $(\gamma -\epsilon)$-good distinguisher with respect to ciphertexts generated according to $D^{(6)}$.

Thus, if $\pirateD_i$ works as a $(\gamma- \epsilon)$-good quantum distinguisher for $\SDFE$, $\pirateD_i^\ast$ works as a $(\gamma - \epsilon )$-good quantum distinguisher for $\oneSDFE$, where $\epsilon= \frac{\gamma}{2}$. This completes the proof of~\cref{lem:1SDFE_reduction}.
\end{proof}

By using the relationships among $\adva{\SDFE,\qA}{}$, $\Adv_{x}^{i}$ for $x\in \setbk{1,\cdots,5}$ and $i\in [I_{i^\ast}]$, $\Adv_{x}^{i^\ast}$ for $x\in \setbk{1,\cdots,4}$, and $\Adv_6$ that we show above, we obtain $\advc{\SDFE,\qA}{strong}{anti}{piracy}(\secp,\gamma) \le \negl(\secp)$.
We almost complete the proof of~\cref{thm:SDFE_anti-piracy_from_1SDFE_IO}. Lastly, we need to set the padding parameters for IO security.

\paragraph{Padding Parameter.} The proof of security relies on the indistinguishability of obfuscated $\sfS_{\sfonefe}$ and $\sfS^*_{\sfonefe}$ defined in Figures~\ref{fig:setup-pkfe} and \ref{fig:setup-sdfe-punc},
and that of obfuscated $\sfE_\sfonefe$, $\sfE_\sfonefe^\ast$, and $\sfE_\sfonefe^{\ast\ast}$ defined in~\cref{fig:enc-one-pkfe,fig:enc-one-sdfe-punc,fig:enc-one-sdfe-punc-last}. Accordingly, we set $\padsize_{\sfS} \seteq \max(\abs{\sfS_{\sfonefe}},\abs{\sfS_{\sfonefe}^*})$ and $\padsize_{\sfE} \seteq \max(\abs{\sfE_{\sfonefe}},\abs{\sfE_{\sfonefe}^*},\abs{\sfE_\sfonefe^{\ast\ast}})$.

The circuits $\sfS_\sfonefe$ and $\sfS_\sfonefe^*$ compute a puncturable PRF over domain $\zo{\ell}$ and a key pair of $\oneFE$, and may have punctured PRF keys and a public key hardwired. The circuits $\sfE_\sfonefe$, $\sfE_\sfonefe^*$, and $\sfE_\sfonefe^{\ast\ast}$ run the circuit $\pkhat$ and compute a puncturable PRF over domain $\zo{\ell}$ and a ciphertext of $\oneSDFE$, and may have punctured PRF keys, tags, plaintexts, and a hard-wired ciphertext. Note that $\ell$ is a polynomial of $\secp$. Thus, it holds that
 \begin{align}
 \pad_{\sfS} & \le \poly(\secp,n,s),\\
 \pad_{\sfE} & \le \poly(\secp,n,s,\abs{\pkhat}).
 \end{align}
 Therefore, we complete the proof of~\cref{thm:SDFE_anti-piracy_from_1SDFE_IO}.
\end{proof}

We obtain the following corollary from~\cref{cor:strong_1SDFE_from_IO_LWE,thm:SDFE_from_1FE_IO,thm:SDFE_anti-piracy_from_1SDFE_IO}.

\begin{corollary}\label{cor:strong_SDFE_from_IO_LWE}
Assuming the existence of sub-exponentially secure IO for $\Ppoly$, and the sub-exponential hardness of QLWE, there exists an SDFE scheme for $\Ppoly$ that satisfies adaptive security (\cref{def:adaptive_SDFE}) and strong $\gamma$-anti-piracy security for any inverse polynomial $\gamma$.
\end{corollary}

\fi
	
	\ifnum\anonymous=0
	\ifnum\acknowledgments=1
	\paragraph{{\bf Acknowledgments.}}	
	\acknowledgmenttext
	\fi
	\fi
	\ifnum\llncs=1
	\bibliographystyle{extreme_alpha}
	\bibliography{abbrev3,crypto,siamcomp_jacm,other-bib}
	\else
	\ifnum\choosebibstyle=0
	\else
	\ifnum\choosebibstyle=1
	\bibliographystyle{alpha}
	\else
	\bibliographystyle{abbrv}
	\fi
	\bibliography{abbrev3,crypto,siamcomp_jacm,other-bib}

\newcommand{\etalchar}[1]{$^{#1}$}
\begin{thebibliography}{GGHW17}

\bibitem[Aar09]{CCC:Aaronson09}
Scott Aaronson.
\newblock Quantum copy-protection and quantum money.
\newblock In {\em Proceedings of the 24th Annual {IEEE} Conference on
  Computational Complexity, {CCC} 2009, Paris, France, 15-18 July 2009}, pages
  229--242. {IEEE} Computer Society, 2009.

\bibitem[ABDP15]{PKC:ABDP15}
Michel Abdalla, Florian Bourse, Angelo {De Caro}, and David Pointcheval.
\newblock Simple functional encryption schemes for inner products.
\newblock In Jonathan Katz, editor, {\em PKC~2015}, volume 9020 of {\em
  {LNCS}}, pages 733--751. Springer, Heidelberg, March~/~April 2015.

\bibitem[ABSV15]{C:ABSV15}
Prabhanjan Ananth, Zvika Brakerski, Gil Segev, and Vinod Vaikuntanathan.
\newblock From selective to adaptive security in functional encryption.
\newblock In Rosario Gennaro and Matthew J.~B. Robshaw, editors, {\em
  CRYPTO~2015, Part~II}, volume 9216 of {\em {LNCS}}, pages 657--677. Springer,
  Heidelberg, August 2015.

\bibitem[AGKZ20]{STOC:AGKZ20}
Ryan Amos, Marios Georgiou, Aggelos Kiayias, and Mark Zhandry.
\newblock One-shot signatures and applications to hybrid quantum/classical
  authentication.
\newblock In Konstantin Makarychev, Yury Makarychev, Madhur Tulsiani, Gautam
  Kamath, and Julia Chuzhoy, editors, {\em 52nd ACM STOC}, pages 255--268.
  {ACM} Press, June 2020.

\bibitem[AJ15]{C:AnaJai15}
Prabhanjan Ananth and Abhishek Jain.
\newblock Indistinguishability obfuscation from compact functional encryption.
\newblock In Rosario Gennaro and Matthew J.~B. Robshaw, editors, {\em
  CRYPTO~2015, Part~I}, volume 9215 of {\em {LNCS}}, pages 308--326. Springer,
  Heidelberg, August 2015.

\bibitem[AJL{\etalchar{+}}19]{C:AJLMS19}
Prabhanjan Ananth, Aayush Jain, Huijia Lin, Christian Matt, and Amit Sahai.
\newblock Indistinguishability obfuscation without multilinear maps: New
  paradigms via low degree weak pseudorandomness and security amplification.
\newblock In Alexandra Boldyreva and Daniele Micciancio, editors, {\em
  CRYPTO~2019, Part~III}, volume 11694 of {\em {LNCS}}, pages 284--332.
  Springer, Heidelberg, August 2019.

\bibitem[AJS15]{EPRINT:AnaJaiSah15}
Prabhanjan Ananth, Abhishek Jain, and Amit Sahai.
\newblock Indistinguishability obfuscation from functional encryption for
  simple functions.
\newblock Cryptology ePrint Archive, Report 2015/730, 2015.
\newblock \url{https://eprint.iacr.org/2015/730}.

\bibitem[AL21]{EC:AnaLaP21}
Prabhanjan Ananth and Rolando~L. {La Placa}.
\newblock Secure software leasing.
\newblock In Anne Canteaut and Fran\c{c}ois-Xavier Standaert, editors, {\em
  EUROCRYPT~2021, Part~II}, volume 12697 of {\em {LNCS}}, pages 501--530.
  Springer, Heidelberg, October 2021.

\bibitem[ALL{\etalchar{+}}21]{C:ALLZZ21}
Scott Aaronson, Jiahui Liu, Qipeng Liu, Mark Zhandry, and Ruizhe Zhang.
\newblock New approaches for quantum copy-protection.
\newblock In Tal Malkin and Chris Peikert, editors, {\em CRYPTO~2021, Part~I},
  volume 12825 of {\em {LNCS}}, pages 526--555, Virtual Event, August 2021.
  Springer, Heidelberg.

\bibitem[ALS16]{C:AgrLibSte16}
Shweta Agrawal, Beno{\^i}t Libert, and Damien Stehl{\'e}.
\newblock Fully secure functional encryption for inner products, from standard
  assumptions.
\newblock In Matthew Robshaw and Jonathan Katz, editors, {\em CRYPTO~2016,
  Part~III}, volume 9816 of {\em {LNCS}}, pages 333--362. Springer, Heidelberg,
  August 2016.

\bibitem[AMVY21]{C:AMVY21}
Shweta Agrawal, Monosij Maitra, Narasimha~Sai Vempati, and Shota Yamada.
\newblock Functional encryption for turing machines with dynamic bounded
  collusion from {LWE}.
\newblock In Tal Malkin and Chris Peikert, editors, {\em CRYPTO~2021, Part~IV},
  volume 12828 of {\em {LNCS}}, pages 239--269, Virtual Event, August 2021.
  Springer, Heidelberg.

\bibitem[AP20]{EC:AgrPel20}
Shweta Agrawal and Alice {Pellet-Mary}.
\newblock Indistinguishability obfuscation without maps: Attacks and fixes for
  noisy linear {FE}.
\newblock In Anne Canteaut and Yuval Ishai, editors, {\em EUROCRYPT~2020,
  Part~I}, volume 12105 of {\em {LNCS}}, pages 110--140. Springer, Heidelberg,
  May 2020.

\bibitem[AR17]{TCC:AgrRos17}
Shweta Agrawal and Alon Rosen.
\newblock Functional encryption for bounded collusions, revisited.
\newblock In Yael Kalai and Leonid Reyzin, editors, {\em TCC~2017, Part~I},
  volume 10677 of {\em {LNCS}}, pages 173--205. Springer, Heidelberg, November
  2017.

\bibitem[AV19]{TCC:AnaVai19}
Prabhanjan Ananth and Vinod Vaikuntanathan.
\newblock Optimal bounded-collusion secure functional encryption.
\newblock In Dennis Hofheinz and Alon Rosen, editors, {\em TCC~2019, Part~I},
  volume 11891 of {\em {LNCS}}, pages 174--198. Springer, Heidelberg, December
  2019.

\bibitem[BGI{\etalchar{+}}12]{JACM:BGIRSVY12}
Boaz Barak, Oded Goldreich, Russell Impagliazzo, Steven Rudich, Amit Sahai,
  Salil~P. Vadhan, and Ke~Yang.
\newblock On the (im)possibility of obfuscating programs.
\newblock {\em Journal of the {ACM}}, 59(2):6:1--6:48, 2012.

\bibitem[BGI14]{PKC:BoyGolIva14}
Elette Boyle, Shafi Goldwasser, and Ioana Ivan.
\newblock Functional signatures and pseudorandom functions.
\newblock In Hugo Krawczyk, editor, {\em PKC~2014}, volume 8383 of {\em
  {LNCS}}, pages 501--519. Springer, Heidelberg, March 2014.

\bibitem[BGMZ18]{TCC:BGMZ18}
James Bartusek, Jiaxin Guan, Fermi Ma, and Mark Zhandry.
\newblock Return of {GGH15}: Provable security against zeroizing attacks.
\newblock In Amos Beimel and Stefan Dziembowski, editors, {\em TCC~2018,
  Part~II}, volume 11240 of {\em {LNCS}}, pages 544--574. Springer, Heidelberg,
  November 2018.

\bibitem[BI20]{TCC:BroIsl20}
Anne Broadbent and Rabib Islam.
\newblock Quantum encryption with certified deletion.
\newblock In Rafael Pass and Krzysztof Pietrzak, editors, {\em TCC~2020,
  Part~III}, volume 12552 of {\em {LNCS}}, pages 92--122. Springer, Heidelberg,
  November 2020.

\bibitem[BJL{\etalchar{+}}21]{TCC:BJLPS21}
Anne Broadbent, Stacey Jeffery, S{\'e}bastien Lord, Supartha Podder, and Aarthi
  Sundaram.
\newblock Secure software leasing without assumptions.
\newblock In Kobbi Nissim and Brent Waters, editors, {\em TCC~2021, Part~I},
  volume 13042 of {\em {LNCS}}, pages 90--120. Springer, Heidelberg, November
  2021.

\bibitem[BNPW20]{JC:BNPW20}
Nir Bitansky, Ryo Nishimaki, Alain Passel{\`e}gue, and Daniel Wichs.
\newblock From cryptomania to obfustopia through secret-key functional
  encryption.
\newblock {\em Journal of Cryptology}, 33(2):357--405, April 2020.

\bibitem[BS17]{EPRINT:BenSat17}
Shalev {Ben-David} and Or~Sattath.
\newblock Quantum tokens for digital signatures.
\newblock Cryptology ePrint Archive, Report 2017/094, 2017.
\newblock \url{https://eprint.iacr.org/2017/094}.

\bibitem[BS18]{JC:BraSeg18}
Zvika Brakerski and Gil Segev.
\newblock Function-private functional encryption in the private-key setting.
\newblock {\em Journal of Cryptology}, 31(1):202--225, January 2018.

\bibitem[BSW11]{TCC:BonSahWat11}
Dan Boneh, Amit Sahai, and Brent Waters.
\newblock Functional encryption: Definitions and challenges.
\newblock In Yuval Ishai, editor, {\em TCC~2011}, volume 6597 of {\em {LNCS}},
  pages 253--273. Springer, Heidelberg, March 2011.

\bibitem[BV11]{C:BraVai11}
Zvika Brakerski and Vinod Vaikuntanathan.
\newblock Fully homomorphic encryption from ring-{LWE} and security for key
  dependent messages.
\newblock In Phillip Rogaway, editor, {\em CRYPTO~2011}, volume 6841 of {\em
  {LNCS}}, pages 505--524. Springer, Heidelberg, August 2011.

\bibitem[BV18]{JACM:BitVai18}
Nir Bitansky and Vinod Vaikuntanathan.
\newblock Indistinguishability obfuscation from functional encryption.
\newblock {\em Journal of the {ACM}}, 65(6):39:1--39:37, 2018.

\bibitem[BW13]{AC:BonWat13}
Dan Boneh and Brent Waters.
\newblock Constrained pseudorandom functions and their applications.
\newblock In Kazue Sako and Palash Sarkar, editors, {\em ASIACRYPT~2013,
  Part~II}, volume 8270 of {\em {LNCS}}, pages 280--300. Springer, Heidelberg,
  December 2013.

\bibitem[CGO21]{EC:CiaGoyOst21}
Michele Ciampi, Vipul Goyal, and Rafail Ostrovsky.
\newblock Threshold garbled circuits and ad hoc secure computation.
\newblock In Anne Canteaut and Fran\c{c}ois-Xavier Standaert, editors, {\em
  EUROCRYPT~2021, Part~III}, volume 12698 of {\em {LNCS}}, pages 64--93.
  Springer, Heidelberg, October 2021.

\bibitem[CHVW19]{TCC:CHVW19}
Yilei Chen, Minki Hhan, Vinod Vaikuntanathan, and Hoeteck Wee.
\newblock Matrix {PRFs}: Constructions, attacks, and applications to
  obfuscation.
\newblock In Dennis Hofheinz and Alon Rosen, editors, {\em TCC~2019, Part~I},
  volume 11891 of {\em {LNCS}}, pages 55--80. Springer, Heidelberg, December
  2019.

\bibitem[CLLZ21]{C:CLLZ21}
Andrea Coladangelo, Jiahui Liu, Qipeng Liu, and Mark Zhandry.
\newblock Hidden cosets and applications to unclonable cryptography.
\newblock In Tal Malkin and Chris Peikert, editors, {\em CRYPTO~2021, Part~I},
  volume 12825 of {\em {LNCS}}, pages 556--584, Virtual Event, August 2021.
  Springer, Heidelberg.

\bibitem[CMP20]{EPRINT:ColMajPor20}
Andrea Coladangelo, Christian Majenz, and Alexander Poremba.
\newblock Quantum copy-protection of compute-and-compare programs in the
  quantum random oracle model.
\newblock Cryptology ePrint Archive, Report 2020/1194, 2020.
\newblock \url{https://eprint.iacr.org/2020/1194}.

\bibitem[CS19]{USENIX:ConSch19}
R.~Joseph Connor and Max Schuchard.
\newblock Blind bernoulli trials: {A} noninteractive protocol for hidden-weight
  coin flips.
\newblock In Nadia Heninger and Patrick Traynor, editors, {\em USENIX Security
  2019}, pages 1483--1500. {USENIX} Association, August 2019.

\bibitem[CV21]{ARXIV:CulVid21}
Eric Culf and Thomas Vidick.
\newblock A monogamy-of-entanglement game for subspace coset states.
\newblock {\em arXiv (CoRR)}, abs/2107.13324, 2021.

\bibitem[DIJ{\etalchar{+}}13]{C:DIJOPP13}
Angelo {De Caro}, Vincenzo Iovino, Abhishek Jain, Adam O'Neill, Omer Paneth,
  and Giuseppe Persiano.
\newblock On the achievability of simulation-based security for functional
  encryption.
\newblock In Ran Canetti and Juan~A. Garay, editors, {\em CRYPTO~2013,
  Part~II}, volume 8043 of {\em {LNCS}}, pages 519--535. Springer, Heidelberg,
  August 2013.

\bibitem[DQV{\etalchar{+}}21]{TCC:DQVWW21}
Lalita Devadas, Willy Quach, Vinod Vaikuntanathan, Hoeteck Wee, and Daniel
  Wichs.
\newblock Succinct {LWE} sampling, random polynomials, and obfuscation.
\newblock In Kobbi Nissim and Brent Waters, editors, {\em TCC~2021, Part~II},
  volume 13043 of {\em {LNCS}}, pages 256--287. Springer, Heidelberg, November
  2021.

\bibitem[GGH{\etalchar{+}}16]{SIAMCOMP:GGHRSW16}
Sanjam Garg, Craig Gentry, Shai Halevi, Mariana Raykova, Amit Sahai, and Brent
  Waters.
\newblock Candidate indistinguishability obfuscation and functional encryption
  for all circuits.
\newblock {\em {SIAM} Journal on Computing}, 45(3):882--929, 2016.

\bibitem[GGHW17]{ALGMC:GGHW17}
Sanjam Garg, Craig Gentry, Shai Halevi, and Daniel Wichs.
\newblock On the implausibility of differing-inputs obfuscation and extractable
  witness encryption with auxiliary input.
\newblock {\em Algorithmica}, 79(4):1353--1373, 2017.

\bibitem[GGLW21]{EPRINT:GGLW21}
Rachit Garg, Rishab Goyal, George Lu, and Brent Waters.
\newblock Dynamic collusion bounded functional encryption from identity-based
  encryption.
\newblock Cryptology ePrint Archive, Report 2021/847, 2021.
\newblock \url{https://eprint.iacr.org/2021/847}.

\bibitem[GGM86]{JACM:GolGolMic86}
Oded Goldreich, Shafi Goldwasser, and Silvio Micali.
\newblock How to construct random functions.
\newblock {\em Journal of the {ACM}}, 33(4):792--807, 1986.

\bibitem[GKP{\etalchar{+}}13]{C:GKPVZ13}
Shafi Goldwasser, Yael~Tauman Kalai, Raluca~A. Popa, Vinod Vaikuntanathan, and
  Nickolai Zeldovich.
\newblock How to run turing machines on encrypted data.
\newblock In Ran Canetti and Juan~A. Garay, editors, {\em CRYPTO~2013,
  Part~II}, volume 8043 of {\em {LNCS}}, pages 536--553. Springer, Heidelberg,
  August 2013.

\bibitem[GKW17]{FOCS:GoyKopWat17}
Rishab Goyal, Venkata Koppula, and Brent Waters.
\newblock Lockable obfuscation.
\newblock In Chris Umans, editor, {\em 58th FOCS}, pages 612--621. {IEEE}
  Computer Society Press, October 2017.

\bibitem[GVW12]{C:GorVaiWee12}
Sergey Gorbunov, Vinod Vaikuntanathan, and Hoeteck Wee.
\newblock Functional encryption with bounded collusions via multi-party
  computation.
\newblock In Reihaneh Safavi-Naini and Ran Canetti, editors, {\em CRYPTO~2012},
  volume 7417 of {\em {LNCS}}, pages 162--179. Springer, Heidelberg, August
  2012.

\bibitem[GZ20]{EPRINT:GeoZha20}
Marios Georgiou and Mark Zhandry.
\newblock Unclonable decryption keys.
\newblock Cryptology ePrint Archive, Report 2020/877, 2020.
\newblock \url{https://eprint.iacr.org/2020/877}.

\bibitem[HMNY21]{AC:HMNY21}
Taiga Hiroka, Tomoyuki Morimae, Ryo Nishimaki, and Takashi Yamakawa.
\newblock Quantum encryption with certified deletion, revisited: Public key,
  attribute-based, and classical communication.
\newblock In Mehdi Tibouchi and Huaxiong Wang, editors, {\em ASIACRYPT~2021,
  Part~I}, volume 13090 of {\em {LNCS}}, pages 606--636. Springer, Heidelberg,
  December 2021.

\bibitem[JKMS20]{C:JKMS20}
Aayush Jain, Alexis Korb, Nathan Manohar, and Amit Sahai.
\newblock Amplifying the security of functional encryption, unconditionally.
\newblock In Daniele Micciancio and Thomas Ristenpart, editors, {\em
  CRYPTO~2020, Part~I}, volume 12170 of {\em {LNCS}}, pages 717--746. Springer,
  Heidelberg, August 2020.

\bibitem[JLS21]{STOC:JaiLinSah21}
Aayush Jain, Huijia Lin, and Amit Sahai.
\newblock Indistinguishability obfuscation from well-founded assumptions.
\newblock In Samir Khuller and Virginia~Vassilevska Williams, editors, {\em
  {STOC} 2021}, pages 60--73. {ACM}, 2021.

\bibitem[JLS22]{EC:JaiLinSah22}
Aayush Jain, Huijia Lin, and Amit Sahai.
\newblock Indistinguishability obfuscation from {LPN} over {$\mathbb{F}_p$},
  {DLIN}, and {PRGs} in {${NC}^0$}.
\newblock In Orr Dunkelman and Stefan Dziembowski, editors, {\em
  EUROCRYPT~2022, Part~I}, volume 13275 of {\em {LNCS}}, pages 670--699.
  Springer, Heidelberg, May~/~June 2022.

\bibitem[KLM{\etalchar{+}}18]{SCN:KLMMRW18}
Sam Kim, Kevin Lewi, Avradip Mandal, Hart Montgomery, Arnab Roy, and David~J.
  Wu.
\newblock Function-hiding inner product encryption is practical.
\newblock In Dario Catalano and Roberto {De Prisco}, editors, {\em SCN 18},
  volume 11035 of {\em {LNCS}}, pages 544--562. Springer, Heidelberg, September
  2018.

\bibitem[KNT21]{JC:KitNisTan21}
Fuyuki Kitagawa, Ryo Nishimaki, and Keisuke Tanaka.
\newblock Simple and generic constructions of succinct functional encryption.
\newblock {\em Journal of Cryptology}, 34(3):25, July 2021.

\bibitem[KNT22]{JC:KitNisTan22}
Fuyuki Kitagawa, Ryo Nishimaki, and Keisuke Tanaka.
\newblock Obfustopia built on secret-key functional encryption.
\newblock {\em J. Cryptol.}, 35(3):19, 2022.

\bibitem[KNY21]{TCC:KitNisYam21}
Fuyuki Kitagawa, Ryo Nishimaki, and Takashi Yamakawa.
\newblock Secure software leasing from standard assumptions.
\newblock In Kobbi Nissim and Brent Waters, editors, {\em TCC~2021, Part~I},
  volume 13042 of {\em {LNCS}}, pages 31--61. Springer, Heidelberg, November
  2021.

\bibitem[KPTZ13]{CCS:KPTZ13}
Aggelos Kiayias, Stavros Papadopoulos, Nikos Triandopoulos, and Thomas
  Zacharias.
\newblock Delegatable pseudorandom functions and applications.
\newblock In Ahmad-Reza Sadeghi, Virgil~D. Gligor, and Moti Yung, editors, {\em
  ACM CCS 2013}, pages 669--684. {ACM} Press, November 2013.

\bibitem[MW16]{EC:MukWic16}
Pratyay Mukherjee and Daniel Wichs.
\newblock Two round multiparty computation via multi-key {FHE}.
\newblock In Marc Fischlin and Jean-S{\'{e}}bastien Coron, editors, {\em
  EUROCRYPT~2016, Part~II}, volume 9666 of {\em {LNCS}}, pages 735--763.
  Springer, Heidelberg, May 2016.

\bibitem[NNL01]{C:NaoNaoLot01}
Dalit Naor, Moni Naor, and Jeffery Lotspiech.
\newblock Revocation and tracing schemes for stateless receivers.
\newblock In Joe Kilian, editor, {\em CRYPTO~2001}, volume 2139 of {\em
  {LNCS}}, pages 41--62. Springer, Heidelberg, August 2001.

\bibitem[NP01]{FC:NaoPin00}
Moni Naor and Benny Pinkas.
\newblock Efficient trace and revoke schemes.
\newblock In Yair Frankel, editor, {\em FC 2000}, volume 1962 of {\em {LNCS}},
  pages 1--20. Springer, Heidelberg, February 2001.

\bibitem[NWZ16]{EC:NisWicZha16}
Ryo Nishimaki, Daniel Wichs, and Mark Zhandry.
\newblock Anonymous traitor tracing: How to embed arbitrary information in a
  key.
\newblock In Marc Fischlin and Jean-S{\'{e}}bastien Coron, editors, {\em
  EUROCRYPT~2016, Part~II}, volume 9666 of {\em {LNCS}}, pages 388--419.
  Springer, Heidelberg, May 2016.

\bibitem[RPB{\etalchar{+}}19]{NIPS:RPBDG19}
Th{\'{e}}o Ryffel, David Pointcheval, Francis~R. Bach, Edouard Dufour{-}Sans,
  and Romain Gay.
\newblock Partially encrypted deep learning using functional encryption.
\newblock In Hanna~M. Wallach, Hugo Larochelle, Alina Beygelzimer, Florence
  d'Alch{\'{e}}{-}Buc, Emily~B. Fox, and Roman Garnett, editors, {\em NeurIPS
  2019}, pages 4519--4530, 2019.

\bibitem[SS10]{CCS:SahSey10}
Amit Sahai and Hakan Seyalioglu.
\newblock Worry-free encryption: functional encryption with public keys.
\newblock In Ehab {Al-Shaer}, Angelos~D. Keromytis, and Vitaly Shmatikov,
  editors, {\em ACM CCS 2010}, pages 463--472. {ACM} Press, October 2010.

\bibitem[SSW12]{C:SahSeyWat12}
Amit Sahai, Hakan Seyalioglu, and Brent Waters.
\newblock Dynamic credentials and ciphertext delegation for attribute-based
  encryption.
\newblock In Reihaneh Safavi-Naini and Ran Canetti, editors, {\em CRYPTO~2012},
  volume 7417 of {\em {LNCS}}, pages 199--217. Springer, Heidelberg, August
  2012.

\bibitem[SW05]{EC:SahWat05}
Amit Sahai and Brent~R. Waters.
\newblock Fuzzy identity-based encryption.
\newblock In Ronald Cramer, editor, {\em EUROCRYPT~2005}, volume 3494 of {\em
  {LNCS}}, pages 457--473. Springer, Heidelberg, May 2005.

\bibitem[SW21]{SIAMCOMP:SahWat21}
Amit Sahai and Brent Waters.
\newblock How to use indistinguishability obfuscation: Deniable encryption, and
  more.
\newblock {\em {SIAM} J. Comput.}, 50(3):857--908, 2021.

\bibitem[WZ17]{FOCS:WicZir17}
Daniel Wichs and Giorgos Zirdelis.
\newblock Obfuscating compute-and-compare programs under {LWE}.
\newblock In Chris Umans, editor, {\em 58th FOCS}, pages 600--611. {IEEE}
  Computer Society Press, October 2017.

\bibitem[Zha20]{TCC:Zhandry20}
Mark Zhandry.
\newblock {Schr{\"o}dinger}'s pirate: How to trace a quantum decoder.
\newblock In Rafael Pass and Krzysztof Pietrzak, editors, {\em TCC~2020,
  Part~III}, volume 12552 of {\em {LNCS}}, pages 61--91. Springer, Heidelberg,
  November 2020.

\end{thebibliography}
	\fi
	\fi

	
\ifnum\cameraready=0
	\ifnum\llncs=0
	\appendix

\section{CPA-Style Anti-Piracy Security Notions}\label{appsec:anti-piracy-cpa}

We introduce the CPA-style anti-piracy security for SDE.
\begin{definition}[Anti-Piracy Security, CPA-Style~\cite{C:CLLZ21}]\label{def:anti-piracy_CPA_PKE}
We consider the CPA-style anti-piracy game $\expc{\SDE,\qA}{anti}{piracy}{cpa}(\secp)$ between the challenger and an adversary $\qA$ below.
\begin{enumerate}
\item The challenger generates $(\pk,\sk)\gets \Setup(1^\secp)$.
\item The challenger generates $\qsk \gets \QKeyGen(\sk)$ and sends $(\pk,\qsk)$ to $\qA$.
\item $\qA$ outputs $(m_0,m_1)$ and two (possibly entangled) quantum decryptors $\pirateD_1 = (\qstateq[\qreg{R}_1],\mat{U}_1)$ and $\pirateD_2 = (\qstateq[\qreg{R}_2],\mat{U}_2)$, where $m_0 \ne m_1$, $\abs{m_0}=\abs{m_1}$, $\qstateq$ is a quantum state over registers $\qreg{R}_1$ and $\qreg{R}_2$, and $\mat{U}_1$ and $\mat{U}_2$ are general quantum circuits.
\item The challenger chooses $\coin_1,\coin_2 \chosen \bit$, and generates $\ct_1 \gets \Enc(\pk,m_{\coin_1})$ and $\ct_2 \gets \Enc(\pk,m_{\coin_2})$. The challenger runs $m_i^\prime \gets \pirateD_i(\ct_i)$ for $i \in \setbk{1,2}$.
If $m_i^\prime = m_{\coin_i}$ for $i\in \setbk{1,2}$, the challenger outputs $1$, otherwise outputs $0$.
\end{enumerate}

We say that $\SDE$ is $\gamma$-anti-piracy secure if for any QPT adversary $\qA$, it satisfies that
\[
\Pr[\expc{\SDE,\qA}{anti}{piracy}{cpa}(\sep)=1]  \le  \frac{1}{2} + \gamma(\secp) + \negl(\sep).
\]
\end{definition}

\begin{theorem}[\cite{C:CLLZ21}]\label{thm:strong-cpa-SDE}
If an SDE scheme satisfies strong $\gamma$-anti-piracy security, it also satisfies $\gamma$-anti-piracry security (in the CPA-style).
\end{theorem}

We can define a similar security notion for SDFE.

\begin{definition}[Anti-Piracy Security for FE, CPA-Style]\label{def:anti-piracy_CPA_FE}
We consider the CPA-style anti-piracy game for FE $\expc{\SDFE,\qA}{anti}{piracy}{cpa}(\secp)$ between the challenger and an adversary $\qA$ below.
\begin{enumerate}
\item The challenger generates $(\pk,\msk)\gets \Setup(1^\secp)$ and sends $\pk$ to $\qA$.
\item $\qA$ sends key queries $f_i$ to the challenger and receives $\qsk_{f_i} \lrun \QKeyGen(\sk_{f_i})$ where $\sk_{f_i}\lrun \KG(\msk,f_i)$.
\item At some point $\qA$ sends a challenge query $f^*$ to the challenger and receives $\qsk_{f^*} \lrun \QKeyGen(\msk,f^\ast)$.
\item Again, $\qA$ sends $f_i$ to the challenger and receives $\qsk_{f_i} \lrun \QKeyGen(\sk_{f_i})$ where $\sk_{f_i}\lrun \KG(\msk,f_i)$.
\item $\qA$ outputs $(x_0,x_1)$ and two (possibly entangled) quantum decryptors $\pirateD_1 = (\qstateq[\qreg{R}_1],\mat{U}_1)$ and $\pirateD_2 = (\qstateq[\qreg{R}_2],\mat{U}_2)$, where $\forall i\ f_i(x_0)=f_i(x_1)$, $f^*(x_0)\ne f^*(x_1)$, $\qstateq$ is a quantum state over registers $\qreg{R}_1$ and $\qreg{R}_2$, and $\mat{U}_1$ and $\mat{U}_2$ are general quantum circuits.
\item The challenger chooses $\coin_1,\coin_2 \chosen \bit$, and generates $\ct_1 \gets \Enc(\pk,x_{\coin_1})$ and $\ct_2 \gets \Enc(\pk,x_{\coin_2})$. The challenger runs $\coin_k^\prime \gets \pirateD_k(\ct_k)$ for $k \in \setbk{1,2}$.
If $\coin_k^\prime = \coin_k$ for $k\in \setbk{1,2}$, the challenger outputs $1$, otherwise outputs $0$.
\end{enumerate}

We say that $\SDFE$ is $\gamma$-anti-piracy secure if for any QPT adversary $\qA$, it satisfies that
\[
\Pr[\expc{\SDFE,\qA}{anti}{piracy}{cpa}(\sep)=1]  \le \frac{1}{2} + \gamma(\secp) + \negl(\sep).
\]
If $\qA$ can send only the challenge query $f^\ast$ during the experiment, we say that $\SDFE$ is challenge-only $\gamma$-anti-piracy secure.
\end{definition}

We can also define the CPA-style anti-piracy for the predictor-based quantum decryptors as~\cref{def:anti-piracy_CPA_PKE}. That is, $\pirateD_i$ outputs $f^\ast(x_{\coin_i})$ for $i \in \setbk{1,2}$ instead of $\coin_i$. We omit the variation.

\begin{theorem}\label{thm:cpa-style_from_anti-piracy}
If an SDFE scheme satisfies strong $\gamma$-anti-piracy security, it also satisfies $\gamma$-anti-piracy security in the CPA-style.
\end{theorem}
The proof of this theorem is almost the same as that of~\cref{thm:strong-cpa-SDE}, so we omit it to avoid replication.

\newcommand{\Sel}{\algo{Sel}}
\newcommand{\onect}{\keys{1ct}}
\newcommand{\funcg}{G}
\newcommand{\funcT}{T}
\newcommand{\adaSKL}{\algo{adaSKL}}
\newcommand{\indSKL}{\algo{indSKL}}
\newcommand{\delflag}{\keys{flag}_{\keys{del}}}

\section{Adaptive and Simulation-based Security for Secure Key Leasing}\label{sec:stronger_FESKL_security}
This section presents how to upgrade FE with secure key leasing security.

\subsection{Adaptive Security for Secure Key Leasing}\label{sec:adatpive_FESKL_security}
Ananth et al.~\cite{C:ABSV15} present transformations from selectively secure FE into adaptively secure FE. Almost the same transformation works for FE with secure key leasing.
The adaptive lessor security for SKFE with SKL is almost the same as~\cref{def:sel_lessor_SKFESKL} except that $\qA$ declares $(x_0^\ast,x_1^\ast)$ when $\qA$ sends $\cert$ to the challenger and we separate the key oracle into two types. The formal definition is as follows.

\begin{definition}[Adaptive Lessor Security]\label{def:ada_lessor_SKFESKL}
We say that $\SKFESKL$ is an adaptively lessor secure SKFE-SKL scheme for $\Xs,\Ys$, and $\Fs$, if it satisfies the following requirement, formalized from the experiment $\expb{\qA,\SKFESKL}{ada}{lessor}(1^\secp,\coin)$ between an adversary $\qA$ and a challenger:
        \begin{enumerate}
            \item At the beginning, $\qA$ sends $1^{\numkey+\numkey^\ast}$ to the challenger. The challenger runs $\msk\gets\Setup(1^\secp,1^{\numkey+\numkey^\ast})$.
            Throughout the experiment, $\qA$ can access the following oracles.
            \begin{description}
            \item[$\Oracle{\Enc}(x)$:] Given $x$, it returns $\Enc(\msk,x)$.
            \item[$\Oracle{\qKG}(f,1^{\numct})$:] Given $(f,1^{\numct})$, it generates $(\qfsk,\vk)\la\qKG(\msk,f,1^{\numct})$, and sends $\qfsk$ to $\qA$. $\qA$ can access this oracle at most $\numkey$ times.
            \end{description}
            \item $\qA$ can access the following oracles before $\qA$ declares the challenge plaintexts.
            \begin{description}
            \item[$\Oracle{\qKG}^\ast(f^\ast,1^{\numct^\ast})$:] Given $(f^\ast,1^{\numct^\ast})$, it generates $(\qfsk^\ast,\vk^\ast)\la\qKG(\msk,f^\ast,1^{\numct^\ast})$, adds $(f^\ast,1^{\numct^\ast},\vk^\ast,\bot)$ to $\List{\qKG}$, and sends $\qfsk^\ast$ to $\qA$. $\qA$ can access this oracle at most $\numkey^\ast$ times.
            \item[$\Oracle{\Vrfy}(f^\ast,\cert^\ast)$:] Given $(f^\ast,\cert^\ast)$, it finds an entry $(f^\ast,1^{\numct^\ast},\vk^\ast,M)$ from $\List{\qKG}$. (If there is no such entry, it returns $\bot$.) If $\top=\Vrfy(\vk^\ast,\cert^\ast)$ and the number of queries to $\Oracle{\Enc}$ at this point is less than $\numct^\ast$, it returns $\top$ and updates the entry into $(f^\ast,1^{\numct^\ast},\vk^\ast,\top)$. Otherwise, it returns $\bot$.
            \end{description}
            \item When $\qA$ sends $(x_0^\ast,x_1^\ast)$ to receive the challenge ciphertext, if there exists $(f^\ast,1^{\numct^\ast},\vk^\ast,\bot)$ in $\List{\qKG}$, or $\Oracle{\qKG}$ received $(f,1^{\numct})$ such that $f(x_0^\ast) \ne f(x_1^\ast)$, the challenger outputs $0$ as the final output of this experiment. Otherwise, the challenger generates $\ct^*\la\Enc(\msk,x_\coin^*)$ and sends $\ct^*$ to $\qA$. Hereafter, $\Oracle{\qKG}$ rejects a query $(f,1^n)$ such that $f(x_0^\ast)\ne f(x_1^\ast)$ and $\qA$ cannot access $\Oracle{\qKG}^\ast$ and $\Oracle{\Vrfy}$.
            \item $\qA$ outputs a guess $\coin^\prime$ for $\coin$. The challenger outputs $\coin'$ as the final output of the experiment.

        \end{enumerate}
        For any QPT $\qA$, it holds that
\ifnum\llncs=0        
\begin{align}
\advb{\SKFESKL,\qA}{ada}{lessor}(\secp) \seteq \abs{\Pr[\expb{\SKFESKL,\qA}{ada}{lessor} (1^\secp,0) \out 1] - \Pr[\expb{\SKFESKL,\qA}{ada}{lessor} (1^\secp,1) \out 1] }\leq \negl(\secp).
\end{align}
\else
\begin{align}
\advb{\SKFESKL,\qA}{ada}{lessor}(\secp) 
&\seteq \abs{\Pr[\expb{\SKFESKL,\qA}{ada}{lessor} (1^\secp,0) \out 1] - \Pr[\expb{\SKFESKL,\qA}{ada}{lessor} (1^\secp,1) \out 1] }\\
&\leq \negl(\secp).
\end{align}
\fi
\end{definition}

\begin{remark}
The definition above considers two key oracles $\Oracle{\qKG}$ and $\Oracle{\qKG}^\ast$ unlike~\cref{def:sel_lessor_SKFESKL} to distinguish a functional decryption key which is deleted from one which is not deleted. The oracle $\Oracle{\qKG}^\ast$ accepts a query $(f^\ast,1^{\numct^\ast})$ such that $f^\ast(x_0^\ast) \ne f^\ast(x_1^\ast)$. We need to distinguish these two types of queries since we need to know whether $f(x_0^\ast) = f(x_1^\ast)$ or not when an adversary sends a query in the security proof of the construction below. It is not clear whether we can achieve a stronger adaptive security where adversaries access a single key generation oracle as~\cref{def:sel_lessor_SKFESKL}. We leave it as an open question.
\end{remark}

We construct an adaptively secure SKFE scheme with secure key leasing $\SKFESKL=(\Setup,\qKG,\Enc,\allowbreak \qDec,\qcert,\Vrfy)$ using the following tools:
\begin{itemize}
\item A selectively secure SKFE scheme with secure key leasing $\Sel.\SKFE=\Sel.(\Setup,\allowbreak \qKG, \Enc, \qDec,\allowbreak \qcert,\Vrfy)$.
\item An adaptively single-ciphertext secure SKFE scheme $\OneCT{}.\SKFE=\OneCT{}.(\Setup,\allowbreak\KG,\Enc,\Dec)$.
\item A pseudorandom-secure SKE scheme $\SKE=\SKE.(\E,\D)$.
\item A PRF $\prf$.
\end{itemize}

Let $\ell_1$ and $\ell_2$ be the length of $\SKE$ ciphertexts and the length of random tags, respectively.
The definition of adaptive lessor security is a natural adaptive variant of~\cref{def:sel_lessor_SKFESKL}.
The description of $\SKFESKL$ is as follows. 
\begin{description}

 \item[$\Setup(1^\secp)$:] $ $
 \begin{itemize}
 \item Generate $\sel.\msk \gets \Sel.\Setup(1^\secp)$.
 \item Output $\msk \seteq \sel.\msk$.
 \end{itemize}
 \item[$\qKG(\msk,f,1^\numct)$:] $ $
 \begin{itemize}
 \item Parse $\sel.\msk \gets \msk$.
 \item Generate $\ct_\ske \gets \zo{\ell_1}$ and choose $\gtag \chosen \zo{\ell_2}$.
 \item Compute $(\sel.\qsk_\funcg,\sel.\vk) \gets \Sel.\qKG(\sel.\msk,\funcg[f,\ct_\ske,\tau],1^\numct)$, where the circuit $\funcg$ is described in~\cref{fig:absv-conversion}.
 \item Output $\qsk_f \seteq \sel.\qsk_{\funcg}$ and $\vk \seteq \sel.\vk$.
 \end{itemize}
 \item[$\Enc(\msk,x)$:] $ $
 \begin{itemize}
 \item Parse $\sel.\msk \gets \msk$.
 \item Generate $\onect.\msk \gets \OneCT{}.\Setup(1^\secp)$ and choose $\prfkey \chosen \zo{\secp}$.
 \item Generate $\onect.\ct \gets \OneCT{}.\Enc(\onect.\msk,x)$.
 \item Generate $\sel.\ct \gets \Sel.\Enc(\sel.\msk,(\onect.\msk,\prfkey,0^\secp,0))$.
 \item Output $\ct_x \seteq (\onect.\ct,\sel.\ct)$.
 \end{itemize}
\item[$\qDec(\qsk_f,\ct_x)$:] $ $
\begin{itemize}
\item Parse $\sel.\qsk_\funcg \gets \qsk_f$ and $(\onect.\ct,\sel.\ct) \gets \ct_x$.
\item Compute $\onect.\sk_f \gets \Sel.\qDec(\sel.\qsk_{\funcg},\sel.\ct)$.
\item Output $y \gets \OneCT{}.\Dec(\onect.\sk_f, \onect.\ct)$.
\end{itemize}
\item[$\qcert(\qsk_f)$:] $ $
\begin{itemize}
\item Parse $\sel.\qsk_f \gets \qsk_f$.
\item Output $\cert \seteq \sel.\cert\gets\Sel.\qcert(\sel.\qsk_{f})$.
\end{itemize}
\item[$\Vrfy(\vk,\cert)$:] $ $
\begin{itemize}
    \item Parse $\sel.\vk = \vk$ and $\sel.\cert = \cert$.
\item Output $\top/\bot\gets\Sel.\Vrfy(\sel.\vk,\sel.\cert)$.
\end{itemize}
\end{description}

\protocol
{Key Generation Circuit $\funcg [f,\ct_\ske,\gtag](\onect.\msk,\prfkey,K_\ske,\beta)$}
{Description of $\funcg[f,\ct_\ske,\gtag]$.
}
{fig:absv-conversion}
{
\ifnum\llncs=1
\scriptsize
\else
\fi
\begin{description}
\setlength{\parskip}{0.3mm} 
\setlength{\itemsep}{0.3mm} 
\item[Hardwired:] function $f$, SKE ciphertext $\ct_\ske$, tag $\gtag$.
\item[Input:] master secret key $\onect.\msk$, PRF key $\prfkey$, SKE secret key $K_\ske$, bit $\beta$.
\end{description}
\begin{enumerate}
\setlength{\parskip}{0.3mm} 
\setlength{\itemsep}{0.3mm} 
\item If $\beta =0$, output $\onect.\sk_f \gets \OneCT{}.\KG(\onect.\msk,f;\prf_\prfkey(\gtag))$.
\item Else, output $\SKE.\D(K_\ske,\ct_\ske)$.
\end{enumerate}
}

We emphasize that $\OneCT{}.\SKFE$ is a classical standard SKFE scheme, and $\OneCT{}.\KG$ is a classical algorithm.
This does not spoil the lessor security of $\Sel.\SKFE$ because $\onect.\sk_f$ generated by the circuit $\funcg$ depends on $\onect.\msk$ and $\prfkey$ from the encryption algorithm and $\gtag$ from the key generation algorithm.
At the decryption phase, we can obtain $\onect.\sk_f$ under $\onect.\msk$ from $\Sel.\qKG(\sel.\msk,\funcg,1^\numct)$ and
\[(\OneCT{}.\Enc(\onect.\msk,x),\Sel.\Enc(\sel.\msk,(\onect.\msk,\prfkey,0^\secp,0))).
\]
We can keep a copy of $\onect.\sk_f$. However, it does not help us to decrypt another ciphertext
\[(\OneCT{}.\Enc(\onect.\msk^\prime,x^\prime),\Sel.\Enc(\sel.\msk,(\onect.\msk^\prime,\prfkey^\prime,0^\secp, 0))).
\]
Note that $\onect.\msk$ is freshly generated at the encryption phase.
That is, a functional decryption key $\onect.\sk_f$ is specific to only one pair of $\qsk_f$ and $\ct_x$

Thus, after we securely delete $\qsk_f$ by using the deletion algorithm $\Sel.\qcert$ of $\Sel.\SKFE$, we can no longer decrypt a fresh ciphertext generated after the deletion. Since the deletion power is preserved, we can upgrade the selective security of $\Sel.\SKFE$ to adaptive security by leveraging the adaptive security of $\OneCT{}.\SKFE$ as the transformation by Ananth et al.~\cite{C:ABSV15}.

\begin{theorem}
If $\Sel.\SKFE$ is selectively lessor secure, $\OneCT{}.\SKFE$ is adaptively single-ciphertext secure and collusion-resistant, $\SKE$ is pseudorandom-secure, and $\prf$ is a pseudorandom function, $\SKFESKL$ above is adaptively lessor secure.
\end{theorem}
Since the proof is almost the same as that by Ananth et al.~\cite[Theorem 4 in the eprint ver.]{C:ABSV15}, we omit it to avoid the replication.
Note that the reduction needs to embed $\onect.\sk_f \gets \OneCT{}.\KG(\onect.\msk,f;\prf_\prfkey(\gtag))$ into $\ct_{\ske}$ for a key query such that $f(x_0^\ast)=f(x_1^\ast)$, but $\ct_{\ske}$ should be a junk ciphertext for a key query such that $f(x_0^\ast)\ne f(x_1^\ast)$.
The reduction knows which simulation-way is correct since~\cref{def:ada_lessor_SKFESKL} distinguishes two types of queries.

\subsection{Simulation-based Security for Secure Key Leasing}\label{sec:simulation_FESKL_security}
Although we present indistinguishability-based security definitions for SKFE with secure key leasing in~\cref{sec:def_SKFE_SKL}, we can also consider simulation-based security definitions for SKFE with secure key leasing.
We present the definition of adaptive lessor simulation-security.

\begin{definition}[Adaptive Lessor Simulation-security]
Let $\SKFESKL$ be an SKFE scheme with secure key leasing.
For a stateful QPT adversary $\qA=(\qA_1,\qA_2,\qA_3)$ and a stateful QPT Simulator $\qS=(\qSEnc,\qSKG)$, we consider the two experiments described in~\cref{fig:FE_SKL_adsim_sec_experiments},
where
\begin{itemize}
\item $\qA$ is allowed to make total at most $\numkey$ queries to $\qKG(\msk,\cdot,\cdot)$ (resp.  $\qKG(\msk,\cdot,\cdot)$ and $\Oracle{\qSKG}(\cdot,\cdot)$) in $\sfreal{ad}{sim}_{\SKFESKL,\qA}(\sep)$ (resp. $\sfsim{ad}{sim}_{\SKFESKL,\qA,\qS}(\sep)$),
\item $\qA$ is allowed to make total at most $\numkey^\ast$ queries to $\Oracle{\qKG}^\ast(\cdot,\cdot)$, which is the same as $\qKG(\msk,\cdot,\cdot)$ except that it also adds $(f^\ast,1^{\numct^\ast},\vk^\ast,\bot)$ to $\List{\qKG}$ when it is invoked,
\item $\qA$ is allowed to make queries to $\Oracle{\Vrfy}(\cdot,\cdot)$, which is given $(f^\ast,\cert^\ast)$, it finds an entry $(f^\ast,1^{\numct^\ast},\vk^\ast,M)$ from $\List{\qKG}$, and does the following: If there is no such entry, it returns $\bot$. If $\top=\Vrfy(\vk^\ast,\cert^\ast)$ and the number of queries to $\Oracle{\Enc}$ at this point is less than $\numct^\ast$, it returns $\top$ and updates the entry into $(f^\ast,1^{\numct^\ast},\vk^\ast,\top)$. Otherwise, it returns $\bot$,
\item $\Oracle{\qSKG}(f,1^\numct)=\qSKG(\state,f,1^{\numct},f(x^*))$.
\end{itemize}

\medskip
\begin{figure}
\centering
 \begin{tabular}{r@{\ }p{0.44\textwidth}r@{\ }p{0.44\textwidth}}\toprule
 \Heading{$\sfreal{ad}{sim}_{\SKFESKL,\qA}(\sep)$}{$\sfsim{ad}{sim}_{\SKFESKL,\qA,\qS}(\sep)$}\midrule
   \setcounter{expitem}{0}
    \expitem{$1^{\numkey+\numkey^\ast}\gets\qA_1(1^\secp)$}{$1^{\numkey+\numkey^\ast}\gets\qA_1(1^\secp)$}
   \expitem{$\msk \gets \Setup(1^\secp,1^{\numkey+\numkey^\ast})$}{$\msk \gets \Setup(1^\secp,1^{\numkey+\numkey^\ast})$}
                   \expitem{$x^\ast \gets \qA_2^{\Enc(\msk,\cdot),\qKG(\msk,\cdot,\cdot),\Oracle{\qKG}^\ast(\cdot,\cdot),\Oracle{\Vrfy}(\cdot,\cdot)}$}{$x^\ast \gets \qA_2^{\Enc(\msk,\cdot),\qKG(\msk,\cdot,\cdot),\Oracle{\qKG}^\ast(\cdot,\cdot),\Oracle{\Vrfy}(\cdot,\cdot)}$}
                       \expitem{Output $0$ if there exists $(f^\ast,1^{\numct^\ast},\vk^\ast,\bot)$ in $\List{\qKG}$. Otherwise, go to the next step.}{Output $0$ if there exists $(f^\ast,1^{\numct^\ast},\vk^\ast,\bot)$ in $\List{\qKG}$. Otherwise, go to the next step.}
                       \expitem{}{Let $\calQ$ be the query/answer list for $\Enc(\msk,\cdot)$}
                       \expitem{}{Let $(f_i,1^{\numcti{i}})_{i\in[\numkey]}$ be the queries for $\KG(\msk,\cdot,\cdot)$}
                       \expitem{}{$y_i:=f_i(x^*)$ for every $i\in[\numkey]$}
                       \expitem{$\ct^*\gets\Enc(\msk,x^*)$}{$(\ct^*,\state)\gets\qSEnc(\msk,\calQ,(f_i,1^{\numcti{i}},y_i)_{i\in[\numkey]})$}
                       \expitem{Output $b\gets \qA_3(\ct^*)^{\Enc(\msk,\cdot),\qKG(\msk,\cdot,\cdot)}$}{Output $b\gets \qA_3(\ct^*)^{\Enc(\msk,\cdot),\Oracle{\qSKG}(\cdot,\cdot)}$}
            \bottomrule
 \end{tabular}
   \caption{Security experiments for adaptively lessor simulation-secure SKFE-SKL}\label{fig:FE_SKL_adsim_sec_experiments}
\end{figure}

\medskip

We say that $\SKFESKL$ is adaptively lessor simulation-secure if there exists a QPT simulator $\qS$ such that for any QPT adversary $\qA$, it satisfies that
\[
\abs{\Pr[\sfreal{ad}{sim}_{\SKFESKL,\qA}(\sep)=1] - \Pr[\sfsim{ad}{sim}_{\SKFESKL,\qA,\qS}(\sep)=1]} \le \negl(\sep).
\]

\end{definition}
This is a natural simulation-based variant of adaptive lessor (indistinguishability-)security.
The simulator $\qS=(\qSEnc,\qSKG)$ does not have $f^{\ast}_j(x^\ast)$, and the simulated ciphertext does not reveal information about $f^\ast_j(x^\ast)$.
\begin{remark}
We consider two key oracles $\qKG$ and $\Oracle{\qKG}^\ast$ as~\cref{def:ada_lessor_SKFESKL} to distinguish two types of queries.
We can achieve this simulation-based security since we use adaptively lessor secure SKFE-SKL in the sense of~\cref{def:ada_lessor_SKFESKL} as a building block below. It is also an open question to achieve a stronger adaptive simulation-based security where adversaries access a single key generation oracle as~\cref{def:sel_lessor_SKFESKL}.
\end{remark}

De Caro, Iovino, Jani, O'Neil, Paneth, and Persiano~\cite{C:DIJOPP13} present a transformation from indistinguishability-based secure FE to simulation-based secure FE. The transformation works if the number of key queries before a challenge ciphertext is given is a-priori bounded.
We can apply almost the same transformation to SKFE with secure leasing since we modify a plaintext in the ciphertext and a function embedded in a functional decryption key.

We construct an adaptively lessor simulation-secure SKFE scheme with secure key leasing $\adaSKL=\adaSKL.(\Setup,\qKG,\Enc,\qDec,\qcert,\Vrfy)$ using the following tools:
\begin{itemize}
\item An adaptively lessor (indistinguishability-)secure SKFE scheme with secure key leasing $\indSKL=\indSKL.(\Setup, \qKG, \Enc,\allowbreak \qDec,\qcert,\Vrfy)$.
\item A pseudorandom-secure SKE scheme $\SKE=\SKE.(\E,\D)$.
\end{itemize}
Let $\ell_1$, $\ell_2$, and $m$ be the length of $\SKE$ ciphertexts, the length of random tags, and the output length, respectively.
The description of $\adaSKL$ is as follows. 
\begin{description}

 \item[$\adaSKL.\Setup(1^\secp)$:] $ $
 \begin{itemize}
 \item Generate $\ind.\msk \gets \indSKL.\Setup(1^\secp)$.
 \item Output $\msk \seteq \ind.\msk$.
 \end{itemize}
 \item[$\adaSKL.\qKG(\msk,f,1^\numct)$:] $ $
 \begin{itemize}
 \item Parse $\ind.\msk\gets \msk$.
 \item Generate $\ct_\ske \gets \zo{\ell_1}$ and choose $\gtag \chosen \zo{\ell_2}$.
 \item Compute $(\ind.\qsk_\funcT,\ind.\vk) \gets \indSKL.\qKG(\sel.\msk,\funcT[f,\ct_\ske,\tau],1^\numct)$, where the circuit $\funcT$ is described in~\cref{fig:sim-conversion}.
 \item Output $\qsk_f \seteq \ind.\qsk_{\funcT}$ and $\vk \seteq \ind.\vk$.
 \end{itemize}
 \item[$\adaSKL.\Enc(\msk,x)$:] $ $
 \begin{itemize}
 \item Parse $(\ind.\msk, K_\ske) \gets \msk$.
 \item Set $x^\prime \seteq (x,0^{\secp},(0^{\ell_2},0^{m}),\ldots,(0^{\ell_2},0^m),0)$, where $m$ is the output length of functions.
 \item Generate $\ind.\ct \gets \indSKL.\Enc(\ind.\msk,x^\prime)$.
 \item Output $\ct_x \seteq \ind.\ct$.
 \end{itemize}
\item[$\adaSKL.\qDec(\qsk_f,\ct_x)$:] $ $
\begin{itemize}
\item Parse $\ind.\qsk_\funcT \gets \qsk_f$ and $\ind.\ct \gets \ct_x$.
\item Output $y \gets \indSKL.\qDec(\ind.\qsk_\funcT, \ind.\ct)$.
\end{itemize}
\item[$\adaSKL.\qcert(\qsk_f)$:] $ $
\begin{itemize}
\item Parse $\ind.\qsk_f \gets \qsk_f$.
\item Output $\cert \seteq \ind.\cert\gets\indSKL.\qcert(\ind.\qsk_{f})$.
\end{itemize}
\item[$\adaSKL.\Vrfy(\vk,\cert)$:] $ $
\begin{itemize}
    \item Parse $\ind.\vk = \vk$ and $\ind.\cert = \cert$.
\item Output $\top/\bot\gets\indSKL.\Vrfy(\ind.\vk,\ind.\cert)$.
\end{itemize}
\end{description}

\protocol
{Trapdoor Circuit $\funcT [f,\ct_\ske,\gtag](x,K_\ske,(\gtag_1,y_1),\ldots,(\gtag_q,y_q),\beta)$}
{Description of $\funcT[f,\ct_\ske,\gtag]$.
}
{fig:sim-conversion}
{
\ifnum\llncs=1
\scriptsize
\else
\fi
\begin{description}
\setlength{\parskip}{0.3mm} 
\setlength{\itemsep}{0.3mm} 
\item[Hardwired:] function $f$, SKE ciphertext $\ct_\ske$, tag $\gtag$.
\item[Input:] plaintext $x$, SKE secret key $K_\ske$, $q$ pairs of tag and output $(\gtag_1,y_1),\ldots,(\gtag_q,y_q)$, bit $\beta$.
\end{description}
\begin{enumerate}
\setlength{\parskip}{0.3mm} 
\setlength{\itemsep}{0.3mm} 
\item If $\beta =1$, do the following
\begin{itemize}
 \item If there exists $i$ such that $\gtag=\gtag_i$, output $y_i$.
 \item Else output $y^\prime \gets \SKE.\D(K_\ske,\ct_\ske)$.
 \end{itemize} 
\item Else, output $f(x)$.
\end{enumerate}
}

As we see above, we expand the plaintext space and use the trapdoor circuit $\funcT$ as De Caro et al.~\cite{C:DIJOPP13}.
We can construct a QPT simulator as follows.
\begin{enumerate}
\item Generates $\ind.\msk \gets \indSKL.\Setup(1^\secp)$ and $K_\ske \chosen \zo{\secp}$.
\item Receives $\qsk_{f_i} = (\ind.\qsk_{f_i},\ct_{\ske,i},\gtag_i)$ for $i\in [q_{\mathsf{pre}}]$ and outputs for pre-challenge key queries, that is, $f_1(x^\ast),\ldots,f_{q_{\mathsf{pre}}}(x^\ast)$.
\item Generates a challenge ciphertext $\indSKL.\Enc(\ind.\msk,(0,K_\ske,(\gtag_1,f_1(x^\ast)),\ldots, \allowbreak(\gtag_{q_{\mathsf{pre}}},f_{q_{\mathsf{pre}}}(x^\ast)),1))$.
\item Simulates functional decryption keys for post-challenge key queries as follows. When an output value $f^\prime_i(x^\ast)$ for a key query $f^\prime_i$ is given, generates $(\ind.\qsk_\funcT,\ind.\vk) \gets \Sel.\qKG(\sel.\msk,\funcT[f^\prime_i,\ct^\prime_{\ske,i},\gtag^\prime_i],1^\numct)$ where $\ct^\prime_{\ske,i} \gets \SKE.\E(K_\ske,f^\prime_i(x^\ast))$ and $\gtag^\prime_i \chosen \zo{\ell_2}$.
\end{enumerate}

Note that the simulator should not have $f^\ast_j(x^\ast)$ for any $j\in [\numkey^\ast]$ since $\setbk{\qsk_{f^\ast_j}}_{j\in [\numkey^\ast]}$ are deleted and the challenge ciphertext should not reveal information about $\setbk{f^\ast_j(x^\ast)}_{j\in[\numkey^\ast]}$, where $\setbk{f^\ast_j}_{j\in[\numkey^\ast]}$ are the challenge functions, and the simulation-security is satisfied.

Since the flag $\beta=1$ in the simulated challenge ciphertext, the trapdoor circuit $\funcT$ works at the trapdoor branch.
For functional decryption keys before the challenge ciphertext, $\funcT$ works at the first item in the $\beta=1$ branch since $\gtag_i$ is a tag embedded in the functional decryption keys for pre-challenge key queries. So, the output of $\funcT$ is an embedded $f_i(x^\ast)$. This is consistent.
For functional decryption keys after the challenge ciphertext, $\funcT$ works at the second item in the $\beta=1$ branch since $\gtag^\prime_i$ is uniformly random and must be different from $\gtag_i$ except negligible probability. In addition, $y^\prime = f^\prime_i(x^\ast)$ since $\ct^\prime_\ske \gets \SKE.\E(K_\ske,f^\prime_i(x^\ast))$. This is consistent. Thus, the simulated ciphertext does not include information about $\setbk{f^\ast_j(x^\ast)}_{j\in[\numkey^\ast]}$.

In reducing to the adaptive lessor (indistinguishability)-security, the simulator can simulate the key generation and encryption oracles in the simulation-based game by using the oracles in the indistinguishability-based game.

\begin{theorem}
If $\indSKL$ is adaptively lessor (indistinguishability-)secure and $\SKE$ is pseuodrandom-secure, $\adaSKL$ is adaptively lessor simulation-secure.
\end{theorem}
Since the proof is almost the same as that by De Caro et al.~\cite[Theorem 14 in the eprint ver.]{C:DIJOPP13}, we omit it to avoid the replication.

\else
	\newpage
	 	\appendix
	 	\setcounter{page}{1}
 	{
	\noindent
 	\begin{center}
	{\Large SUPPLEMENTAL MATERIALS}
	\end{center}
 	}
	\setcounter{tocdepth}{2}
	 	\ifnum\noaux=1
 	\else
{\color{red}{We attached the full version of this paper as a separated file (auxiliary supplemental material) for readability. It is available from the program committee members.}}
\fi
  

\section{Preliminaries}\label{sec:prelim}

\paragraph{Notations and conventions.}
In this paper, standard math or sans serif font stands for classical algorithms (e.g., $C$ or $\algo{Gen}$) and classical variables (e.g., $x$ or $\keys{pk}$).
Calligraphic font stands for quantum algorithms (e.g., $\qalgo{Gen}$) and calligraphic font and/or the bracket notation for (mixed) quantum states (e.g., $\qstateq$ or $\ket{\psi}$).
For strings $x$ and $y$, $x \concat y$ denotes the concatenation of $x$ and $y$.
Let $\mv{0}$ denote a string consisting of an appropriate number of $0$.
Let $[\ell]$ denote the set of integers $\{1, \cdots, \ell \}$, $\secp$ denote a security parameter, and $y \seteq z$ denote that $y$ is set, defined, or substituted by $z$.

In this paper, for a finite set $X$ and a distribution $D$, $x \chosen X$ denotes selecting an element from $X$ uniformly at random, $x \chosen D$ denotes sampling an element $x$ according to $D$. Let $y \gets \algo{A}(x)$ and $y \gets \qalgo{A}(\qstate{x})$ denote assigning to $y$ the output of a probabilistic or deterministic algorithm $\algo{A}$ and a quantum algorithm $\qalgo{A}$ on an input $x$ and $\qstate{x}$, respectively. When we explicitly show that $\algo{A}$ uses randomness $r$, we write $y \gets \algo{A}(x;r)$.
PPT and QPT algorithms stand for probabilistic polynomial-time algorithms and polynomial-time quantum algorithms, respectively.
Let $\negl$ denote a negligible function.

	\setcounter{tocdepth}{1}
	\tableofcontents

	\fi
	\else
\fi
	
\end{document}